\documentclass[11pt]{article}
\usepackage [dvips]{graphics}
\usepackage[centertags]{amsmath}
\usepackage{amsfonts}
\usepackage{amssymb}
\usepackage{amsthm}
\usepackage{newlfont}
\usepackage{stmaryrd}
\usepackage{color,epsfig}
\usepackage{amsfonts,graphicx,psfrag}
\usepackage{graphicx}
\usepackage{float}

\usepackage{txfonts}
\usepackage{mathrsfs}

\theoremstyle{plain}
\newtheorem{thm}{Theorem}[section]
\newtheorem{cor}[thm]{Corollary}
\newtheorem{lem}[thm]{Lemma}
\newtheorem{prop}[thm]{Proposition}
\newtheorem{exam}[thm]{Example}
\newtheorem{rem}[thm]{Remark}
\newtheorem{defi}[thm]{Definition}

\def\sqr#1#2{{\vcenter{\vbox{\hrule height.#2pt
              \hbox{\vrule width.#2pt height#1pt \kern#1pt \vrule
width.#2pt}
              \hrule height.#2pt}}}}
\def\dbF{{\mathbb{F}}}

\def\be{\begin{equation}}
\def\ee{\end{equation}}

\def\bb{{\beta}}
\def\ga{{\gamma}}

\def\om{{\omega}}

\def\ep{{\epsilon}}

\def\dm{{\mathrm \diamond}}

\def\Sp{{\mathrm {Sp}}}
\def\lb{\label}
\def\bb{{\beta}}
\def\ga{{\gamma}}

\def\d{{d\over dt}}
\def\sd{{d^2\over dt^2}}

\def\bS{\bar{S}}

\def\cM{{\cal M}}

\def\no{\noindent}

\def\bs{\bigskip}

\def\lan{\mathop{\langle}}
\def\ran{\mathop{\rangle}}

\def\dim{\hbox{\rm dim$\,$}}

\def\({\Big (}
\def\){\Big )}
\def\[{\Big[}
\def\]{\Big]}

\def\ol#1{\overline{#1}}

\def\be{\begin{equation}}
\def\bel{\begin{equation}\label}
\def\ee{\end{equation}}
\def\bea{\begin{eqnarray}}
\def\eea{\end{eqnarray}}
\def\bt{\begin{theorem}}
\def\et{\end{theorem}}
\def\bc{\begin{corollary}}
\def\ec{\end{corollary}}
\def\bl{\begin{lemma}}
\def\el{\end{lemma}}
\def\bp{\begin{proposition}}
\def\ep{\end{proposition}}
\def\br{\begin{remark}}
\def\er{\end{remark}}
\def\ba{\begin{array}}
\def\ea{\end{array}}
\def\bd{\begin{definition}}
\def\ed{\end{definition}}

\makeatletter
   
   \@addtoreset{equation}{section}
\makeatother

\begin{document}

\title{\bf Trace formula for linear Hamiltonian systems with its applications to elliptic Lagrangian solutions}
\author{Xijun Hu\thanks{Partially supported
by NSFC(No.11131004), E-mail:xjhu@sdu.edu.cn }
\quad Yuwei Ou \thanks{Partially supported
by NSFC(No.11131004), E-mail:yuweiou@163.com }
 \quad Penghui Wang\thanks{ Partially supported by NSFC(No.11101240),
 E-mail: phwang@sdu.edu.cn }    \\ \\
 Department of Mathematics, Shandong University\\
Jinan, Shandong 250100, The People's Republic of China\\
}
\date{}
\maketitle
\begin{abstract}

In the present paper,  we build  up  trace formulas for both the linear Hamiltonian systems and Sturm-Liouville systems. The formula connects the monodromy
matrix of a symmetric periodic orbit with the infinite sum of eigenvalues of the Hessian of the action functional.  A natural application is to study the non-degeneracy of linear Hamiltonian systems. Precisely, by the trace formula,  we can give an estimation for the upper bound such that the non-degeneracy preserves. Moreover, we could  estimate  the relative Morse index by the trace formula. Consequently,  a series of  new stability criteria for the symmetric periodic orbits is given. As a concrete application, the trace formula is used to study the linear stability of elliptic Lagrangian solutions of the classical planar three-body problem. It is well known that the linear stability of elliptic Lagrangian solutions depends on the mass parameter $\bb=27(m_1m_2+m_2m_3+m_3m_1)/(m_1+m_2+m_3)^2\in [0,9]$ and the
eccentricity $e\in [0,1)$. Based on the trace formula,  we estimate the stable region and hyperbolic region of the elliptic Lagranian solutions.

\end{abstract}

\bs

\no{\bf AMS Subject Classification:} 37J25, 47E05, 70F07, 37B30, 37J45

\bs

\no{\bf Key Words}. trace formula, Hamiltonian systems, Sturm-Liouville systems, planar three-body problem, linear stability

\section{Introduction}

In the study of symmetric periodic solutions or quasi-periodic
solutions in $n$-body problem, it is natural to consider the
$S$-periodic solution in Hamiltonian system
\begin{eqnarray} \dot{z}(t)&=& J H'(t,z(t)), \label{1.1}\\
  z(0)&=& S z(T), \label{1.2}
  \end{eqnarray}
where $J=\left(\begin{array}{cc}0&-I_n\\
                                I_n&0\end{array}\right)$, $S$ is a symplectic orthogonal matrix on $\mathbb R^{2n}$, and $H(t,x)\in C^2(\mathbb R^{2n+1};\mathbb R)$. Please refer \cite{Ch},  \cite{CM}, \cite{FT} and references therein for the background of $S$-periodic orbits in $n$-body problems.
For the solution  $z$  of (\ref{1.1}-\ref{1.2}),  let $\gamma\equiv
\gamma_z(t)$ be the corresponding  fundamental solution, that is
$ \dot{\gamma}(t)= J B(t)\gamma(t),
\gamma(0)= I_{2n},  $ where
$B(t)=B(t)^T=H''(t,z(t))$.
 $\gamma(T)$ is called the monodromy matrix.

The linear stability of $S$-periodic solution $z(t)$ depends on the location of
 eigenvalues of $S\gamma(T)$(see e.g. \cite{HW}). But
 due to the non-commutativity, in general, the fundamental solution  could not be obtained directly.
 In the present paper, we obtain a kind of  trace formula for linear Hamiltonian system. Using the trace formula, we can estimate the relative Morse index, and hence, based on the theory of Maslov-type index  \cite{Lon4}, we give some new stability criteria for Hamiltonian system. Finally, the trace formula will be used to study the stable region and hyperbolic region of Lagrangian solutions in planar three body problem.

For $k\in \mathbb{N}$, $\mathbb \dbF=\mathbb R$ or $\mathbb C$,  let $\mathcal{M}(k,\mathbb{F})$  be the set of $k\times k$ matrices  on $\mathbb{F}^k$. We denote by
 $\Sp(2k)=\{\mathcal{P}\in \mathcal{M}(2k,\mathbb{R}) , \mathcal{P}^TJ\mathcal{P}=J \}$  the symplectic group,  $\mathcal{S}(k)$  the set of $k\times k$ real symmetric
matrices and  $\mathcal{B}(k)=C([0,T];\mathcal{S}(k))$, the space of continuous paths on $[0,T]$ of matrices in $\mathcal{S}(k)$. For
$B(t), D(t)\in \cal{B}(2n)$, consider the eigenvalue problem of the following linear
Hamiltonian systems,
\begin{eqnarray} \dot{z}(t)&=& J(B(t)+\lambda D(t))z(t), \label{0.1.5}\\
 z(0)&=& S z(T). \label{0.1.6}\end{eqnarray} Denote by $A=-J\frac{d}{dt},$ which is densely
defined in the Hilbert space $E=L^2([0,T];\mathbb C^{2n})$ with the
domain
\begin{eqnarray*}
D_S=\left\{\left.z(t)\in W^{1,2}([0,T];\mathbb C^{2n})\,\right|\,z(0)=Sz(T)\right\}.
\end{eqnarray*}
 $B$ is a bounded linear operator defined by $(Bz)(t)=B(t)z(t)$ on $E$. Then $A$ is a self-adjoint operator with compact resolvent; moreover for $\lambda\in\rho(A)$, the resolvent set of $A$, $(\lambda-A)^{-1}$ is Hilbert-Schimidt.

As above, let $\gamma_\lambda(t)$ be the fundamental solution of (\ref{0.1.5}).  To state  the trace formula for Hamiltonian system, we  need some notations. 
 Write $M=S\gamma_0(T)$ and $\hat{D}(t)=\gamma_{0}^T(t)D(t) \gamma_{0}(t)$. For $k\in\mathbb{N}$, let
\begin{eqnarray*}
M_k=\int_0^TJ\hat{D}(t_1)\int_0^{t_1}J\hat{D}(t_2)\cdots\int_0^{t_{k-1}}J\hat{D}(t_k)dt_k\cdots
dt_2 dt_1, \label{adc0.4.13}
\end{eqnarray*} and
\begin{eqnarray*}\label{adc4.14}
G_k=M_k M\left(M-e^{\nu T} I_{2n}\right)^{-1}.
\end{eqnarray*}

\begin{thm}\label{thm1.1}
For $\nu\in \mathbb C$ such that $A-B-\nu J$ is invertible, we have
for any positive integer $m$,
\begin{eqnarray}\label{0.0.0}
Tr \[\left(D\left(A-B-\nu J\right)^{-1}\right)^{m}\]=m\sum_{k=1}^m
\frac{(-1)^{k}}{k}\[\sum\limits_{j_1+\cdots+j_k=m}Tr(G_{j_1}\cdots
G_{j_k})\].
\end{eqnarray}
\end{thm}
There are two reasons why we consider the parameter $\nu$ in Theorem \ref{thm1.1}. Firstly, for a given $B\in\mathcal{B}(2n)$, we can not expect that $A-B$ is invertible. However, for every $\nu\in\mathbb C$ except countable points, $A-B-\nu J$ is invertible. Secondly, the operator $D(A-B-\nu J)^{-1}$ comes from the following boundary value problem naturally
 \begin{eqnarray} \dot{z}(t)&=& J(B(t)+\lambda D(t))z(t) \label{0.1.1}\\
  z(0)&=& \omega S z(T), \label{0.1.2}
  \end{eqnarray}
where $\lambda\in\mathbb R\setminus\{0\}$ and $\omega=e^{\nu T}$.
In fact, if we set $A_\omega=-J{d\over dt}$ with the domain $D_{ S}=\Big\{ z(0)=\omega Sz(T)\,\big|\, z(t)\in W^{1,2}([0,T];\mathbb C ^{2n})\Big\},$ then $e^{-\nu t}A_{\omega}e^{\nu t}=A-\nu J$. Thus $z\in \ker(A_\omega-B-\lambda D)$ if and only if $e^{-\nu T}z(t)\in \ker(A-\nu J-B-\lambda D)$, which is equivalent to that $1\over\lambda$ is an eigenvalue $D(A-\nu J-B)^{-1}$ provided that $A-\nu J-B$ is invertible.
\begin{rem}
\begin{itemize}
 \item[(1).] For $m=1$, $D(A-\nu J-B)^{-1}$ is not a trace class operator but a Hilbert-Schmidt operator. And hence $Tr(D(A-\nu J-B)^{-1})$ is not the usual trace but a kind of conditional trace\cite{HW}.
\item[(2).] For $m\geq 2$, $\left(D\left(A-\nu J-B\right)^{-1}\right)^m$ are trace class operators. By the preceding argument, $\lambda$ is a nonzero eigenvalue of system (\ref{0.1.1})-(\ref{0.1.2}) if and only if $1\over \lambda$ is an eigenvalue of $D(A-\nu J-B)^{-1}$. And hence, if we let $\{\lambda_i\}$ be the set of nonzero eigenvalues of the system (\ref{0.1.1})-(\ref{0.1.2}),
    \begin{eqnarray}
    Tr\left[\left(D(A-B-\nu J)^{-1}\right)^m\right]=\sum_{j=1}^\infty \frac{1}{\lambda_j^m}=m\sum_{k=1}^m \frac{(-1)^{k}}{k}\Big[\sum\limits_{j_1+\cdots+j_k=m}Tr(G_{j_1}\cdots
G_{j_k})\Big], \,\ m\geq2. \label{0.1.4}
    \end{eqnarray}
\end{itemize}
\end{rem}

For large $m$, the right hand side of (\ref{0.0.0}) is a little
complicated. However, for $m=1,2$, we can write it down more
precisely.
\begin{cor}\label{cor1.1} For $\nu\in\mathbb C$ such that $A-B-\nu J$ is invertible,
\begin{eqnarray}\label{ht1}
Tr\left[D(A-B-\nu J)^{-1}\right]=-Tr\[J\int_0^T\gamma_{0}^T(t)D(t)
\gamma_{0}(t)dt\cdot M(M-e^{\nu T} I_{2n})^{-1}\],
\end{eqnarray}
and
\begin{eqnarray}\label{ht2}
&&Tr\big([D(A-B-\nu J)^{-1}]^2\big)\nonumber\\&&=-2Tr \[J\int_0^T\gamma_{0}^T(t)D(t) \gamma_{0}(t)J\int_0^s\gamma_{0}^T(s)D(s) \gamma_{0}(s)ds dt \cdot M(M-e^{\nu T} I_{2n})^{-1} \]\nonumber\\
                                &&\ \ \ +Tr\[\Big(J\int_0^T\gamma_{0}^T(t)D(t) \gamma_{0}(t)dtM(M-e^{\nu T} I_{2n})^{-1}\Big)^2\].
\end{eqnarray}
Especially, in the case that $M=\pm I_{2n}$,
\begin{eqnarray} Tr\[\left(D(A-\nu J-B)^{-1}\right)^{2}\]=\frac{\pm e^{\nu T}}{(1\mp e^{\nu T})^2} Tr\[\(J\int_0^T\gamma_{0}^T(s)D(s) \gamma_{0}(s)ds\)^2\]. \label{eq5b.5.1.1a}
\end{eqnarray}
\end{cor}

In some concrete problem, such as the estimation of hyperbolic region of elliptic Lagrangian solution,  the trace formula for Lagrangian system is more convenient to be used. In order to introduce the trace formula for Lagrangian system,  it is natural to consider the following  eigenvalue problem of  Sturm-Liouville
system with $\bar{S}$-periodic boundary condition
 \begin{eqnarray}
-(P\dot{y}+Qy)^\cdot+Q^T\dot{y}+(R+\lambda R_1)y=0,\quad
y(0)=\bar{S}y(T),\quad \dot{y}(0)=\bar{S}\dot{y}(T),\label{e1} \end{eqnarray}
where $\bar{S}$ is an orthogonal matrix on  $\mathbb R^n$, $P, R, R_1\in\mathcal{ B}(n)$, $Q\in C([0,T];\mathcal{M}(n,\mathbb{R}))$. Instead of Legendre convexity condition, we assume for any $t\in [0,T]$, $P(t)$  is  invertible. Moreover we assume
 \begin{eqnarray}
 \bar{S}P(T)=P(0)\bar{S}\quad\text{and}\quad \bar{S}Q(T)=Q(0)\bar{S}. \label{c1}
\end{eqnarray}
Such a boundary value problem with condition (\ref{c1}) comes naturally from the study of symmetric periodic orbits  in $n$-body problem.

By the standard Legendre transformation,  the linear  system
(\ref{e1}) corresponds to the linear Hamiltonian system,
 \bea \dot{z}=JB_\lambda(t)z, \quad z(t)=\bar{S}_dz(T), \label{h1}\eea
with \bea \bar{S}_d=\left(\begin{array}{cc}\bar{S}& 0_n \\ 0_n  & \bar{S}
\end{array}\right),\quad \text{and}\quad B_\lambda(t)=\left(\begin{array}{cc}P^{-1}(t)& -P^{-1}Q(t) \\
-Q(t)^TP^{-1}(t)  & Q(t)^TP^{-1}(t)Q(t)-R(t)-\lambda R_1(t)
\end{array}\right).\label{b2} \eea
Obviously, $\bar{S}_d$ is a symplectic orthogonal  matrix on $\mathbb R^{2n}$, and the eigenvalue problem (\ref{h1}) is a special case of the eigenvalue problem (\ref{0.1.5}-\ref{0.1.6}).  Without confusion, for Lagrangian system, denote by $\gamma_\lambda(t)$  the
fundamental solution of (\ref{h1}).

Using the notations in Theorem \ref{thm1.1}, take  $D=\left(\begin{array}{cc} 0_n & 0_n \\
0_n& -R_1\end{array}\right)$. Temporarily, we assume the unperturbed systems is non-degenerate, that is,   $0$ is not the eigenvalue of (\ref{e1}), which is equivalent to that
 $1$ is not the eigenvalue of $M=\bar{S}_d\gamma_0(T)$.

\begin{thm}\label{thm1.1a}
Let $\{\lambda_j\}$ be the  eigenvalues for the boundary value problem (\ref{e1}), then
 \bea
\sum_{j}\frac{1}{\lambda_j^m}=m\sum_{k=1}^m
\frac{(-1)^{k}}{k}\[\sum\limits_{j_1+\cdots+j_k=m}Tr(G_{j_1}\cdots
G_{j_k})\],  \forall m\in\mathbb{N}, \label{lt.1a} \eea
 especially for $m=1$,
 \bea \sum_{j}\frac{1}{\lambda_j}=-
Tr\[J\int_0^T\gamma_0^{T}(t)D(t) \gamma_0(t)dt\cdot
M(M-I_{2n})^{-1}\].  \label{lt.2}
\eea
\end{thm}

It should be pointed out that from Proposition \ref{cor2.5}, for $m\geq 2$, the trace formula (\ref{lt.1a}) is a special case of the formula (\ref{0.1.4}). However, for $m=1$, the meanings of the formula (\ref{ht1}) and (\ref{lt.2}) are totally different. In fact, $Tr \Big(D(A-B-\nu J)^{-1}\Big)$ is a kind of conditional trace. Details could be found in Remark \ref{cor.lh}. The formula (\ref{lt.2}) is proved for Sturm-Liouville system, and we do not know for general Hamiltonian system whether it holds true or not. Fortunately, (\ref{lt.2}) is easy to be calculated.

During the study of the above trace formula, thanking for Chongchun
Zeng's suggestion, we can find the original work by Krein\cite{K1,K2} in 1950s. In fact,
Krein considered the following system
\begin{eqnarray} \dot{z}(t)&=& \lambda J D(t) z(t),  \label{0.1.7}\\
  z(0)&=& -z(T), \label{0.1.8}
  \end{eqnarray}
where $D\geq 0$ and $\int_0^T D(t)dt> 0$. The system
(\ref{0.1.7}-\ref{0.1.8}) is a special case of our system
(\ref{0.1.5}-\ref{0.1.6}). For the system (\ref{0.1.7}-\ref{0.1.8}),
Krein proved that
$
\lim\limits_{r\to\infty}
\sum\limits_{|\lambda_j|<r}\frac{1}{\lambda_j}=0,$
and
\begin{eqnarray}
\sum
\frac{1}{\lambda_j^2}&=&\frac{T^2}{2}Tr(A_{11} A_{22}-A_{12}^2),\label{0.1.9}
\end{eqnarray}
where $\lambda_j$ are the eigenvalues for the system
(\ref{0.1.7}-\ref{0.1.8}), and
$\left(\begin{array}{cc}A_{11}&A_{12}\\ A_{21} &
A_{22}\end{array}\right)=\frac{1}{T}\int_0^T D(t) dt$.
Moreover, under the condition  $D\geq 0$,   $\int_0^T D(t)dt> 0$, Krein gave an interesting stability criteria:
  \begin{eqnarray*}
\label{0.1.9b} \frac{T^2}{2}Tr(A_{11} A_{22}-A_{12}^2)<1.  \end{eqnarray*}
Obviously, by taking $\nu=0$ and $M=-I_{2n}$ in
the formula (\ref{eq5b.5.1.1a}), it is easy to see that  Theorem \ref{thm1.1} generalizes
Krein's formula (\ref{0.1.9}).
\begin{rem}
Krein considered the simplest Hamiltonian system with some special conditions such as $D\geq 0$ and $\int_0^T D(t) dt>0$. For the system coming from $n$-body problem, the conditions are not satisfied. Hence, Krein's trace formula can not be used to study the $n$-body problem. However, Krein's trace formula is a powerful tool to study the stability. It is surprised  that, to the best of our knowledge, there is no further study along this line.
\end{rem}

Next, we will introduce  some applications of the trace formula.  As one application, we will give some estimations on the non-degeneracy of the linear system. It is well-known that the system preserves the non-degeneracy under small perturbations. A natural question will be arisen: can we give an upper bound for the perturbation, such that, under the smaller perturbation, the systems preserve the non-degeneracy? By the trace formula, we can answer this question partly. Details could be found in Section 4.  As another application, the trace formula could be used to estimate the relative Morse index for Hamiltonian systems and Morse index for Lagrangian systems.
It is well-known that the relative Morse index (or Morse index) is equal to the
Maslov-type index for path of symplectic matrices  and the Maslov-type index is a successful tool in
judging the linear stability \cite{Lon4}, \cite{HS}. In Section \ref{sec 4}, by
using the trace formula, we can  give some new  stability criteria.

Before  giving the further application of the trace formula on $n$-body problem, we want to interpret the proof of the trace formula intuitively. For a matrix $F$, to calculate the trace $Tr F^m$ for $m>0$, the most effective method is to consider the determinant $\det(I+\alpha F)$, where $I$ is the identity matrix and $\alpha$ is a parameter. In the case of trace formula of differential equation, the idea does work too. From this viewpoint, Hill-type formula is the cornerstone to get the trace formula. The study of such a formula begins with the original work of Hill \cite{Hi} in 1877. In his study of the motion of lunar perigee,
Hill considered the following equation:
\begin{eqnarray}\label{eq0.1}
\ddot{x}(t)+\theta(t) x(t)=0,
\end{eqnarray}
where $\theta(t)=\sum\limits_{j\in\mathbb{Z}}\theta_je^{2j\sqrt{-1}t}$ with $\theta_0\neq0$ is a real $\pi$-periodic function. Let $\gamma(t)$ be the fundamental solution of the  associated  first order system of (\ref{eq0.1}), that is,
\begin{eqnarray*}
\dot{\gamma}(t)&=&\left(\begin{array}{cc} 0 & -\theta(t) \\ 1 & 0 \end{array}\right)\gamma(t),\\ \gamma(0)&=&I_2.
\end{eqnarray*}
 Suppose $\rho=e^{c\sqrt{-1}\pi}$, $\rho^{-1}=e^{-c\sqrt{-1}\pi}$ are  the eigenvalues of the monodromy matrix $\gamma(\pi)$. In order to compute $c$, Hill obtained the following  formula which
connects  the infinite determinant, corresponding to the differential operator, and the the characteristic polynomial:
\begin{eqnarray}\label{eq0.2}
\frac{\sin^2(\frac{\pi}{2}c)}{\sin^2(\frac{\pi}{2}\theta_0)}=\det\[\(-\sd-\theta_0\)^{-1}\(-\sd-\theta\)\],
\end{eqnarray}
where  the right hand side of (\ref{eq0.2}) is the Fredholm determinant.   We should point out that the right hand side of the original formula of Hill \cite{Hi} is a determinant of an infinite matrix.  In \cite{Hi}, Hill did not prove the convergence of the infinite determinant, and the convergence was proved by Poincar\'{e} \cite{Po}.
The Hill formula for a periodic solution of  Lagrangian system on manifold was given  by Bolotin\cite{B}.  In  \cite{BT},  Bolotin and Treschev studied the Hill-type formula for both continuous and discrete Lagrangian systems with Legendre convexity condition. For the periodic solution of ODE, the Hill-type formula was given by Denk \cite{De}.

 For $S$-periodic orbit of Hamiltonian system, the Hill-type formula was given by the first and the third authors \cite{HW}, for $B,D\in\mathcal{B}(2n)$
 \begin{eqnarray}  \det\[\(A-(B+\lambda D)-\nu J\)(A+P_0)^{-1}\]=C(S) e^{-n\nu T}\det( S\gamma_\lambda(T)- e^{\nu T} I_{2n}). \label{thf1.1.2}  \end{eqnarray}
where $C(S)>0$ is a constant depending  only on $S$, and $\gamma_\lambda(t)$ satisfies $\dot{\gamma}_\lambda(t)= J(B(t)+\lambda D(t))\gamma_\lambda(t),$ and $ \gamma_\lambda(0)=I_{2n}$.
The equality (\ref{thf1.1.2}) is our starting point to get the trace formula of Hamiltonian system. In fact, both sides of (\ref{thf1.1.2}) are analytic functions on $\lambda$. Then, by taking Taylor expansion and comparing the coefficients on both sides of (\ref{thf1.1.2}), we get the trace formula in Theorem \ref{thm1.1}. Based on this idea, in order to obtain the trace formula for Lagrangian system, in the present paper we will get the following Hill-type formula.

\begin{thm}\label{thm1.2} Let $\{\lambda_j\}$ be the  nonzero eigenvalues for the boundary value problem (\ref{e1}), then
\begin{eqnarray}
 \prod_j \(1-\frac{1}{\lambda_j}\)  =\det(\bar{S}_d\gamma_1(T)-I_{2n})\cdot \det(\bar{S}_d\gamma_0(T)-I_{2n})^{-1}, \label{hillag}
\end{eqnarray}
where $\gamma_\lambda$ is the fundamental solution of the system (\ref{h1}).
\end{thm}
\begin{rem} The Hill-type formula for periodic orbits of  Lagrangian system with the Legendre convex condition was given by
Bolotin \cite{B} in 1988, and Theorem \ref{thm1.2} can be considered as a generalization of Bolotin's work to indefinite Lagrangian systems.

\end{rem}

At the end of this paper, we will study the stability of Lagrangian orbits in planar three body problems.
In 1772, Lagrange \cite{Lag} discovered some celebrated periodic solutions, now named after him,
to the planar three-body problem, namely the three bodies form an equilateral triangle at any instant
of the motion and at the same time each body travels along a specific Keplerian elliptic orbit about the
center of masses of the system. All these orbits are homographic solutions. When $0\le e<1$, the Keplerian orbit is elliptic, following Meyer and Schmidt \cite{MS}, we call such
elliptic Lagrangian solutions {\it elliptic relative equilibria}. Specially when $e=0$, the Keplerian
elliptic motion becomes circular motion and then all the three bodies move around the center of masses
along circular orbits with the same frequency, which are called {\it relative equilibria} traditionally.
Moreover,  Meyer and Schmidt (cf. \cite{MS}) used heavily the central configuration nature
of the elliptic Lagrangian orbits and decomposed the fundamental solution of the elliptic Lagrangian
orbit into two parts symplectically, one of which is the same as that of the Keplerian solution and
the other is the essential part for the stability.

For the planar three-body problem with masses $m_1, m_2, m_3>0$, it turns out that the stability
of elliptic Lagrangian solutions depends on two parameters, namely the mass parameter $\beta\in [0,9]$
defined below and the eccentricity $e\in [0,1)$,
\be  \beta=\frac{27(m_1m_2+m_1m_3+m_2m_3)}{(m_1+m_2+m_3)^2}. \nonumber  \lb{1.4}\ee
 In the current paper, the fundamental
solution of the linearized Hamiltonian system of the essential part of the elliptic Lagrangian orbit
is denoted by $\gamma_{\beta,e}(t)$ for $t\in [0,2\pi]$, which is a path of $4\times 4$ symplectic
matrices starting from the identity. The Lagrangian orbits is called spectrally stable (or elliptic) if all the eigenvalues of $\gamma_{\beta,e}(2\pi)$ belong to the unite circle $\mathbb{U}$, is called linear stable if moreover $\gamma_{\beta,e}(2\pi)$ is semi-simple.  In contrast, Lagrangian orbits are called hyperbolic if no eigenvalue of $\ga_{\beta,e}(2\pi)$ locates on  $\mathbb{U}$.

The linear stability of relative equilibria ($e=0$) were known more than a century ago and it is due to Gascheau
(\cite{Ga}, 1843) and Routh (\cite{R2}, 1875) independently. For the elliptic  relative equilibria ($e>0$), the linear stability problem
is difficult, many interesting results could be found in \cite{MS},  \cite{MSS1}, \cite{MSS2}, \cite{R1}.  For  the historical literature  on linear stability of Lagrangian orbits, readers are referred to \cite{HLS}. Recently, Y.Long, S.Sun and the first author introduced   Maslov-type index and operator theory in studying the stability in $n$-body problem \cite{HLS},\cite{HS1}.
In \cite{HLS}, the authors  gave an analytic proof for the the stability  bifurcation diagram of Lagrangian equilateral triangular homographic orbits in the $(\beta; e)$ rectangle $[0, 9] \times [0, 1)$ and  proved that bifurcation curve is real analytic.  But it is difficult to estimate the bifurcation curve.

To the best of our  knowledge,  we don't know any result before to estimate the stability region. For the hyperbolic region,  till now,  we only know two results. Firstly,  it was proved in  \cite{HLS}
 that the Lagrangian orbits is hyperbolic for $\beta=9$ (equal mass case) with any eccentricity $e\in[0, 1)$. Secondly, based on the result in \cite{HLS}, it was  proved by
 the second author \cite{ou}  that Lagrangian orbits are hyperbolic for $\beta>8$. However, for $\beta$ near 1, we know nothing about the estimation of the hyperbolic region before.
 In the present paper, based on works in \cite{HLS},\cite{HS1} and via trace formula, we estimate the stability region and hyperbolic region for the elliptic Lagrangian orbits.

\begin{thm}\label{la1.1} The elliptic Lagrangian orbits is linear stable if \bea e<\frac{1}{1+f(\bb,-1)^{\frac{1}{2}}}, \,\ \bb\in[0,3/4),\nonumber \lb{th1.1a}  \eea
or
 \bea e<min\left\{{1\over \sqrt{f(\beta,-1)}},\frac{1}{1+\sqrt{f(\beta,e^{i\sqrt{2}\pi})}}  \right\}, \,\ \bb\in(3/4,1),\nonumber\lb{th1.1b}  \eea
 where $f(\beta,\omega)$ is a function on $[0,9]\times \mathbb{U}$  given by (\ref{fbb}). Let $\hat{f}(\beta)=\sup\{f(\bb,\omega),\omega\in\mathbb{U}\}$, then
  for $\bb\in(1,9]$, $\gamma_{\bb,e}$ is hyperbolic if \bea e<\hat{f}(\bb)^{-1/2}. \lb{th1.1c} \eea

\end{thm}
It will be seen that $f(\beta, \omega)$ is a elementary function determined by the trace formula.   By Theorem \ref{la1.1}, we can draw a picture as follows.

\begin{figure}[H]
 \centering
   \includegraphics[height=0.34\textwidth,width=0.94\textwidth]{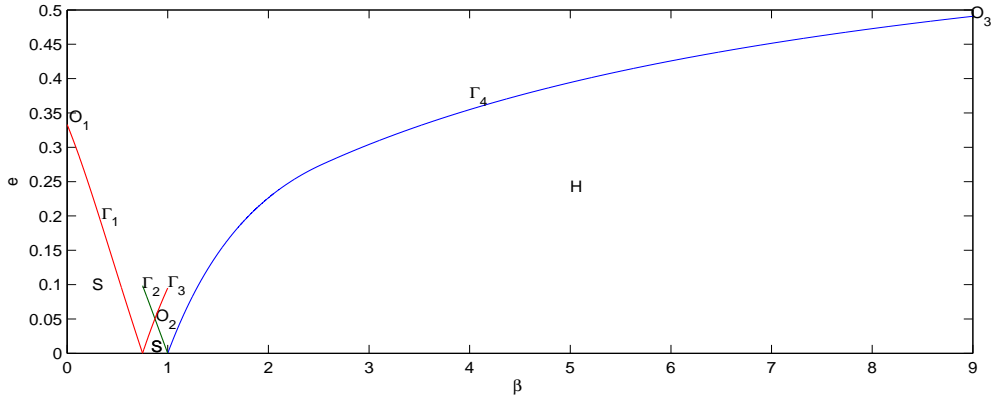}
     \caption{The stable region S and hyperbolic region H given by Theorem \ref{la1.1}.}
\end{figure}
\noindent In Figure 1, the points $O_{1}\approx(0, 0.3333)$, $O_{2}\approx(0.8730, 0.0504)$, $O_{3}\approx(9, 0.4907)$. The curves  $$\Gamma_{1}=\left\{(\beta,e)\,\left|\,  e=1\left/(1+\sqrt{f(\beta,-1)}), 0\leq\beta\leq3/4\right.\right.\right\},\quad \Gamma_{2}=\left\{(\beta,e)\,\left|\, e=1\left/\sqrt{f(\beta,-1)}, 3/4\leq\beta\leq 1\right.\right.\right\},$$
and $$     \Gamma_{3}=\left\{(\beta,e)\,\left|\, e=1\left/(1+\sqrt{f(\beta,e^{i\sqrt{2}\pi})}), 3/4\leq\beta< 1\right.\right.\right\},\quad \Gamma_{4}=\left\{(\beta,e)\,\left|\, e=1\left/\sqrt{\hat{f}(\beta)}, 1\leq\beta\leq 9\right.\right.\right\} .$$

This paper is organized as follows. In Section 2, we give the proof
of the trace formula for linear Hamiltonian systems. Moreover, some application of the trace
formula on the  identity which  related to the Zeta function is
given. In Section 3, we prove the Hill-type formula and trace formula for Sturm-Liouville systems.  The applications of the trace formula on the study of stability
for Hamiltonian systems  are given in Section 4, where we estimate
the relative Morse index (Morse index for Sturm-Liouville systems) and  some new stability criteria will be given.
The study of stability of elliptic Lagrangian solutions will be given in Section 5.

\section{Trace formula for linear Hamiltonian system}\label{sec2}

In this section, we will give the proof of the trace formula for linear Hamiltonian system. As been pointed out in the introduction, we will consider the Taylor expansion for the conditional Fredholm determinant of Hamiltonian system and the Monodromy matrices separately in \S 2.1 and \S 2.2. Based on it,  we prove Theorem \ref{thm1.1} in \S 2.3, some example on infinite identity and relation with the Zeta function is discussed.

\subsection{Taylor expansion for conditional Fredholm determinant of the linear perturbation  of  Hamiltonian system}\label{subsec2.1}

In this subsection, we will mainly consider the Taylor expansion of the
conditional Fredholm determinant  for linearly parameterized
Hamiltonian system. Let $B(\alpha):\Omega\to C([0,T], \mathcal{M}(2n,\mathbb{C}))$ be an
analytic function. For that $(A-B-\nu J)$ is invertible,  denote by \begin{eqnarray*}
p(\alpha)=\det\Big(id-(B_\alpha-B_{\alpha_0})(A-B_{\alpha_0}-\nu
J)^{-1}\Big).
\end{eqnarray*}
Notice that $(B_\alpha-B_{\alpha_0})(A-B_{\alpha_0}-\nu
J)^{-1}$ is not trace class but Hilbert-Schmidt. Hence  $p(\alpha)$ is not the usual Fredholm determinant, but a kind of conditional Fredholm determinant. The theory of conditional Fredholm determinant was studied in \cite{HW}. For readers convenience, we recall it briefly.
For  integer $N>0$, let $P_N$ be the projection onto the subspace
\begin{eqnarray*}W_N=\bigoplus\limits_{\nu\in\sigma(A),|\nu|\leq
N}\ker(A-\nu).\end{eqnarray*}
We need the following definition, which comes from \cite{HW}.
\begin{defi}\label{def1} For a Hilbert-Schmidt operator $F$, it is said to
have the \emph{trace finite condition}, if the limit
$\lim\limits_{N\to\infty}Tr\big(P_N F P_N\big)$ exists, which is
called the conditional trace and denoted by $Tr(F)$ without
confusion.
 \end{defi}

Obviously, if $F$ is a trace class operator, then the conditional
trace coincides with the traditional trace.  Moreover, if both  $F$ and $\widetilde{F}$ have
the trace finite condition, then $F+\widetilde{F}$ has the trace finite condition.
Now, for a Hilbert-Schmidt operator $F$ with trace finite condition,
by \cite{HW}, the limit
\begin{eqnarray*}  \det(id+F)=\lim\limits_{N\to\infty} \det(id+P_N F P_N) \label{00.1}  \end{eqnarray*}
is well defined, which depends on $\{P_N\}$ and is called the
\emph{Conditional Fredholm Determinant} of $id+F$.

By \cite[Corollary 3.4]{HW}, we know that $p(\alpha)$ is analytic on
$\Omega$. Now, for $B,D\in\mathcal{B}(2n)$, let $B(\alpha)=B+\alpha D$,
we have the following Theorem.
\begin{thm}\label{thm3.1b}
Let $f(\alpha)=\det\Big[(A-B-\alpha D-\nu J)(A+P_0)^{-1}\Big]$, suppose $A-B-\nu J$ is invertible, then
the Taylor expansion of $f$ at $0$ is
$f(\alpha)=\sum\limits_{m=0}^\infty \hat{b}_m \alpha ^m$, where
\begin{eqnarray*}
\hat{b}_m={a_m\over m!}\det((A-B-\nu J)(A+P_0)^{-1}),
\end{eqnarray*}
and
\begin{eqnarray}\label{eq3.1.1a}
a_m=(-1)^m\det\left( \begin{array}{cccc}Tr(F)& m-1& \cdots & 0 \\
Tr(F^2)& Tr(F) & \cdots & 0\\ \vdots & \vdots & \ddots & \vdots
\\ Tr(F^m) & Tr(F^{m-1})& \cdots & Tr(F)
\end{array}\right),
\end{eqnarray}
with $F= D (A-B-\nu J)^{-1}$.
\end{thm}
We first prove the following simple lemma.
 \begin{lem}\label{lem3.2b}
Let $B(\alpha):\Omega\to C([0,T], \mathcal{M}(2n,\mathbb{C}))$ be an analytic mapping. Write
\begin{eqnarray*}
p_N(\alpha)=\det\Big(id-P_N(B_\alpha-B_{\alpha_0})(A-B_{\alpha_0}-\nu
J)^{-1}P_N\Big),
\end{eqnarray*} then $p_N(\alpha)$ is analytic on $\Omega$.
\end{lem}
\begin{proof}
Let $\{e_j\}_{j=1}^\infty$ be an orthonormal basis, defined by the
eigenvectors of $A$. Set
\begin{eqnarray*}
F(\alpha)=(B_\alpha-B_{\alpha_0})(A-B_{\alpha_0}-\nu J)^{-1},
\end{eqnarray*}
then $F(\alpha)$ can be considered as an infinite matrix $(\lan
F(\alpha)e_j,e_i\ran)_{i,j}$. Notice that $\lan
F(\alpha)e_j,e_i\ran$ is an analytic function on $\alpha$, which
implies that $P_N F(\alpha) P_N$ is an analytic function on
$\Omega$. By the definition of $p_N(\alpha)$, we know that
$p_N(\alpha)$ is analytic.
\end{proof}

To prove Theorem \ref{thm3.1b},  write
\begin{eqnarray*}f_N(\alpha)=\det\Big(id-\alpha P_N D(A-B-\nu
J)^{-1}P_N\Big).\end{eqnarray*} Firstly, please note that
$f_N(\alpha)$ is analytic. Secondly, we will show that there is a
subsequence of $\{f_N(\alpha)\}$, which is convergent uniformly on
any compact subset of $\Omega$. Obviously, $f_N(\alpha)\to
f(\alpha)$ point-wisely on $\Omega$. Thirdly, by the theory in
\cite{Si}, we will give the expansion of $f_N(\alpha)$. Finally, by
the convergence of $f_N(\alpha)$, we get the Taylor expansion of
$f(\alpha)$.

 To prove that
there is a subsequence of $\{f_N(\alpha)\}$, which is convergent
uniformly to $f(\alpha)$ on any compact subset, we will recall some
properties of conditional Fredholm determinant and conditional
trace.

Recall that, if $\hat{F}$ is a Hilbert-Schmidt operator, then
$(id+\hat{F})e^{-\hat{F}}-id$ is a trace class operator, thus, we
can define
\begin{eqnarray*}
{\det}_2(id+\hat{F})=\det((id+\hat{F})e^{-\hat{F}}).
\end{eqnarray*}
 In the classical settings, if $\hat{F}$ is trace class, then ${\det}_2(id +\hat{F})=\det(id+\hat{F})e^{-Tr (\hat{F})}$. Inspired from this, in the case that $\hat{F}$ has the trace finite condition, we proved in \cite{HW},
\begin{eqnarray}\label{eq2.1b}
{\det}_2(id+\hat{F})=\det((id+\hat{F}))e^{-Tr (\hat{F})}
\end{eqnarray}
still holds, however, where $Tr( \hat{F})=\lim\limits_{N\to\infty}Tr(
P_N \hat{F} P_N)$ is the conditional trace.

The conditional Fredholm determinant preserves  almost all the
properties that the determinant of matrix has.   Such as, the
multiplicity of the determinant. Let $\hat{D}$ and $\hat{F}$ be two
Hilbert-Schmidt operators which have trace finite condition. Then
\begin{eqnarray}\label{eq2.2b}
\det(id+\hat{D})\det(id+\hat{F})=\det(id+\hat{D}+\hat{F}+\hat{D}\hat{F}),
\end{eqnarray}
where ``$\det$'' represents conditional Fredholm determinant. Similar to \cite[Proposition 3.2]{HW},  we have the following lemma. The proof of the lemma is almost the same as that was given for \cite[Proposition 3.2]{HW}, and we will omit the proof.
\begin{lem}\label{lemma1}
 Under the assumption of Lemma \ref{lem3.2b}, $\{p_N\}$ is a normal
family, that is, for any sequence in $\{p_N\}$ there is a subsequence which is uniformly convergent on any compact subset of $\mathbb C$.
\end{lem}

 For $\nu\in \mathbb C$ such that $A-B-\nu J$ is invertible,
\begin{eqnarray*}
(A-B-\alpha D-\nu J)(A+P_0)^{-1}=(id-\alpha D (A-B-\nu
J)^{-1})(A-B-\nu J)(A+P_0)^{-1}.
\end{eqnarray*}
Now set
\begin{eqnarray*}
g(\alpha)=\det(id-\alpha D(A-B-\nu J)^{-1}),
\end{eqnarray*}
where the ``$\det$'' is the conditional Fredholm determinant, and
\begin{eqnarray*}
g_N(\alpha)=\det (id-\alpha P_N D(A-B-\nu J)^{-1}P_N).
\end{eqnarray*}
By Lemma \ref{lemma1}, $f_N$ and $f$ are entire functions, and there is a
subsequence $\{f_{N_k}\}$ which is convergent to $f$ uniformly on any
compact subset in $\Omega$. Set  \begin{eqnarray*}F_N=P_N D (A-B-\nu
J)^{-1}P_N,\end{eqnarray*} then all of $F_N$ are finite-rank
operators, hence they are trace class operators, by \cite[Theorem
5.4]{Si}, we have the following lemma.
\begin{lem}
Let $g_N(\alpha)=\det (id+\alpha (-F_N))$. Then  the Taylor
expansion near $0$ for $g_N(\alpha)$ is
\begin{eqnarray*}
g_N(\alpha)=\sum\limits_{m=0}^\infty\alpha^m a_{N,m}/m! ,
\end{eqnarray*}
where
\begin{eqnarray*}
a_{N,m}=(-1)^m\det\left( \begin{array}{ccccc}Tr(F_N)& m-1&0 &\cdots & 0 \\ Tr(F_N^2)& Tr(F_N) & m-2 &\cdots & 0\\ \vdots & \vdots &\ddots& \ddots & \vdots \\
Tr(F_N^{m-1}) & Tr(F_N^{m-2}) & \cdots & Tr(F_N) & 1\\
 Tr(F_N^m) & Tr(F_N^{m-1})& \cdots & Tr(F_N^2) & Tr(F_N)   \end{array}\right).
\end{eqnarray*}
\end{lem}
Let $h_n$ be a sequence of analytic functions, which is convergent
to $h$ uniformly on any compact subset. Write the power series
expansions as
\begin{eqnarray*}
h_n(\alpha)=\sum\limits_{m=0}^\infty c_{n,m} \alpha^m, \quad
\text{and}\quad g(\alpha)=\sum\limits_{m=0}^\infty c_m \alpha^m,
\end{eqnarray*}
then, it is easy to see that $c_{n,m}$ converges to $c_m$ as
$n\to\infty$.

\vskip2mm\noindent\emph{Proof of Theorem \ref{thm3.1b}.} Now, notice
that $F=D (A-B-\nu J)^{-1}$ is a Hilbert-Schmidt operator with trace
finite condition, hence  the conditional trace
\begin{eqnarray*}
Tr(F)=\lim\limits_{N\to\infty} Tr(F_N).
\end{eqnarray*}
Set
\begin{eqnarray*}
a_m=(-1)^m\det\left( \begin{array}{cccc}Tr(F)& m-1& \cdots & 0 \\
Tr(F^2)& Tr(F) & \cdots & 0\\ \vdots & \vdots & \ddots & \vdots
\\ Tr(F^m) & Tr(F^{m-1})& \cdots & Tr(F)
\end{array}\right),
\end{eqnarray*}
Then $a_{N,m}$ tends to $a_m$ as $N\to\infty$. By Lemma
\ref{lemma1}, there is a subsequence $g_{N_j}(\alpha)$ of
$g_N(\alpha)$, which is convergent to $g(\alpha)$ on any compact
subset.  Then  \begin{eqnarray*} g(\alpha)=\sum\limits_{m=0}^\infty
{a_{m}\over m!} \alpha^m.
\end{eqnarray*}
 Since $$f(\alpha)=g(\alpha)\det((A-B-\nu J)(A+P_0)^{-1}),$$ we have
\begin{eqnarray*}
f(\alpha)=\sum\limits_{m=0}^\infty \alpha^m \big[{a_m\over
m!}\det((A-B-\nu J)(A+P_0)^{-1})\big]
\end{eqnarray*}
The proof is finished.\hfill$\Box$

Note that for $\alpha$
small, by \cite[p.47, (5.12)]{Si}, for a matrix $D$,
\begin{eqnarray} \det(I+\alpha D)=\exp\Big(\sum\limits_{m=1}^\infty\frac{(-1)^{m+1}}{m}\alpha^m Tr(D^m)\Big).  \label{cc4.20}  \end{eqnarray}
Thus for $\alpha$ small enough, write $g_N(\alpha)=e^{h_N(\alpha)}$, then
\begin{eqnarray*}
h_N(\alpha)=\sum\limits_{m=1}^\infty (-1)^{m+1}d_m(N) \alpha^m/m,
\end{eqnarray*}
with $d_m(N)=Tr ((-F_N)^m)$. On the other hand, since $(A-B-\nu J)$ is invertible, hence, $id-\alpha
D(A-B-\nu J)^{-1}$ is invertible in a neighborhood of $0$. It
follows that $g(\alpha)$ vanishes nowhere  in a neighborhood of $0$.
Write $g(\alpha)=e^{h(\alpha)}$ near $0$ with
\begin{eqnarray*}
h(\alpha)=\sum\limits_{m=1}^\infty (-1)^{m+1}d_m \alpha^m/m
\end{eqnarray*} be the Taylor expansion for $h(\alpha)$. Since $g_N$ converge to $g$ and is normal family, we have that $d_m=(-1)^mTr (F^m)$. We get the following theorem, which is the main result in this section.
\begin{thm}\label{thm3.3b}
Under the above assumption, we have
\begin{eqnarray*}
f(\alpha)=\det((A-B-\nu J)(A+P_0)^{-1})
\exp\Big\{\sum\limits_{m=1}^\infty b_m \alpha^m\Big\},
\end{eqnarray*}
where $b_m=-{{1}\over m}{Tr(F^m)}$.
\end{thm}

Notice that $F$ is a Hilbert-Schmidt operator with trace finite condition. Hence, $Tr(F)$ is not the usual  trace of $F$, but the conditional trace. However, in \cite[Theorem 5.4]{Si} $F$ is a trace class operator, and then $Tr(F)$ is the usual trace.

\subsection{Taylor expansion for linearly parameterized  Monodromy matrices  }\label{subsec3.1}

Set $B_\alpha=B+\alpha D$, for $\alpha\in\mathbb C$, let $\gamma_\alpha$ be the corresponding fundamental solutions, that is
\begin{eqnarray*}
\dot{\gamma}_\alpha(t)=J B_\alpha(t)\gamma_\alpha(t) \label{c4.1}.
\end{eqnarray*}
Fixed $\alpha_0\in\mathbb C$, direct computation  shows that
\begin{eqnarray*}
\frac{d}{dt}(\gamma_{\alpha_0}^{-1}(t)\gamma_\alpha(t))&=& \gamma_{\alpha_0}^{-1}(t)J(B_\alpha(t)-B_{\alpha_0}(t))\gamma_\alpha(t) \nonumber \\
&=& J(\gamma_{\alpha_0}^T(t)(B_\alpha(t)-B_{\alpha_0}(t))
\gamma_{\alpha_0}(t))\gamma_{\alpha_0}^{-1}(t)\gamma_\alpha(t)) \nonumber \\
&=& (\alpha-\alpha_0)J(\gamma_{\alpha_0}^T(t)D(t)
\gamma_{\alpha_0}(t))\gamma_{\alpha_0}^{-1}(t)\gamma_\alpha(t)
 \label{c4.2}.
\end{eqnarray*}
Without loss of generality, assume $\alpha_0=0$. In what follows, write
\begin{eqnarray*}\hat{\gamma}_\alpha(t)=\gamma_{0}^{-1}(t)\gamma_\alpha(t),\end{eqnarray*} and \begin{eqnarray*}\hat{D}(t)=\gamma_{0}^T(t)D(t) \gamma_{0}(t),\end{eqnarray*} thus
\begin{eqnarray}
\frac{d}{dt}\hat{\gamma}_\alpha(t)= \alpha
J\hat{D}(t)\hat{\gamma}_\alpha(t) \label{c4.3}.
\end{eqnarray}
To simplify the notation, we use ``$^{(k)}$'' to denote the $k$-th
derivative on $\alpha$. Taking derivative on $\alpha$ for both sides
of (\ref{c4.3}), we get
\begin{eqnarray}
\frac{d}{dt}\hat{\gamma}^{(1)}_\alpha(t)&=&J\hat{D}(t)\hat{\gamma}_\alpha(t)+\alpha
J\hat{D}_\alpha(t)\hat{\gamma}^{(1)}_\alpha(t)\label{c4.4}.
\end{eqnarray}
By taking $\alpha=0$,  $\hat{\gamma}_{0}(t)\equiv I_{2n}$, we have
\begin{eqnarray*}
\hat{\gamma}^{(1)}_{0}(t)=J\int_0^t\hat{D}(s)ds. \label{c4.5}
\end{eqnarray*}
Now, taking derivative on $\alpha$ for both sides of (\ref{c4.4}), we get
\begin{eqnarray*}
\frac{d}{dt}\hat{\gamma}^{(2)}_\alpha(t)&=&2J\hat{D}(t)\hat{\gamma}^{(1)}_\alpha(t)
+J\alpha\hat{D}(t)\hat{\gamma}^{(2)}_\alpha(t). \label{c4.6}
\end{eqnarray*}
Take $\alpha=0$, and we get
\begin{eqnarray*}
\hat{\gamma}^{(2)}_{0}(t)=2
J\int_0^t\hat{D}(s)\hat{\gamma}^{(1)}_{0}(s) ds. \label{c4.7}
\end{eqnarray*}
By induction,
\begin{eqnarray*}
\frac{d}{dt}\hat{\gamma}^{(k)}_0(t)=k
 J\hat{D}(t)\hat{\gamma}^{(k-1)}_0(t), \label{c4.9}
\end{eqnarray*}
and
\begin{eqnarray*}
\hat{\gamma}^{(k)}_0(t)=k
J\int_0^t\hat{D}(s)\hat{\gamma}^{(k-1)}_0(s)ds. \label{c4.10}
\end{eqnarray*}
For $t=T$,  by  Taylor's formula,
\begin{eqnarray*}
\hat{\gamma}_{\alpha}(T)=I_{2n}+\alpha\hat{\gamma}^{(1)}_{0}(T)+\cdots+\alpha^k\hat{\gamma}^{(k)}_{0}(T)/k!
+\cdots, \label{c4.11}
\end{eqnarray*}
where
\begin{eqnarray*}\hat{\gamma}^{(1)}_{0}(T)=\int_0^TJ\hat{D}(t)dt
\end{eqnarray*}and
\begin{eqnarray*}
\hat{\gamma}^{(k)}_{0}(T)/k!=\int_0^TJ\hat{D}(t)\hat{\gamma}^{(k-1)}_{0}(t)/(k-1)!dt,
k\in\mathbb N  \label{c4.13}.  \end{eqnarray*} By induction, we have
\begin{eqnarray*}
\hat{\gamma}^{(k)}_{0}(T)/k!=\int_0^TJ\hat{D}(t_1)\int_0^{t_1}J\hat{D}(t_2)\cdots\int_0^{t_{k-1}}J\hat{D}(t_k)dt_k\cdots dt_2dt_1,
k\in\mathbb N\label{adc4.13}.  \end{eqnarray*}
 Obviously
$\hat{\gamma}_{\alpha}(T)$ is an entire function on the variable $\alpha$. We summarize the above reasoning as the following proposition.
\begin{prop}
Let $B_\alpha=B+\alpha D$, $\gamma_\alpha(T)$ be the corresponding fundamental solutions. Write $\hat{\gamma}_\alpha=\gamma_0^{-1} \gamma_\alpha$. Then, the Taylor expansion for $\hat{\gamma}_\alpha(T)$ at $0$ is
\begin{eqnarray*}
\hat{\gamma}_{\alpha}(T)=I_{2n}+\alpha\hat{\gamma}^{(1)}_{0}(T)+\cdots+\alpha^k\hat{\gamma}^{(k)}_{0}(T)/k!
+\cdots, \label{c4.11a}
\end{eqnarray*}
where
\begin{eqnarray*}
\hat{\gamma}^{(k)}_{0}(T)/k!=\int_0^TJ\hat{D}(t_1)\int_0^{t_1}J\hat{D}(t_2)\cdots\int_0^{t_{k-1}}J\hat{D}(t_k)dt_k\cdots dt_2dt_1,
k\in\mathbb N\label{adc4.13a}.  \end{eqnarray*}
\end{prop}

In what follows, to simplify the notation,  set \begin{eqnarray*}
M(\alpha)=\hat{\gamma}_{\alpha}(T),\quad
M_0=I_{2n} \quad \text{and}\quad M_j=\hat{\gamma}^{(j)}_{0}(T)/j!,\,j\in\mathbb N,\end{eqnarray*}  then
$$M(\alpha)=\sum_{j=0}^\infty \alpha^jM_j.$$
Direct computation shows that
\begin{eqnarray*}  M(\alpha)^TJM(\alpha)=J+\alpha C_1+\alpha^2 C_2+\cdots+\alpha^k C_k+\cdots\label{c4.14}  \end{eqnarray*}
where $C_1=M_1^TJ+JM_1$, $C_2=M_2^TJ+JM_2+M_1^TJM_1$, and  in
general
\begin{eqnarray*}  C_k=\sum\limits_{j=0}^kM^T_jJM_{k-j},\,\ k\in\mathbb N  \label{c4.15}.  \end{eqnarray*}
By the fact that  $M(\alpha)\in {\mathrm {Sp}}(2n)$,
$M(\alpha)^TJM(\alpha)=J$, thus $C_k=0$ for $ k\in\mathbb N $. We have the following proposition.
\begin{prop}\label{prop3.1} Under the above assumptions
\begin{eqnarray} \sum\limits_{j=0}^kM^T_jJM_{k-j}=0,\,\ \forall\, k\in\mathbb N \label{c4.16}.  \end{eqnarray}
\end{prop}
Please note that, by taking $k=1$ in (\ref{c4.16}), we have
\begin{eqnarray}
 J M_1+M^T_1 J=0,
\end{eqnarray} which coincides with the fact that $JM_1$
is a symmetric matrix. Now,   multiplying  $-J$ on both sides  of (\ref{c4.16}) and taking  trace,  we have
\begin{cor}\label{cor3.1}Under the above assumptions
\begin{eqnarray*} \sum\limits_{j=0}^m  Tr(-JM^T_jJM_{m-j})=0,\,\ \forall m\in\mathbb N \label{c4.17}.  \end{eqnarray*}
\end{cor}
Especially,  for $m=2$, we get
\begin{eqnarray*} 2 Tr(M_2)=Tr(JM_1^TJM_1)=Tr(M_1^2) \label{c4.18}.  \end{eqnarray*}

Set $M=S\gamma_0(T)$, then $S\gamma_\alpha(T)=M M(\alpha)$. For
$\lambda\in\mathbb C$, which is not an eigenvalue of $M$,  by some easy computations, we have that
\begin{eqnarray*}
 \det(S\gamma_\alpha(T)-\lambda I_{2n}) &=& \det(MM(\alpha)-\lambda I_{2n}) \nonumber \\  &=&  \det(M-\lambda I_{2n}+\alpha M
M_1+\cdots+\alpha^k M M_k+\cdots) \nonumber \\  &=& \det(M-\lambda
I_{2n})\det(I+\cdots+\alpha^k(M-\lambda I_{2n})^{-1}M M_k+\cdots).
\label{c4.19}
\end{eqnarray*}
Let $G_k=(M-\lambda I_{2n})^{-1}M M_k$,
$$f(\alpha)=\det(I+\cdots+\alpha^kG_k+\cdots), $$ which is an analytic
function  on $\mathbb C$. Next, we will compute the Taylor expansion for $f(\alpha)$.
Let $G(\alpha)=\sum\limits_{k=1}^\infty\alpha^{k-1}G_k$, then
for $\alpha$ small enough, by (\ref{cc4.20}), we have
\begin{eqnarray} f(\alpha) &=& \det(I+\alpha G(\alpha)) \nonumber \\  &=& \exp\Big(\sum_{m=1}^\infty\frac{(-1)^{m+1}}{m}\alpha^m Tr\big(G(\alpha)^m\big)\Big) \nonumber \\  &=&
\exp\Big(\sum_{m=1}^\infty\frac{(-1)^{m+1}}{m}\alpha^m
Tr\Big[\Big(\sum_{k=1}^\infty \alpha^{k-1}G_k\Big)^m\Big]\Big) \nonumber \\  &=&
\exp\Big(\sum_{m=1}^\infty\frac{(-1)^{m+1}}{m}\Big[\sum_{k_1,\cdots,k_m=1}^\infty\alpha^{k_1+\cdots+k_m}Tr(G_{k_1}\cdots
G_{k_m})\Big]\Big).
 \label{cc4.21}  \end{eqnarray}
Since $f(\alpha)$ vanishes nowhere near $0$, we can write  $f(\alpha)=e^{g(\alpha)}$, then by (\ref{cc4.21}), some direct computation  shows that
\begin{eqnarray}  g^{(m)}(0)/m != \sum_{k=1}^m \frac{(-1)^{k+1}}{k}\Big(\sum_{j_1+\cdots+j_k=m}Tr(G_{j_1}\cdots
G_{j_k})\Big). \label{cc4.22}  \end{eqnarray}
For $\alpha$
small enough, let $g(\alpha)$ be the function satisfying  \begin{eqnarray}
\lambda^{-n}\det(S\gamma_\alpha(T)-\lambda
I_{2n})=\lambda^{-n}\det(M-\lambda I_{2n})\cdot \exp(g(\alpha))
\label{cc4.29},
\end{eqnarray}
then the coefficients of $g^{(k)}(0)/k!$ could be determined  by
(\ref{cc4.22}). And we have the following theorem, which is the main result in this subsection.

\begin{thm}\label{thm3.7b}
Under the above assumption, let $g(\alpha)$ be the function in
(\ref{cc4.29}). Let $g(\alpha)=\sum\limits_{m=1}^\infty c_m
\alpha^m$ be its Taylor expansion. Then
\begin{eqnarray*}
c_m=\sum_{k=1}^m \frac{(-1)^{k+1}}{k}\Big(\sum_{j_1+\cdots+j_k=m}Tr(G_{j_1}\cdots
G_{j_k})\Big).
\end{eqnarray*}
\end{thm}
 We only list the first $4$
terms
\begin{eqnarray*} g^{(1)}(0)=Tr(G_1) \label{cc4.23},
\end{eqnarray*}
\begin{eqnarray*} g^{(2)}(0)/2=Tr(G_2)-\frac{1}{2}Tr(G_1^2) \label{cc4.24},
\end{eqnarray*}
\begin{eqnarray*} g^{(3)}(0)/3!=Tr(G_3)-Tr(G_1G_2)+\frac{1}{3}Tr(G_1^3) \label{cc4.25},
\end{eqnarray*}
\begin{eqnarray*} g^{(4)}(0)/4!=Tr(G_4)-\frac{1}{2}Tr(G_2^2)-Tr(G_1G_3)+Tr(G_1^2G_2)-\frac{1}{4}Tr(G_4)
\label{cc4.26}.
\end{eqnarray*}
By the definition of $G_k$,
\begin{eqnarray*}Tr(G_1)=Tr(M_1M(M-\lambda I_{2n})^{-1})=Tr \Big(J\int_0^T\hat{D}(s)ds \cdot M(M-\lambda I_{2n})^{-1} \Big)  \label{cc4.27},
\end{eqnarray*}
\begin{eqnarray*}Tr(G_2)=Tr(M_2M(M-\lambda I_{2n})^{-1})=Tr \Big(J\int_0^T\hat{D}(s)J\int_0^s\hat{D}(\sigma)d\sigma ds \cdot M(M-\lambda I_{2n})^{-1} \Big)
\label{cc4.28}.
\end{eqnarray*}
Generally, \begin{eqnarray*}
Tr(G_k^m)=Tr\Big(\Big[\int_0^TJ\hat{D}(t_1)\int_0^{t_1}J\hat{D}(t_2)\cdots\int_0^{t_{k-1}}J\hat{D}(t_k)dt_k\cdots dt_2 dt_1\cdot
M(M-\lambda I_{2n})^{-1}\Big]^m\Big)\label{cc4.28},  \end{eqnarray*} and
$Tr(G_{j_1}\cdots G_{j_k}) $  could be given similarly.

\subsection{The proof of the Trace formula for Hamiltonian system }\label{sec 3}
In this subsection,  we will give proof of Theorem \ref{thm1.1}.
\vskip2mm\noindent\emph{Proof of Theorem \ref{thm1.1}.}
 We begin with the formula
\begin{eqnarray*}  \det\((A-B-\alpha D-\nu J)(A+P_0)^{-1}\)=C(S)e^{-n\nu T}\det( S\gamma_\alpha(T)- e^{\nu T} I_{2n}). \end{eqnarray*}
On the one hand, by Theorem \ref{thm3.3b},
\begin{eqnarray*}
\det((A-B-\alpha D-\nu J)(A+P_0)^{-1})=\det((A-B-\nu J)(A+P_0)^{-1}) \exp\Big\{\sum\limits_{m=1}^\infty b_m \alpha^m\Big\},
\end{eqnarray*}
where $b_m=-{{1}\over m}{Tr\big((D(A-B-\nu J)^{-1})^m\big)}$. On the other hand, by Theorem \ref{thm3.7b},
\begin{eqnarray*}
C(S)e^{-n\nu T}\det( S\gamma_\alpha(T)- e^{\nu T} I_{2n})=C(S)e^{-n\nu T}\det(S\gamma(T)-e^{\nu T} I_{2n})\cdot \exp\(\sum\limits_{n=1}^\infty c_m \alpha^m\),
\end{eqnarray*}
where
\begin{eqnarray*}
c_m=\sum_{k=1}^m \frac{(-1)^{k+1}}{k}\Big(\sum_{j_1+\cdots+j_k=m}Tr(G_{j_1}\cdots
G_{j_k})\Big).
\end{eqnarray*}
Since \begin{eqnarray*}\det((A-B-\nu J)(A+P_0)^{-1})=C(S)e^{-n\nu T}\det(S\gamma(T)-e^{\nu T} I_{2n}),\end{eqnarray*} we have that $b_m=c_m$, that is
\begin{eqnarray*}\label{cc3.43}
-{{1}\over m}{Tr\big((D(A-B-\nu J)^{-1})^m\big)}=\sum_{k=1}^m \frac{(-1)^{k+1}}{k}\(\sum_{j_1+,...,+j_k=m}Tr(G_{j_1}\cdots
G_{j_k})\).
\end{eqnarray*}
It follows that,
\begin{eqnarray}\label{cc3.44a}
Tr \[\(D(A-B-\nu J)^{-1}\)^{m}\]=m\sum_{k=1}^m \frac{(-1)^{k}}{k}\[\sum\limits_{j_1+\cdots+j_k=m}Tr(G_{j_1}\cdots
G_{j_k})\].
\end{eqnarray}
The  proof is completed. \hfill$\Box$
\vskip2mm

By the equation (\ref{cc3.44a}), theoretically, we can calculate the
trace of $\big[D(A-B-\nu J)^{-1}\big]^{m}$, at least, numerically by
computer. Notice that the right hand side of (\ref{cc3.44a}) is a kind of multiple
integral, and it is a little complicated. Next, we will write down the
first four terms.
\begin{prop}
\begin{eqnarray*}
Tr (D(A-B-\nu J)^{-1})=-Tr(G_1).
\end{eqnarray*}
\begin{eqnarray*}
Tr \Big(\big[D(A-B-\nu J)^{-1}\big]^2\Big)=Tr(G_1^2)-2Tr(G_2).
\end{eqnarray*}
\begin{eqnarray*}
Tr \Big(\big[D(A-B-\nu J)^{-1}\big]^3\Big)=-3Tr(G_3)+3Tr(G_1G_2)-Tr(G_1^3).
\end{eqnarray*}
\begin{eqnarray*}
Tr \Big(\big[D(A-B-\nu J)^{-1}\big]^4\Big)=-4Tr(G_4)+2Tr(G_2^2)+4Tr(G_1G_3)-4Tr(G_1^2G_2)+Tr(G_1^4).
\end{eqnarray*}
\end{prop}
Moreover, for the first two terms, we can write it more precisely.
\begin{eqnarray}\label{eq3.49ab}
Tr[D(A-B-\nu J)^{-1}]=-Tr\Big(J\int_0^T\gamma_{0}^T(t)D(t)
\gamma_{0}(t)dt\cdot M(M-e^{\nu T} I_{2n})^{-1}\Big),
\end{eqnarray}
and
\begin{eqnarray}
&&Tr\big([D(A-B-\nu J)^{-1}]^2\big)\nonumber\\&&=-2Tr \Big(J\int_0^T\gamma_{0}^T(t)D(t) \gamma_{0}(t)J\int_0^s\gamma_{0}^T(s)D(s) \gamma_{0}(s)ds dt \cdot M(M-e^{\nu T} I_{2n})^{-1} \Big)\nonumber\\
                                &&\ \ \ +Tr\Big(\Big[J\int_0^T\gamma_{0}^T(t)D(t) \gamma_{0}(t)dtM(M-e^{\nu T} I_{2n})^{-1}\Big]^2\Big), \label{eq2.14}
\end{eqnarray}
which are (\ref{ht1}) and (\ref{ht2}) in Corollary \ref{cor1.1}.

It is worth to be pointed out that, on the left hand side of (\ref{eq3.49ab}),
the trace is the conditional trace, and on the right hand side of
it, it is the trace of matrix on $\mathbb C^{2n}$. Next, we will consider some special cases.
\begin{prop}
Assume that $B(t)\equiv B_0$ is a constant matrix and $S=\pm I_{2n}$, then,
\begin{eqnarray*}
Tr(D(A-B-\nu J)^{-1})=-Tr\(J\int_0^T D(t)dt\cdot M(M-e^{\nu T}
I_{2n})^{-1}\). \label{adpro4.1}
\end{eqnarray*}
\end{prop}
\begin{proof}
Since $B(t)\equiv B_0$, obviously $\gamma_0(t)=e^{JB_0t}$,  thus
$\gamma_0(t)$ commutes with $\gamma_0(T)$ and also commutes with $M$
since $S=\pm I_{2n}$. Easy computation  shows that
\begin{eqnarray*}
Tr\Big(J\int_0^T \hat{D}(t)dt\cdot M(M-e^{\nu T}
I_{2n})^{-1}\Big)&=&Tr\Big(\int_0^T e^{-JB_0t}JD(t)e^{JB_0t}dt\cdot
M(M-e^{\nu T} I_{2n})^{-1}\Big) \nonumber \\ &=&Tr\Big(\int_0^T
e^{-JB_0t}JD(t) M(M-e^{\nu T} I_{2n})^{-1}e^{JB_0t}dt\Big)\nonumber
\\ &=&Tr\Big(J\int_0^T D(t)dt\cdot M(M-e^{\nu T}
I_{2n})^{-1}\Big). \label{adpro4.2}
\end{eqnarray*}
By (\ref{eq3.49ab}), the proposition is proved.
\end{proof}

The following proposition considers the case that $MJ=JM$, $M^T=M$.
\begin{prop}\label{prop5b.21} If $MJ=JM$, $M^T=M$, then
 \begin{eqnarray*} Tr\Big(\big(D(A-\nu J-B)^{-1}\big)^{2}\Big)\nonumber &=&Tr\[\(J\int_0^T\hat{D}(s)ds \cdot M(M-e^{\nu T} I_{2n})^{-1}\)^2\]\\&&-Tr\[\(J\int_0^T\hat{D}(s)ds\)^2M(M-e^{\nu T} I_{2n})^{-1}\]. \label{5b.5}  \end{eqnarray*}
\end{prop}
\begin{proof}
Suppose $MJ=JM$, $M=M^T$ then \begin{eqnarray*}Tr(M_2M(M-\omega
I_{2n})^{-1})&=&Tr(-JM_2JM(M-e^{\nu T} I_{2n})^{-1}) \\ &=& Tr(-M(M-e^{\nu T} I_{2n})^{-1}JM^T_2J) \\ &=& Tr(-JM^T_2JM(M-e^{\nu T} I_{2n})^{-1}).\end{eqnarray*}
 By Proposition
\ref{prop3.1} \begin{eqnarray*} M_1^TJ+JM_1=0,  \label{5b.4.00}  \end{eqnarray*} and \begin{eqnarray*}
-JM_2^TJ+M_2=JM_1^TJM_1. \label{5b.4.01}  \end{eqnarray*} Thus
\begin{eqnarray*}2Tr(G_2)&=& Tr\(JM_1^TJM_1M(M-e^{\nu T} I_{2n})^{-1}\) \nonumber \\ &=& Tr\(M_1^2M(M-e^{\nu T} I_{2n})^{-1}\)\nonumber\\&=&Tr\[\(J\int_0^T\hat{D}(s)ds\)^2M(M-e^{\nu T} I_{2n})^{-1}\] \label{5b.4.1}.
\end{eqnarray*}
By the formula (\ref{eq2.14}), the proposition is proved.
\end{proof}

Some easy computation shows that, if moreover $M$ commutes with $J\int_0^T\hat{D}(s)ds $, then
 \begin{eqnarray} Tr\((D(A-\nu J-B)^{-1})^{2}\)&=&Tr\[\(J\int_0^T\hat{D}(s)ds\)^2M(M-e^{\nu T} I_{2n})^{-1}(M(M-e^{\nu T} I_{2n})^{-1}-I_{2n})\]\nonumber\\
 &=&e^{\nu T} Tr\[\(J\int_0^T\hat{D}(s)ds\)^2M(M-e^{\nu T} I_{2n})^{-2}\]. \label{5b.5.ad}  \end{eqnarray}
More specially, we have the following corollary.
\begin{cor}\label{cor5b.11}
If $M=\pm I_{2n}$,  then
\begin{eqnarray} Tr\big((D(A-\nu J-B)^{-1})^{2}\big)=\frac{\pm e^{\nu T}}{(1\mp e^{\nu T})^2} Tr\[\(J\int_0^T\gamma_0^T(s)D(s)\gamma_0(s)ds\)^2\]. \label{5b.5.1a}  \end{eqnarray}
 Especially in the case $B=0$, $\hat{D}=D$ and $S=\pm I_{2n}$,
\begin{eqnarray} Tr\big((D(A-\nu J)^{-1})^{2}\big)=\frac{\pm e^{\nu T}}{(1\mp e^{\nu T})^2} Tr\[\(J\int_0^TD(s)ds\)^2\]. \label{5b.6.1c}  \end{eqnarray}
\end{cor}
Notice that (\ref{5b.5.1a}) is just the formula (\ref{eq5b.5.1.1a}) in Corollary \ref{cor1.1}.

\begin{exam}
In the case $D(t)=I_{2n}$, then
$\hat{D}(t)=\gamma_0^T(t)\gamma_0(t)$, so we have
\begin{eqnarray*}
Tr((A-B-\nu J)^{-1})=
Tr\(J\int_0^T\gamma_0^T(s)\gamma_0(s)ds\cdot M(M-e^{\nu T}
I_{2n})^{-1}\),
\end{eqnarray*}
and for $k\geq2$,
\begin{eqnarray*} Tr \[\((A-B-\nu
J)^{-1}\)^k\]=\sum\limits_{j=-\infty}^\infty
\frac{1}{\lambda_j^k},
\end{eqnarray*}
where $\lambda_j$ are eigenvalues of $A-B-\nu J$. From the trace
formula, we have
\begin{eqnarray}\label{cc3.64}
\sum\limits_{j=-\infty}^\infty \frac{1}{\lambda_j^m}=m\sum_{k=1}^m
\frac{(-1)^{k}}{k}\Big[\sum\limits_{j_1+\cdots+j_k=m}Tr(G_{j_1}\cdots
G_{j_k})\Big], \,\ \forall m\geq2.
\end{eqnarray}
\end{exam}
The equation (\ref{cc3.64}) has its own interests. In fact, we can deduce some interesting equalities from this.
\begin{exam}
Let $B=0$, $D=I_{2}$,  $S=I_{2}$ and $T=1$. Then, for each fixed $\alpha\in \mathbb C$,  it is easy to check that the eigenvalues for $A-\nu J-\alpha$
are $\{2k\pi\pm\sqrt{-1}\nu-\alpha|k\in \mathbb{Z}\}$.   For $\nu\not\in 2\pi\sqrt{-1}\mathbb Z-\alpha$,
$A-\nu J-\alpha$ is invertible, and the left hand side of (\ref{cc3.64}) is \begin{eqnarray*} Tr((A-\nu J-\alpha)^{-m})= \sum\limits_{k\in\mathbb Z}\frac{1}{(2k\pi+\sqrt{-1}\nu-\alpha)^m}+\sum\limits_{k\in\mathbb Z}\frac{1}{(2k\pi-\sqrt{-1}\nu-\alpha)^m},\,\ \forall m\in \mathbb{N}, \label{iden1} \end{eqnarray*}
where for $m=1$, the infinite sum in the right side is understand by $\lim\limits_{\beta\rightarrow\infty}\sum_{|k|\leq\beta}$.
For the right hand side, the traces $Tr(G_{j_1}\cdots G_{j_k})$ can be calculated directly. We only list the first 3 equalities. For $m=1$,
direct computation shows that $Tr(G_1)=\frac{2e^\nu\sin\alpha}{(\cos\alpha-e^\nu)^2+\sin^2\alpha}$, thus we have
\begin{eqnarray*}  \sum\limits_{k\in\mathbb Z}\frac{1}{2k\pi+\sqrt{-1}\nu-\alpha}+\sum\limits_{k\in\mathbb Z}\frac{1}{2k\pi-\sqrt{-1}\nu-\alpha}=\frac{-2e^\nu\sin\alpha}{(\cos\alpha-e^\nu)^2+\sin^2\alpha}. \label{iden1.0}   \end{eqnarray*}
For $m=2$, by (\ref{5b.5.ad}), direct computation shows that
\begin{eqnarray*} Tr((A-\nu J-\alpha)^{-2})&=&\frac{-2e^\nu(\cos\alpha(1+e^{2\nu})-2e^\nu)}{e^\nu\cos\alpha(4e^{\nu}\cos\alpha -4e^{2\nu}-4)+(1+e^{2\nu})^2}\nonumber\\
&=&\frac{1-\cosh\nu\cos\alpha}{(\cos\alpha-\cosh\nu)^2},\label{iden1.2} \end{eqnarray*}
thus we have the identity
 \bea  \sum\limits_{k\in\mathbb Z}\frac{1}{(2k\pi+\sqrt{-1}\nu-\alpha)^2}+\sum\limits_{k\in\mathbb Z}\frac{1}{(2k\pi-\sqrt{-1}\nu-\alpha)^2}=\frac{1-\cosh\nu\cos\alpha}{(\cos\alpha-\cosh\nu)^2}. \label{iden1.1}   \eea
 Especially in the case $\alpha=0$,
$$Tr((A-\nu J)^{-2})=2\sum\limits_{k\in\mathbb Z}\frac{1}{(2k\pi+\sqrt{-1}\nu-\alpha)^2},$$  and 
 the right hand side of (\ref{iden1.1}) is reduced to $\frac{-2e^\nu}{(1-e^\nu)^2}$, thus  we have the identity
    \begin{eqnarray*}\label{eq3.3.1}
    \sum\limits_{k\in \mathbb Z}\frac{1}{(2k\pi+\sqrt{-1}\nu)^2}=\frac{1+\cos \sqrt{-1}\nu}{2\sin ^2 \sqrt{-1}\nu}.
    \end{eqnarray*}
 Similarly, for $m=3$, we get
 \bea  \sum\limits_{k\in\mathbb Z}\frac{1}{(2k\pi+\sqrt{-1}\nu-\alpha)^3}+\sum\limits_{k\in\mathbb Z}\frac{1}{(2k\pi-\sqrt{-1}\nu-\alpha)^3}\nonumber \\=\frac{1/2\sin\alpha(\cosh^2\nu+\cosh\nu\cos\alpha-2)}{\cosh^3\nu-3\cosh^2\nu\cos\alpha+3\cosh\nu\cos^2\alpha-\cos^3\alpha}. \label{iden1.3} \nonumber  \eea

\end{exam}
The equality in the above example can be  deduced by using techniques in complex analysis. However, the above example is only a kind of easiest case. If we take a non-constant path $B$, then  the formula will be far from  trivial.
\begin{rem}
Recall that, in \cite{APS}, Atiyah, Patodi and Singer defined a kind of zeta function for self-adjoint elliptic differential operator $\mathcal{A}$(the operator may be not positive). Let $\{\lambda\}$ be the eigenvalues for $\mathcal{A}$, then
\begin{eqnarray*}
\eta_{\mathcal{A}}(s)=\sum\limits_{\lambda\not=0}(\mathrm{sign} \lambda)|\lambda|^{-s},
\end{eqnarray*}
for $Re(s)$ large, and it can be extended
meromorphically to the whole $s$-plane.
Now, for the differential operator $A$, if we can take some proper $B$, $D$ and $S$ in our framework, such that $\lambda$ are the eigenvalues of
 $\mathcal{A}=D^{-1}(A-B-\nu J)$ is real, then by the trace formula, we can obtain the values for $\eta_\mathcal{A}(s)$ at odd integers.
\end{rem}

\section{Hill-type formula and  Trace formula for Sturm-Liouville systems}

In the study of  $\bar{S}$-periodic orbits in  Lagrangian systems, it is natural to consider the standard Sturm systems:
\bea -(P\dot{y}+Qy)^\cdot+Q^T\dot{y}+Ry=0, \quad
y(0)=\bar{S}y(T),\quad \dot{y}(0)=\bar{S}\dot{y}(T), \label{lag.1} \eea
where $\bar S$ is an orthogonal matrix on $\mathbb R^n$.
We assume $P(t)$ is invertible for any $t$, which is a more general condition than the usual Legendre convexity assumptions.
Denote $\hat{Q}=P^{-1}(Q^T-Q-\dot{P})$, $\hat{R}=P^{-1}(R-\dot{Q})$. Obviously, the system (\ref{lag.1}) is equivalent to
\bea
- \ddot{ z}(t)+\hat{Q}(t)\dot{z}(t)+\hat{R}(t)z(t)=0\label{1.1.1a1},\\
                                       z(0)=\bS z(T),   \quad                                  \dot{z}(0)&=&\bS\dot{z}(T).
                                       \label{1.1.3a1}
\eea

Please note that if $\hat{z}(t)$ satisfies the equation (\ref{1.1.1a1}) with $\hat{z}(0)=e^{-\nu T}\bS \hat{z}(T)$, $\dot{\hat{z}}(0)=e^{-\nu T}\bS \dot{\hat{z}}(T)$, then $z(t)=e^{\nu t}\hat{z}(t)$
satisfies the following   second order ODE
\begin{eqnarray}
- \Big(\d+\nu\Big)^2{ z}(t)+\hat{Q}(t)\Big(\d+\nu\Big){z}(t)+\hat{R}(t)z(t)&=&0,\label{1.1.1a} \\
z(0)=\bS z(T), \,\
\dot{z}(0)&=&\bS \dot{z}(T).\label{1.1.3a}
\end{eqnarray}
Let $y(t)=\dot{z}(t)+\nu z(t)$, then we can write (\ref{1.1.1a}-\ref{1.1.3a}) as the following first order ODE
\begin{eqnarray}
\left(\begin{array}{c}\dot{y}(t)\\ \dot{z}(t)\end{array}\right)&=&\left(\begin{array}{cc}\hat{Q}(t)-\nu & \hat{R}(t) \\I_n  & -\nu \end{array}\right)\left(\begin{array}{c} y(t)\\ z(t)\end{array}\right),\label{eq5.19}\\
\left(\begin{array}{c} y(0)\\ z(0)\end{array}\right)&=&\left(\begin{array}{cc} \bS & 0_n \\ 0_n & \bS \end{array}\right)\left(\begin{array}{c} y(T)\\ z(T)\end{array}\right).\label{eq5.20}
\end{eqnarray}
For simplicity, we denote \begin{eqnarray*}
\hat{B}(t)=\left(\begin{array}{cc}I_n& 0_n \\-\hat{Q}(t)  & -\hat{R}(t) \end{array}\right),\quad \bar{S}_d=\left(\begin{array}{cc} \bar{S} & 0_n \\ 0_n & \bar{S} \end{array}\right), \quad\text{and}\quad x(t)=\left(\begin{array}{c}y(t)\\ z(t)\end{array}\right).
\end{eqnarray*}
 The system (\ref{eq5.19}-\ref{eq5.20}) can be written as the following Hamiltonian system,
\bea  \dot{x}(t) &=& J(\hat{B}(t)+\nu J)x(t), \label{5.b1} \\ x(0)&=&\bar{S}_dx(T). \label{5.b2} \eea
It follows that $z(t)$  is solution of (\ref{1.1.1a}-\ref{1.1.3a}) if and only if $x(t)=\left(\begin{array}{c}\dot{z}(t)+\nu z(t)\\ z(t)\end{array}\right)$ is solution of (\ref{5.b1}-\ref{5.b2}).
Therefore, we have
\bea \dim\ker \(-\(\d+\nu\)^2+\hat{Q}(t)\(\d+\nu\)+\hat{R}(t)\)=\dim\ker \(-J\d-\hat{B}-\nu J \). \label{5.ke}  \eea

Now, we will give the Hill-type formula for indefinite Lagrangian system.  For $N\in\mathbb{N}$, let
$$
\hat{W}_N=\bigoplus\limits_{\nu\in\sigma(\d),|\nu|\leq N}\ker\(\d-\nu I_n\),
$$
and denote by $\hat{P}_N$ the orthogonal projection onto $\hat{W}_N$. Then $Q(t)\left(\d+\nu\right)^{-1}$ is a Hilbert-Schmidt operator with the trace finite condition with respect to $\{\hat{P}_N\}$.  We define the conditional Fredholm determinant with respect to $\hat{P}_N$, $${\det}\[\(-\(\d+\nu I_n\)^2+Q(t)\(\d+\nu I_n\)+R(t)\)\cdot\(-\(\d+\nu I_n\)^2\)^{-1}\].$$

At first, we recall   Hill-type formula for linear Hamiltonian systems \cite{HW}. For $B\in C([0,T];\mathcal{M}(2n,\mathbb{C}))$, which is not have to be real symmetric, we have that
\begin{eqnarray}
{\det}\(\(-J{d\over dt}-B-\nu J\)\(-J\d+P_0\)^{-1}\)=C(S)e^{-\frac{T}{2}\int_0^T Tr (JB(t)) dt}e^{-n\nu T}\det\( S\gamma(T)- e^{\nu T} I_{2n}\),\label{hg}
\end{eqnarray}
where $\gamma$ is the fundamental solution corresponding to $B$.
 We firstly prove  the following proposition.
\begin{prop}\label{thm5.1a} For $\nu\in\mathbb{C}$ such that $\d+\nu I_n$ is invertible, we have
\begin{eqnarray}\label{eq5.17}
&&{\det}\[\(-J{d\over dt}-\hat{B}-\nu J\)\(-J{d\over dt}-\nu J \)^{-1}\]\nonumber\\&&\ \ = {\det}\[\(-\(\d+\nu I_n\)^2+\hat{Q}(t)\(\d+\nu I_n\)+\hat{R}(t)\)\(-\(\d+\nu I_n\)^2\)^{-1}\].\end{eqnarray}
\end{prop}
\begin{proof} Let $K_n=\left(\begin{array}{cc} I_n & 0_n \\ 0_n & 0_n \end{array}\right)$,
note that
\begin{eqnarray}\label{eq4.17}
-J\d-\nu J -K_n=\left(\begin{array}{cc}-I_n & \d+\nu I_n \\ -\left(\d+\nu I_n\right) & 0_n \end{array}\right).
\end{eqnarray}
It follows that $\d+\nu I_n$ is invertible if and only if $-J\d-\nu J -K_n$ is invertible; moreover
\begin{eqnarray*}
\(-J\d-\nu J -K_n\)^{-1}=\left(\begin{array}{cc}0_n & -\left(\d+\nu I_n\right)^{-1}\\ \left(\d+\nu I_n\right)^{-1} & -\left(\d+\nu I_n\right)^{-2}\end{array}\right).
\end{eqnarray*}
It follows that
\begin{eqnarray*}
(K_n-\hat{B})\(-J\d-\nu J -K_n\)^{-1}=\left(\begin{array}{cc} 0_n & 0_n\\ \hat{R}(t)\(\d+\nu I_n\)^{-1} & -\hat{Q}(t)\(\d+\nu I_n\)^{-1}-\hat{R}(t)\(\d+\nu I_n\)^{-2}\end{array}\right).
\end{eqnarray*}
Thus we have
\bea && {\det}\[\(-J{d\over dt}-\hat{B}-\nu J\)\(-J{d\over dt}-K_n-\nu J \)^{-1}\]\nonumber \\ &=& \det\[id-(\hat{B}-K_n\hat{B})\(-J\d-\nu J -K_n\)^{-1} \]\nonumber
\\ &=& \det\[id -\hat{Q}(t)\(\d+\nu I_n\)^{-1}-\hat{R}(t)\(\d+\nu I_n\)^{-2}\] \nonumber\\ &=&{\det}\[\(-\(\d+\nu I_n\)^2+\hat{Q}(t)\(\d+\nu I_n\)+\hat{R}(t)\)\(-\(\d+\nu I_n\)^2\)^{-1}\]. \label{eq4.2}
\eea
Now, direct computation  shows that, \begin{eqnarray*}
\det\[\(-J\d-K_n-\nu J\)\(-J\d-\nu J\)^{-1}\]=1.
\end{eqnarray*}
Therefore,
\begin{eqnarray}
&&{\det}\[\(-J{d\over dt}-\hat{B}-\nu J\)\(-J{d\over dt}-\nu J \)^{-1}\]\nonumber\\&=&{\det}\[\(-J{d\over dt}-\hat{B}-\nu J\)\(-J{d\over dt}-K_n-\nu J \)^{-1}\]\cdot\det\[\(-J\d-K_n-\nu J\)\(-J\d-\nu J\)^{-1}\]\nonumber\\
&=&{\det}\[\(-J{d\over dt}-\hat{B}-\nu J\)\(-J{d\over dt}-K_n-\nu J \)^{-1}\].\label{eq4.1}
\end{eqnarray}
Combining (\ref{eq4.2}) and (\ref{eq4.1}), we have the desired result.
\end{proof}

For $R_1\in \mathcal{B}(n)$, let $\hat{B}_\lambda(t)=\left(\begin{array}{cc}I_n& 0_n \\-\hat{Q}(t)  & -\hat{R}(t)-\lambda P^{-1}R_1 \end{array}\right)$, let
$\hat{\gamma}_\lambda$ be the corresponding fundamental solutions.
With the above preparation, we have the following theorem.
\begin{thm}For $\nu\in \mathbb C$ such that $\d+\nu I_n$ is invertible,  we have
\begin{eqnarray}
&&{\det}\[\(-\(\d+\nu I_n\)^2+\hat{Q}(t)\(\d+\nu I_n\)+\hat{R}(t)\)\(-\(\d+\nu I_n\)^2\)^{-1}\]\nonumber  \\&&=e^{-\frac{T}{2}\int_0^T Tr (\hat{Q}) dt}\det( \bar{S}_d\hat{\gamma}_0(T)- e^{\nu T} I_{2n})\det(\bar{S}_d-e^{\nu T} I_{2n})^{-1}. \label{h.7}
\end{eqnarray}

\end{thm}
\begin{proof}
By the multiplicative property of conditional Fredholm determinant
\begin{eqnarray}
{\det}\[\(-J{d\over dt}-\hat{B}-\nu J\)\(-J{d\over dt}-\nu J \)^{-1}\]
&=&{\det}\[\(-J{d\over dt}-\hat{B}-\nu J\)\(-J\d+P_0\)^{-1}\]\nonumber\\ &&\cdot{\det}\[\(-J\d+P_0\)\(-J{d\over dt}-\nu J \)^{-1}\].\label{eq4.3}
\end{eqnarray}
 By the Hill-type formula for Hamiltonian system (\ref{hg}), we have that
\begin{eqnarray}
{\det}\[\(-J{d\over dt}-\hat{B}-\nu J\)\(-J\d+P_0\)^{-1}\]=C(\bar{S}_d)e^{-\frac{T}{2}\int_0^T Tr (\hat{Q}(t)) dt}e^{-n\nu T}\det\left( \bar{S}_d\hat{\gamma}_0(T)- e^{\nu T} I_{2n}\right)\label{eq4.4}
\end{eqnarray}
and
\begin{eqnarray}
{\det}\[\(-J\d+P_0\)\(-J{d\over dt}-\nu J \)^{-1}\]&=& \[ {\det}\[\(-J{d\over dt}-\nu J \)\(-J\d+P_0\)^{-1}\] \]^{-1} \nonumber\\
&=& C(\bar{S}_d)^{-1}e^{n\nu T}\det( \bar{S}_d- e^{\nu T} I_{2n})^{-1}.\label{eq4.5}
\end{eqnarray}
Substituting (\ref{eq4.5}) and (\ref{eq4.4}) into (\ref{eq4.3}), by Proposition \ref{thm5.1a} we have the result.
\end{proof}

We come back to the Lagrangian systems.  To simplify  the notation, let
\begin{eqnarray*}\mathcal{A}(\nu)=-\(\d+\nu\)\(P\(\d+\nu\)+Q\)+Q^T\(\d+\nu\)+R(t).\end{eqnarray*}

\begin{thm}\label{thm4.10} Under the condition (\ref{c1}), for any $\nu\in\mathbb{C}$ such that $\mathcal{A}(\nu)$ is invertible,
\begin{eqnarray}
\det\[(\mathcal{A}(\nu)+R_1)\mathcal{A}(\nu)^{-1}\] =\det(\bar{S}_d\gamma_1(T)-e^{\nu T}I_{2n})\cdot \det(\bar{S}_d\gamma_0(T)-e^{\nu T}I_{2n})^{-1}, \label{hillag}
\end{eqnarray}
where $\gamma_\lambda(t)$ is the fundamental solution of (\ref{h1}).
\end{thm}
\begin{proof} Please note that $F=R_1\mathcal{A}(\nu)^{-1}$ is a trace class operator, thus $\det(id+F)$ is the usual Fredholm determinant. Therefore
\bea \det(id+F)=\det(id+P^{-1}F P), \nonumber  \eea
hence
 \bea
 \det\left[(\mathcal{A}(\nu)+R_1)\mathcal{A}(\nu)^{-1}\right]= \det\left[P^{-1}(\mathcal{A}(\nu)+R_1)\mathcal{A}(\nu)^{-1}P\right]=\det\[\(P^{-1}(\mathcal{A}(\nu)+R_1)\)\(P^{-1}\mathcal{A}(\nu)\)^{-1}\]. \nonumber\label{h.8}
  \eea
Easy computation  shows that
\begin{eqnarray*}P^{-1}\mathcal{A}(\nu)=-\(\d+\nu\)^2+\hat{Q}\(\d+\nu\)+\hat{R},\end{eqnarray*} where
   $\hat{Q}=P^{-1}(Q^T-Q-\dot{P})$,  $\hat{R}=P^{-1}(R-\dot{Q})$. By the multiplicative property (\ref{eq2.2b}) of Fredholm determinant,
\bea
\det\[P^{-1}(\mathcal{A}(\nu)+R_1)\(P^{-1}\mathcal{A}(\nu)\)^{-1}\]\nonumber&=&\det\[P^{-1}(\mathcal{A}(\nu)+R_1)\(-\(\d+\nu\)^2\)^{-1}\]\\
&&\cdot
{\det\[\(P^{-1}\mathcal{A}(\nu)\)\cdot\(-\(\d+\nu\)^2\)^{-1}\]}^{-1}.   \label{h.9}
\eea
Substituting (\ref{h.7}) into (\ref{h.9}), we have
\bea
\det\left[(\mathcal{A}(\nu)+R_1)\mathcal{A}(\nu)^{-1}\right]=\det(\bar{S}_d\hat{\gamma}_1(T)-e^{\nu T}I_{2n})\cdot \det(\bar{S}_d\hat{\gamma}_0(T)-e^{\nu T}I_{2n})^{-1}. \label{h.10}
\eea

 To prove the theorem, we will make clear the relationship between $\hat{\gamma}_\lambda(T)$ with $\gamma_\lambda(T)$.
Let $\eta(t)=\left(\begin{array}{cc}P(t)& Q(t) \\0_n  & I_n \end{array}\right)$, then direct computation shows that
\bea \d(\eta(t)\hat{\gamma}_\lambda(t)\eta(0)^{-1})=JB_\lambda(t)\eta(t)\hat{\gamma}_\lambda(t)\eta(0)^{-1},  \label{h.4}\nonumber\eea
which implies $\gamma_\lambda(t)=\eta(t)\hat{\gamma}_\lambda(t)\eta(0)^{-1}$. Moreover, from (\ref{c1}), $\bar{S}_d\eta(T)=\eta(0)\bar{S}_d$, easy computation  shows that
\bea  \bar{S}_d\gamma_\lambda(T)=\eta(0)\bar{S}_d\hat{\gamma}_\lambda(T)\eta(0)^{-1}. \label{h.5}   \eea
It follows that $$\det (\bar{S}_d\gamma_\lambda(T)-e^{\nu T}I_{2n})=\det (\bar{S}_d\hat{\gamma}_\lambda(T)-e^{\nu T}I_{2n}).$$  Combining with (\ref{h.10}), we have the desired result.
\end{proof}

Obviously, by taking $\nu=0$ in Theorem \ref{thm4.10} we have Theorem \ref{thm1.2}.\vskip2mm

To get the trace formula,  let $\lambda R_1$ take place of $R_1$
in the Hill-type formula (\ref{hillag}), and we have
\bea \det(id+\lambda R_1\mathcal{A}(\nu)^{-1})=\det\left(\bar{S}_d\gamma_\lambda(T)-e^{\nu T}I_{2n}\right)\cdot\det\left(\bar{S}_d\gamma_0(T)-e^{\nu T} I_{2n}\right)^{-1}.\label{n3.1}  \eea
Almost the same as the proof of Theorem \ref{thm1.1}, the trace formula for Lagrangian system could be obtained  by taking Taylor expansion on the variable $\lambda$ and comparing
the coefficients of $\lambda^n$ on both sides of (\ref{n3.1}), and the proof will be omitted.  We have the trace formula for Lagrangian system,  for $m\in\mathbb{N}$,
\begin{eqnarray}\label{eq3.27}
Tr \Big(\big[R_1\mathcal{A}(\nu)^{-1}\big]^{m}\Big)=m\sum_{k=1}^m
\frac{(-1)^{m+k}}{k}\[\sum\limits_{j_1+\cdots+j_k=m}Tr(G_{j_1}\cdots
G_{j_k})\], \,
\end{eqnarray}
\emph{where  for Lagrangian system, we always denote  $D=\left(\begin{array}{cc} 0_n & 0_n\\
0_n& -R_1\end{array}\right)$}, $G_k$ is defined in Theorem \ref{thm1.1a}.

Since $\mathcal{A}(\nu)^{-1}$ is a trace class operator, let $\{\lambda_i\}$ be  the nonzero eigenvalues of $\mathcal{A}(\nu)y+\lambda R_1 y=0$, then for positive integers $m$, \bea \label{eq3.28}\sum\limits_{j} \frac{1}{\lambda_j^m}=(-1)^m \cdot Tr \left[\left(R_1\mathcal{A}(\nu)^{-1}\right)^{m}\right]. \eea
Combining (\ref{eq3.27}) and (\ref{eq3.28}) we prove Theorem \ref{thm1.1a}.

 Especially,
\begin{eqnarray}\label{eq3.49a}
Tr[R_1\mathcal{A}(\nu)^{-1}]=Tr\(J\int_0^T\gamma_{0}^T(t)D(t)
\gamma_{0}(t)dt\cdot M(M-e^{\nu T} I_{2n})^{-1}\).
\end{eqnarray}

Comparing with the Trace formula in Hamiltonian systems, we have
\begin{cor}\label{cor4.1a} For positive integers $m$,
\begin{eqnarray} (-1)^{m}\cdot Tr \left[\left(R_1\mathcal{A}(\nu)^{-1}\right)^{m}\right]=Tr \left[\left(D(A-B_0-\nu J)^{-1}\right)^{m}\right]. \label{fc4.1}
\end{eqnarray}
where $D=\left(\begin{array}{cc}0_n & 0_n \\ 0_n & -R_1\end{array}\right)$, $B_0$ is defined in (\ref{b2}) and $A=-J\d$ with $\bS_d-$boundary condition.
\end{cor}

Obviously $\mathcal{A}(\nu)+\lambda R_1$ is degenerate if only if $-J\d-\nu J-B_0-\lambda D$ is degenerate, moreover, we have
\begin{prop}\label{cor2.5}
Let  $\nu\in\mathbb C$, such that $\mathcal{A}(\nu)$ is invertible.  Then  $-{1\over\lambda_0}$ is an eigenvalue of $R_1\mathcal{A}(\nu)^{-1}$ of algebraic multiplicity $k$  if and only if $1\over\lambda_0$ is an eigenvalue of $ D(-J\d-\nu J-B_0)^{-1} $  of algebraic multiplicity $k$.
\end{prop}

\begin{rem}\label{cor.lh}
\begin{itemize}
\item[1.] For $m\geq 2$, notice that both  $\left(R_1\mathcal{A}(\nu)^{-1}\right)^m$ and $\left(D(A-B-\nu J)^{-1}\right)^m$ are trace class, and hence by Proposition \ref{cor2.5} we can get the trace formula for Lagrangian system from that of Hamiltonian system directly.
\item[2.] For $m=1$, since the operator $D(A-B-\nu J)^{-1}$ is not trace class operator, but a Hilbert-Schmidt operator with trace finite condition. Therefore,
 \begin{eqnarray}
 Tr(D(A-B-\nu J)^{-1})=\lim\limits_{N\to \infty} Tr \left[P_ND(A-B-\nu J)^{-1}P_N\right].
 \end{eqnarray}
 For a general Hamiltonian system, we don't know whether $Tr(D(A-B-\nu J)^{-1})=\sum\limits_{j}{1\over \lambda_j}$ true or not. It follows that, the trace formula (\ref{eq3.49a}) can not be obtained by the trace formula from Hamiltonian system.
\end{itemize}
\end{rem}

To prove Proposition \ref{cor2.5}, we need the following lemma, which is of interest itself.
\begin{lem}\label{lem3.6}
Let $F$ be a Hilbert-Schmidt operator with trace finite condition, and $1\over \lambda_0$ is its nonzero eigenvalue. Then $\lambda_0$ is a zero point of $\det (id-\lambda F)$ of degree $k$ if and only if $1\over\lambda_0$ is an eigenvalue of $F$ of algebraic multiplicity $k$.
\end{lem}
\begin{proof}
Since $F$ is a Hilbert-Schmidt operator, so $\sigma_1=\{{1\over \lambda_0}\}$ and $\sigma_2=\sigma(F)\setminus \sigma_1$ are two disjoint closed subsets of the spectral of $F$. By Riesz Decomposition Theorem for operators, let
\begin{eqnarray*}
P_1={1\over 2\pi i}\int_{\Gamma}(\lambda-F)^{-1} d\lambda,
\end{eqnarray*}
where $\Gamma$ is a contour in the resolvent set of $F$ such that $\sigma_1$ in its interior and $\sigma_2$ in its exterior. Then $P_1$ is its Riesz projection, and let $P_2=id-P_1$. Since $1\over\lambda_0$ is a nonzero eigenvalue, then $P_1$ is a finite projection, and $P_1F=FP_1$. Now, let $F_1=FP_1$ and $F_2=FP_2$, then $F_1F_2=0$. By the multiplicative property of conditional Fredholm determinant,
\begin{eqnarray*}
\det(id-\lambda F)=\det(id-\lambda F_1-\lambda F_2-\lambda^2 F_1 F_2)=\det(id-\lambda F_1)\det(id-\lambda F_2).
\end{eqnarray*}
Since $1\over \lambda_0$ is not in the spectrum of $F_2$, hence $1\over \lambda_0$ is not zero point of $\det(id-\lambda F_2)$; moreover, it is not hard to see that $\det (id-\lambda F_1)= \left(1-{\lambda\over \lambda_0}\right)^k$ where $k$ is the algebraic multiplicity of the eigenvalue $1\over \lambda_0$ of $F$. The proof is complete.

\end{proof}

\noindent\emph{Proof of Proposition \ref{cor2.5}.}  By (\ref{n3.1}) and Lemma \ref{lem3.6}, $-{1\over \lambda_0}$ is an eigenvalue of $R_1\mathcal{A}(\nu)^{-1}$ of algebraic multiplicity $k$ if and only if it is a zero point the analytic function $\det(\bar{S}_d\gamma_\lambda(T)-e^{\nu T}I_{2n})$ of degree $k$. On the other hand, by (\ref{thf1.1.2}) and the multiplicative property, for $B_\lambda$ defined in (\ref{b2})  we have
 \begin{eqnarray*}
 \det\(id-\lambda D\(-J\d-\nu J-B_0\)^{-1}\)&=&\det\[\(-J\d-\nu J-B_\lambda\)\(-J\d+P_0\)^{-1}\] \\
                                                                     && \cdot \det\[\(-J\d+P_0\)\(-J\d-\nu J-B_0\)^{-1}\]\\
                                                                     &=&(\bS_d \gamma_\lambda(T)-e^{\nu T} I_{2n})(\bS_d \gamma_0(T)-e^{\nu T}I_{2n}).
 \end{eqnarray*}
 Again, by Lemma \ref{lem3.6}, $1/\lambda_0$ is an eigenvalue of $D(-J\d-\nu J-B_0)^{-1} $  of algebraic multiplicity $k$ if and only if it is also a zero point $\det(\bar{S}_d{\gamma}_\lambda(T)-e^{\nu T}I_{2n})$ of degree $k$. The desired result is proved. \hfill$\Box$

\begin{exam}

We will compute  the simplest case, that is
 $\mathcal{A}(\nu)=-(\d+\nu)^2$, $R_1=-R$.
  Recall that $K_n=\left(\begin{array}{cc} I_n & 0_n \\ 0_n & 0_n\end{array}\right)$, $D=\left(\begin{array}{cc} 0_n & 0_n \\ 0_n & R\end{array}\right)$.
Recall that
  $\gamma_{0}(t)$ satisfied
 $\dot{\gamma}_0(t)=JK_n\gamma_{0}(t)$ with $\gamma_{0}(0)=I_{2n}$. Direct computation shows that
 $ \gamma_0(t)=\left(\begin{array}{cc}I_n & 0_n\\ tI_n&I_n \end{array}\right),$ and obviously $\gamma_{0}(t)^{-1}=\left(\begin{array}{cc}I_n & 0_n\\ -tI_n&I_n \end{array}\right)$. Therefore,
 $$ J\hat{D}= \gamma_{0}(t)^{-1}JD\gamma_{0}(t)=\left(\begin{array}{cc} -tR & -R\\ t^2R&tR \end{array}\right),    $$
thus
$$  J\int_0^T \hat{D}dt= \left(\begin{array}{cc}  -\int_0^TtRdt &  -\int_0^TRdt\\  \int_0^Tt^2Rdt&  \int_0^TtRdt \end{array}\right).$$
Let $\bS^T$ be the transposition of $\bS$, then $\bS^T=\bS^{-1}$. For $\omega=e^{\nu T}$
 $$ M(M-\omega)^{-1}=\left(\begin{array}{cc}(I_n-\omega \bS^T)^{-1} & 0_n\\ 0_n & (I_n-\omega \bS^T)^{-1}\end{array}\right)\left(\begin{array}{cc} I_n & 0_n\\ -\omega T \bS^T(I_n-\omega \bS^T)^{-1}&I_n \end{array}\right), $$
 and
 \bea  G_1=\left(\begin{array}{cc}  -\int_0^TtRdt &  -\int_0^TRdt\\  \int_0^Tt^2Rdt& \int_0^TtRdt \end{array}\right)\cdot\left(\begin{array}{cc}(I_n-\omega \bS^T)^{-1} &  0_n\\ 0_n & (I_n-\omega \bS^T)^{-1}\end{array}\right)\left(\begin{array}{cc} I_n & 0_n\\ -\omega T \bS^T(I_n-\omega \bS^T)^{-1}&I_n \end{array}\right). \eea
 Thus  $$ Tr(G_1)=\omega Tr\left(T\int_0^T Rdt\cdot \bS^T(I_n-\omega \bS^T)^{-2}\right). $$
To simplify  the notation, we denote by $$ R_{ave}=\frac{1}{T}\int_0^TR(t)dt, $$
 which is a constant matrix. Then
\bea   Tr (R\mathcal{A}(\nu)^{-1})=-\omega T^2\cdot Tr(R_{ave}\cdot
\bS(\bS-\omega)^{-2}). \label{aa.1}\eea

Please note that by take derivative with respect
to $\nu$  on both sides of (\ref{aa.1}), we get
\bea   Tr \left(R\mathcal{A}(\nu)^{-2}\right)=\frac{\omega T^4}{6}Tr\left(R_{ave} \bS(\bS^2+4\omega \bS+\omega^2)(\bS-\omega)^{-4}\right). \label{aa.2}\eea
\end{exam}

\begin{rem}\label{remark}
In \cite{K1}, Krein also
consider the boundary value problem \bea y'' +\lambda R(t)y=0,\,\
y(0)+y(T)=y'(0)+y'(T)=0,\label{k1} \eea where $R(t)\in \mathcal{B}(n)$.  Let $\lambda_j$,
$j\in\mathbb{Z}$ or $\mathbb{N}$ (assume $\lambda_j\leq
\lambda_{j+1}$ ), be the  eigenvalues of boundary value
problem (\ref{k1}), that means the system \bea y'' +\lambda_j
R(t)y=0,\,\ y(0)+y(T)=y'(0)+y'(T)=0,\label{k1.1} \eea has a
nontrivial solution. Each $\lambda_j$ appears as many times as its
multiplicity. To state Krein's work, set
 \bea  X(t)=\int_0^t(R(s)-R_{ave})ds+C, \label{k2} \eea where $C$
is a constant matrix which is chosen such that $X_{ave}=0$. Krein
proved \cite{K1} \bea
\sum\frac{1}{\lambda_j}=\frac{T}{4}\int_0^TTr(R(t))dt, \label{k3}
\eea and \bea
\sum\frac{1}{\lambda_j^2}=\frac{T}{2}\int_0^TTr(X^2(t))dt+\frac{T^2}{48}Tr\[\(\int_0^TR(t)dt\)^2\].
\label{k4} \eea
Please note that (\ref{aa.1}) is a generalization of (\ref{k3}). Please note that, in the formula  (\ref{lt.1a}),  the expression of  $\frac{1}{\lambda_j^2}$ is different from  (\ref{k4}).
 The precise  generalization with the same form as Krein's formula  will be given in the forthcoming paper.
\end{rem}

\section{Application}\label{sec 4}

The Maslov-type index is a very useful tool in studying the
multiplicity and stability of  periodic solution in Hamiltonian
systems \cite{Lon3},\cite{Lon4}. It is well-known that the relative Morse index for linear Hamiltonian system equals to the Maslov-type index for the corresponding fundamental solutions. It will be seen that, by  the trace formula, we could estimate the relative Morse index, and therefore the trace formula could be used to
judge the linear stability via the Maslov-type index. For reader's convenience, we  review the   relative Morse index and stability criteria via Maslov-type index in \S \ref{subsec4.1}, details could be found in \cite{HS},\cite{Lon4}. The estimation of relative Morse index by the trace of operator is given in \S \ref{subsec4.2}, some new criteria for the  stability is given in \S \ref{subsec4.3}, at \S \ref{subsec4.4}, we give some estimation of Morse index for
Sturm-Liouville systems.

In the whole of this section, $\nu$ will be assumed to be an imaginary number.

\subsection{Brief  review of the relative Morse index, spectral flow and stability criteria via  Maslov-type index }\label{subsec4.1}

As we have reviewed in \cite{HW}, let $\tilde{A}$,  $\tilde{B}$ be bounded self-adjoint operators on Hilbert space $E$, $\tilde{A}$ is a Fredholm
operator and $\tilde{B}$ is compact, then the relative Morse index $I(\tilde{A}, \tilde{A}-\tilde{B})$ is defined by
\begin{eqnarray*}
I(\tilde{A},\tilde{A}-\tilde{B})=\dim (E_{-}(\tilde{A}-\tilde{B}),E_{-}(\tilde{A})). \label{5a.1}
\end{eqnarray*}
where $E_-(\tilde{A})$ and $E_-(\tilde{A}-\tilde{B})$ are respectively the subspaces on which $\tilde{A}$ and $\tilde{A}-\tilde{B}$ is negative definite, and
\begin{eqnarray*}
\dim (E_1,E_2)=\dim(E_1\cap E_2^\perp)-\dim(E_2\cap E_1^\perp). \label{5a.2}
\end{eqnarray*}
For the Hamiltonian system, let $A=-J\frac{d}{dt}$. Denote by $\mathcal{M}$ the linear space  generated by the eigenvectors of $A$. The
$1/2$ inner product on $\cM$ is defined by, \begin{eqnarray*} \langle
x,y\rangle_{1/2}=\langle(|A|+id)x,y\rangle,\quad \mathrm{for}\ x,\,y\in \mathcal{M}, \label{5a.4}
\end{eqnarray*} where $\langle\cdot,\cdot\rangle$ is the inner product in $E$.  Denote by $\tilde{E}$ the Hilbert space completed by $\mathcal{M}$ under the $1/2$ norm. Let $\tilde{A}=(id+|A|)^{-1}A$, $\tilde{B_j}=(id+|A|)^{-1}B_j$ (j=0,1), then
both of them are self-adjoint operators on $\tilde{E}$. Define
\begin{eqnarray*}
I(A-B_0,A-B_1)=I(\tilde{A}-\tilde{B_0},\tilde{A}-\tilde{B_1}). \label{5a.5}
\end{eqnarray*}
The relationship between the conditional Fredholm determinant and the relative Morse index had been given in \cite{HW}.

On the other hand, the relative Morse index could be defined by spectral flow \cite{HS}. As is well known, spectral flow was introduced by Atiyah, Patodi and
Singer \cite{APS} in their study of index theory on manifolds with
boundary. It is a very useful tool to understand the relative Morse index.
Let $\{A(\theta),\theta\in[0,1]\}$ be a continuous path of
self-adjoint Fredholm operators on a Hilbert space $\mathcal{H}$. Roughly
speaking, the spectral flow of path $\{A(\theta),\theta\in[0,1]\}$
counts the net change in the number of negative eigenvalues of
$A(\theta)$ as $\theta$ goes from $0$ to $1$, where the enumeration
follows from the rule that each negative eigenvalue crossing to the
positive axis contributes $+1$ and each positive eigenvalue crossing
to the negative axis contributes $-1$, and for each crossing, the
multiplicity of eigenvalue is counted. More precisely, as shown in
\cite{APS},  let
$$\wp=\bigcup_{\theta\in[0,1]} \sigma (A({\theta})),
$$ where $\sigma(A(\theta))$ is the spectrum for $A(\theta)$, then $\wp$ is a closed subset of the $(\theta,\lambda)$-plane. The spectral flow $Sf(\{A({\theta})\})$
 is defined to be the intersection number of $\wp$ with the line $\lambda=-\epsilon$ with respect
 to the usual orientation for some small positive $\epsilon$.
Obviously,
$Sf(\{A(\theta)\})=Sf(\{A(\theta)+\epsilon id\})$ if $id$ is the
identity operator on $\mathcal{H}$, and $0\leq\epsilon\leq\epsilon_0$ for some sufficiently small positive number
$\epsilon_0$.

We come back to the Hamiltonian systems, suppose $B(s,t)\in C([0,1]\times[0,T], S(2n))$. For $s\in[0,1]$,  let $B_s\in\mathcal{B}(2n)$. For such two operators $A-B_0$ and $A-B_1$, we can define the relative Morse index via spectral flow.
In fact, by \cite{HS}, we have,
\begin{eqnarray*}
I(A-B_0,A-B_1)=-Sf(\{A-B(s), s\in[0,1]\}). \label{5a.8}
\end{eqnarray*}
For $B_0$, $B_1$, $B_2$, then \begin{eqnarray*}
I(A-B_0,A-B_1)+I(A-B_1,A-B_2)=I(A-B_0,A-B_2). \label{a5.10}  \end{eqnarray*}
Let $D=B_1-B_0$, and we can simply let $B(s)=B_0+sD$. The next proposition is obvious from the
definition of spectral flow.
\begin{prop}\label{prop5a.1} Let $\kappa=\{s_0\in[0,1], \ker(A-B(s_0))\neq0\}$,
\begin{eqnarray*} I(A-B_0, A-B_1)\leq \sum_{s_0\in\kappa}  \mathrm{dim}\,\ker(A-B(s_0)).  \label{5a.12}  \end{eqnarray*}

\end{prop}
It is not hard to see that, if $D>0$, then $I(A-B,A-B-D)\geq0$. By careful analysis\cite{HS},   the crossing form
\begin{eqnarray} I(A-B, A-B-D)=\sum_{s_0\in\kappa\cap[0,1)} \dim\ker(A-B(s_0)).  \label{5a.14}  \end{eqnarray}
Similarly
\begin{eqnarray} I(A-B, A-B+D)=-\sum_{s_0\in\kappa\cap(0,1]} \dim\ker(A-B(s_0)).  \label{5a.16}  \end{eqnarray}
Thus we have
\begin{cor}\label{indexno} Suppose $D_1\leq D\leq D_2$, then
\begin{eqnarray}I(A-B,A-B-D_1)\leq I(A-B,A-B-D)\leq I(A-B,A-B-D_2).\label{indexn}
\end{eqnarray}
\end{cor}

To get the stability criteria, we consider the following Hamiltonian system,
\bea \dot{z}(t)&=& JB(t)z(t)\nonumber \label{3b.3}\\
  z(0)&=& \omega S z(T). \nonumber\label{3b.4} \eea
Denote $A_\omega$, $B_\omega$ be the operators
corresponding to $A,B$ respectively under the $\omega
S$-boundary condition, then $A_\omega$ is a self-adjoint operator
with the domain $D_{\omega S}$.  Since $\nu$ is an imaginary number,  $e^{\nu t}$ is a unitary operator on $E$ and $e^{\nu t}D_S=D_{\omega S}$.   Simple calculations  show that $ e^{-\nu t}A_\omega e^{\nu t}=A-\nu J$.
 Thus we have
 \bea I(A_\omega, A_\omega-B_\omega)=I(A-\nu J, A-\nu J-B).  \label{5a.17} \eea

To judge the stability, we use the Maslov-type index $i_\omega(\gamma)$, which is essentially same as  the relative Morse index \cite{Lon4}. Roughly speaking, for a continuous path $\gamma(t)\in\Sp(2n)$, $\omega\in \mathbb{U}$, the Maslov-type index $i_\omega(\gamma)$ is defined by the intersection number of $\gamma$ and $\Sp_\omega^0(2n)=\{M\in\Sp(2n)\,|\, \det(M-\omega I_{2n})=0\}$.  Details could be found in \cite{Lon2},\cite{Lon4}, some brief review could be found in \cite{HS1}. For simplicity, we assume  $S=I_{2n}$.

\begin{prop}\label{prop5a.2}
Suppose $S=I_{2n}$, then, for  imaginary number $\nu$  such that   $\omega=e^{\nu T}\in \mathbb U\setminus \{1\}$, we have
\begin{eqnarray*} I(A, A-B)=i_1(\gamma)+n.  \label{5a.18}  \end{eqnarray*}
and
\begin{eqnarray} I(A-\nu J, A-\nu J-B)=i_\omega(\gamma).  \label{5a.20}  \end{eqnarray}
\end{prop}

\begin{proof} 
From \cite[Theorem 2.5 and Lemma 4.5]{HS}, we have
\begin{eqnarray*} I(A, A-B)=i_1(\gamma)+n,  \label{ad5a.18}  \end{eqnarray*} and
\begin{eqnarray*} I(A_\omega, A_\omega-B_\omega)=i_\omega(\gamma),  \omega\in \mathbb U\setminus\{1\}.  \label{ad5a.20} \end{eqnarray*}
By (\ref{5a.17}), we have (\ref{5a.20}).
 \end{proof}

 We will continue to review the stability criteria by the Maslov-type index.  Details for the stability
criteria and the Maslov-type index are given in  \cite{Lon4}.
For $\omega\in \mathbb U$, the unit circle, $\omega=e^{i\theta_0}$ with $\theta_0\in[-\pi,\pi]$, let $\mathbb U_\omega=\{e^{i\theta}, \theta\in[-|\theta_0|,|\theta_0|]\}$, denote by $e_\omega(M)$ the total algebraic multiplicities of all eigenvalues of $M$ in $\mathbb U_\omega$. We also simply denote by $e(M)$ the total algebraic multiplicities of
all eigenvalues of $M$ on $\mathbb U$. \emph{Obviously, for $M=\gamma(T)$ if $e(M)=2n$ then $M$ is spectral stable.}

For a bounded variation  function $g(w)$ defined on some closed interval $[a,b]$, we
define its variation by
\begin{eqnarray*}
 var(g(w),[a,b])=\sup
 \Big\{\sum_{j=0}^{k-1}|g(w_{j+1})-g(w_j)|\, \Big|\, a=w_0<.....<w_k=b \,\,\mathrm {is} \,\mathrm {any}\,\mathrm {partition}\Big\}. \label{n1.3.5}
\end{eqnarray*}
Notice that $i_{e^{\theta\sqrt{-1}}}$ is a bounded variation function on $[0,\theta_0]$. And the next proposition can be proved easily by the property of Masolv-type index (readers are referred to \cite{Lon4} or \cite{HS}).
\begin{prop}\label{ellipad}  Let $\gamma$ be an arbitrary path in $\Sp(2n)$ connecting $I_{2n}$ to  $M$,
\begin{eqnarray} e_\omega(M)/2\geq
 var(i_{e^{\theta\sqrt{-1}}}(\gamma),\theta\in[0,\theta_0]).
 \label{5a.24} \end{eqnarray}

\end{prop}

\begin{cor}\label{ellip} With the notations as above ,
\begin{eqnarray*} e(M)/2\emph{}\geq
 var(i_{e^{\theta\sqrt{-1}}}(\gamma),\theta\in[0,\pi]).
 \label{5a.22} \end{eqnarray*}

\end{cor}
Obviously, for $\omega\neq\{1,-1\} $
\begin{eqnarray} e(M)/2\emph{}\geq
 |i_{-1}(\gamma)-i_\omega(\gamma)|+ |i_{1}(\gamma)-i_\omega(\gamma)|.
\label{sta1}\end{eqnarray}
Especially,
\begin{eqnarray} e_\omega(M)/2\emph{}\geq
 |i_\omega(\gamma)-i_1(\gamma)|.
 \label{5a.25} \end{eqnarray}
\begin{eqnarray} e(M)/2\emph{}\geq
 |i_{-1}(\gamma)-i_1(\gamma)|.
 \label{5a.26} \end{eqnarray}

\begin{rem}\label{rem4.1}
All the above results, for that the relative Morse index equals to Maslov-type index and for the stability criteria,   could be proved
for any $S$ boundary condition with $S\in\Sp(2n)\cap O(2n)$, and details could be found in \cite{HS}.
\end{rem}

\subsection{Estimate relative Morse index by trace formula }\label{subsec4.2}
In this subsection, we will give the application of the trace formula on the estimation of the non-degeneracy. Moreover,  we will  estimate Maslov-type index by using the trace formula. Suppose $A-\nu J-B$ is non-degenerate,  we will estimate the relative Morse index $I(A-\nu J-B,A-\nu J-B-D)$. Firstly assume $D>0$, thus $I(A-\nu J-B,A-\nu J-B-D)\geq 0$.

\begin{lem}\label{lem5b.1}
Suppose $D>0$, $\nu$ is an imaginary number, then all the eigenvalues of $D(A-\nu J-B)^{-1}$ are  real.
\end{lem}

\begin{proof} Let $D^{1/2}$ be the unique positive operator such that $D^{1/2}D^{1/2}=D$, then $D(A-\nu J-B)^{-1}$ is similar to $D^{1/2}(A-\nu J-B)^{-1}D^{1/2}$,
which is a self-adjoint compact operator. Hence
\begin{eqnarray*}
\sigma(D(A-\nu J-B)^{-1})=\sigma(D^{1/2}(A-\nu J-B)^{-1}D^{1/2})\subset \mathbb R.\end{eqnarray*}
\end{proof}

Let $1\over\lambda_j$ be the eigenvalues of $D(A-\nu J-B)^{-1}$. By Lemma \ref{lem5b.1}, $\lambda_j\in\mathbb{R}$, we can make the order such that
\bea   \cdots\leq\lambda_2^-\leq\lambda_1^-<0<\lambda_1^+\leq\lambda_2^+\leq\cdots .\label{4.a1} \nonumber\eea
Moreover, we have
\begin{lem}\label{lem5b.1.1}
Suppose $D>0$, then $\lim\limits_{j\rightarrow\infty}\lambda_j^+=+\infty$ and  $\lim\limits_{j\rightarrow\infty}\lambda_j^-=-\infty$.
\end{lem}

\begin{proof} We will use the contradiction argument. Suppose there is $ \lambda^+_0$ such that, for each $j\in\mathbb N$, $\lambda_j^+<\lambda^+_0$. We claim that
$$\sigma(A-\nu J-B-\lambda^+_0D)\subset(-\infty,0].$$ In fact, notice that $$\sigma(A-\nu J-B-\lambda^+_0D)\subset(-\infty,0]$$ if and only if $$\sigma(D^{-{1\over 2}}(A-\nu J-B)D^{-{1\over 2}}-\lambda^+_0)\subset(-\infty,0].$$ Moreover, it is easy to see that $$\sigma(D^{-{1\over 2}}(A-\nu J-B)D^{-{1\over 2}})=\{\lambda _j\},$$ and hence $$\sigma(D^{-{1\over 2}}(A-\nu J-B)D^{-{1\over 2}}-\lambda^+_0)\subset(-\infty,0].$$ Now, notice that $A$ is an unbounded operator $\pm\infty$ is the limitation of its eigenvalues, and $\nu J-B-\lambda_0^+$ is a bounded operator. By the spectral theory for unbounded operator with perturbation by  bounded operator, we have that
\begin{eqnarray*}
\sigma (A)\subset \Big\{\lambda\,\Big|\,|\lambda-\lambda_0|\leq\|\nu J-B-\lambda^+_0\|, \text{for some }\lambda_0\in\sigma(A-\nu J-B-\lambda^+_0 D)\Big\}.
\end{eqnarray*}
This is a contradiction. The other part of the lemma can be proved similarly.
\end{proof}

Recall the formula (\ref{0.1.4}) that,
\begin{eqnarray} Tr \left[\left(D(A-\nu J-B)^{-1}\right)^m\right]=\sum_{j=1}^\infty \frac{1}{\lambda_j^m}, \,\ m\geq2. \label{4.0.1.3}
  \end{eqnarray}

\begin{prop}\label{prop5b.1} Suppose $D>0$, we have that, for $\forall k\in\mathbb{N}$
\begin{eqnarray*} I(A-\nu J-B,A-\nu J-B-D)+\mathrm{dim}\ker(A-\nu J-B-D)< Tr((D(A-\nu J-B)^{-1})^{2k}).   \label{5b.1}  \end{eqnarray*}
\end{prop}

\begin{proof}

From Lemma \ref{lem5b.1}, $\lambda_j$ are real numbers, and hence $\lambda_j^{2k}>0$. By Lemma \ref{lem5b.1.1} and (\ref{4.0.1.3}), we have
 \begin{eqnarray*}  Tr((D(A-\nu J-B)^{-1})^{2k}>\sum_{|\lambda_j|\leq 1}\frac{1}{\lambda_j^{2k}},\,\ \forall k\in\mathbb{N} \label{5b.1.1}.  \end{eqnarray*}
Obviously, $\sum\limits_{|\lambda_j|\leq 1}\frac{1}{\lambda_j^{2k}}$ is no less than the  total multiplicity
of eigenvalues with $|\lambda_j|\leq1$.  Please note that $\lambda_j\in D(A-\nu J-B)^{-1}$ if and only if $\ker(A-\nu J-B-\lambda_j D)$ is degenerate. Moreover, the multiplicity of the eigenvalue for $D(A-\nu J-B)^{-1}$ at $\lambda_j$ is equal to $\mathrm{dim}\ker(A-\nu J-B-\lambda_j D)$. By Proposition \ref{prop5a.1} and (\ref{5a.14}), the proposition is proved.
\end{proof}

Similar to Proposition \ref{prop5b.1}, we have the following proposition.

\begin{prop}\label{prop5b.2} Suppose $D>0$, then
\begin{eqnarray} - Tr((D(A-\nu J-B)^{-1})^{2k})<I(A-\nu J-B,A-\nu J-B+D)\leq 0,  \,\ \forall k\in\mathbb{N}\label{5b.2}.\nonumber  \end{eqnarray}
\end{prop}

We have the following corollary.
\begin{cor}\label{cor5b.1a}
Suppose $D>0$, if for some $k\in\mathbb{N}$, $Tr((D(A-\nu J-B)^{-1})^{2k})\leq 1$,
then \begin{eqnarray*}&&I(A-\nu J-B,A-\nu J-B+D)\\&&=I(A-\nu J-B,A-\nu
J-B-D)+\dim\ker(A-\nu J-B-D)=0.\end{eqnarray*}
\end{cor}
Now we can give the estimation on the upper bound that preserves the non-degeneracy.
\begin{thm}\label{thm4.2} Suppose $A-B-\nu J$ is non-degenerate. Suppose that there are $D_1,D_2\in \mathcal{B}(2n)$ such that $D_1< D < D_2$, with $D_1<0$, $D_2>0$, if there exists  $k\in 2\mathbb N$, such that  $Tr \big((D_j(A-B-\nu J)^{-1})^k\big)\leq1$
for $j=1,2$, then $A-B-D-\nu J$ is non-degenerate.
\end{thm}
\begin{proof}
By the condition $Tr \big((D_j(A-B-\nu J)^{-1})^{2k}\big)\leq1$, for
$j=1,2$, applying   Corollary \ref{cor5b.1a}, we have that, for any $s\in[0,1]$, \bea &&I(A-\nu J-B,A-\nu
J-B-sD_1)\nonumber
\\&&=I(A-\nu J-B,A-\nu J-B-sD_2)+\dim\ker(A-\nu J-B-sD_2)=0.
\label{aaa} \eea Next, we will prove the result by
contradiction argument. Assume that $A-\nu J-B-D$ is degenerate.  Now, let
$$ s_0=\inf\{s\in[0,1], \dim\ker(A-\nu J-B-sD)\neq0\}.$$ Notice that $A-\nu J-B$ is non-degenerate, thus $s_0>0$.
 From the spectral theory of self-adjoint operators \cite
{Ka}, the eigenvalues of $A-\nu J-B-sD$ can be considered as a  smooth function on $s$.
Denote the eigenvalue functions by $\lambda_j(s)$. Since $A-B-\nu J-s_0D$ is degenerate, there is some $\lambda_j(s_0)=0$. We may assume that
$\lambda_j(s_0)=0$ for $j=1,...,m$. By the definition of $s_0$,
$\lambda_j(s)\neq0$ on $[0,s_0)$ for $j=1,...,m$. Without loss of
generality, assume $\lambda_j(s)>0$ on $[0,s_0)$ for $j=1,...,m_1$
and $\lambda_j(s)<0$ on $[0,s_0)$ for $j=m_1+1,...,m$, where $m_1$
can take value $0$ or $m$.

Firstly, if $m_1>0$, by the property of
relative morse index, we have $$I(A-\nu J-B,A-\nu
J-B-s_0D_2)=I(A-\nu J-B,A-\nu J-B-s_0D)+I(A-\nu J-B-s_0D,A-\nu
J-B-s_0D_2), $$ and $I(A-\nu J-B,A-\nu J-B-s_0D)=m_1-m$ by the
definition of $s_0$. On the other hand, since $D_2>0$ form
(\ref{5a.14}), \bea&& I(A-\nu J-B,A-\nu J-B-s_0D_2)\nonumber\\ &=& m_1-m+I(A-\nu
J-B-s_0D,A-\nu J-B-s_0D-s_0(D_2-D)) \nonumber \\ &\geq&
m_1-m+\dim\ker(A-\nu J-B-s_0D)=m_1>0, \nonumber \eea which  contradicts to
(\ref{aaa}).

Next, if $m_1=0$, noting that $D_1<D$, by the property
of spectral flow,  $I(A-\nu J-B,A-\nu J-B-s_0D)=-m$. By some similar discussion as above, we get \bea&& I(A-\nu
J-B,A-\nu J-B-s_0D_1)\nonumber\\ &=&-m+ I(A-\nu J-B-s_0D,A-\nu J-B-s_0D-s_0(D_1-D))
\nonumber \\ &\leq&-m<0, \nonumber \eea
which also contradicts to (\ref{aaa}). The proof is complete.
\end{proof}

Next, we are going to give the estimation of the relative Morse index by the trace formula.
 \begin{thm}\label{thm4.3}
Suppose $A-B-\nu J$ is non-degenerate and $D_1\leq D\leq D_2$, where
$D_1<0$, $D_2>0$. Let
$$m^-=\inf\{[Tr ((D_1(A-B-\nu J)^{-1})^k)], k\in 2\mathbb N\} \quad \text{and}\quad
m^+=\inf\{[Tr ((D_2(A-B-\nu J)^{-1})^k)], k\in 2\mathbb N\}, $$ then
\begin{eqnarray*} -m^-\leq I(A-B-\nu J,A-B-D-\nu J)\leq m^+. \label{th1.3f}\end{eqnarray*}
\end{thm}
\begin{proof}
Firstly, we will prove that
\begin{eqnarray*}I(A-B-\nu J, A-B-D_2-\nu J)\leq m^+.\end{eqnarray*}
Infact, by Proposition \ref{prop5b.1}, we have that, for any $k\in 2\mathbb N$,
\begin{eqnarray*}
I(A-B-\nu J, A-B-D_2-\nu J)< Tr((D(A-\nu J-B)^{-1})^{2k}).
\end{eqnarray*}
It follows that
\begin{eqnarray*}
I(A-B-\nu J, A-B-D_2-\nu J)\leq m^+.
\end{eqnarray*}
By Proposition \ref{prop5b.2} and some similar reasoning, we have \begin{eqnarray*}I(A-B-\nu
J, A-B-D-\nu J)\geq-m^-.\end{eqnarray*}
Since $D_1\leq D\leq D_2$, we get the result by (\ref{indexn}).
\end{proof}

Motivated by Krein's work \cite{K1}, we consider the symmetric case, that is, $D(t)=D(T-t)$. Suppose first that $D$ is real and invertible. Then
\begin{eqnarray}\label{5b.31}
\sigma(D(A+\nu J)^{-1})=\{\bar\lambda\,|\,\lambda\in\sigma(D(A-\nu J)^{-1})\}.
\end{eqnarray}
In fact,  a nonzero  $\lambda\in \sigma(D(A-\nu J)^{-1})=\sigma((A-\nu J)^{-1}D)$ if and only if there is $x\not=0$, such that
\begin{eqnarray*}
(A-\nu J)^{-1}Dx=\lambda x,
\end{eqnarray*}
if and only if
\begin{eqnarray*}
(A+\nu J)\bar x=\bar\lambda^{-1} D \bar x,
\end{eqnarray*}
if and only if $\bar\lambda \in\sigma(D(A+\nu J)^{-1})$. Therefore, (\ref{5b.31}) holds true. Now,
suppose  $D>0$. Then, $\sigma(D(A-\nu J)^{-1})\subset \mathbb R$, and hence $\sigma(D(A-\nu J)^{-1})=\sigma(D(A+\nu J)^{-1})$.
If moreover $D(t)=D(T-t)$, then, by some direct computation,   $x(t)\in\ker (A-\nu J-\lambda D)$ if and only if $x(T-t)\in \ker (A+\nu J+\lambda D)$. We summarize the above reasoning as the following lemma.
\begin{lem}\label{lem5b.2}  Suppose $D>0$  and $D(t)=D(T-t)$, then $\lambda\in\sigma(D(A-\nu J)^{-1})$ if and only if $-\lambda\in\sigma(D(A-\nu J)^{-1})$,
and with the same multiplicity.
\end{lem}

As an application, we have
\begin{prop}\label{prop5b.22} Suppose $S=I_{2n}$, $B=0$, $D>0$, and $\omega\neq1$, if one of the following conditions holds
 \begin{itemize}
 \item[1)] $\frac{\omega}{(1-\omega)^2} Tr\[\(J\int_0^TD(s)ds\)^2\]\leq 1$
  \item[2)] $D(t)=D(T-t)$, $\frac{\omega}{2(1-\omega)^2} Tr\[\(J\int_0^TD(s)ds\)^2\]\leq 1$,
 \end{itemize}
 then $i_\omega(\gamma)=0$, where $\gamma$ is the fundamental solution with respect to $D$.
\end{prop}
\begin{proof}  Since $M=S=I_{2n}$, by (\ref{5b.6.1c}), \begin{eqnarray*}Tr((D(A-\nu J)^{-1})^{2})=\frac{\omega}{(1-\omega)^2} Tr\[\(J\int_0^TD(s)ds\)^2\].\end{eqnarray*}
The proofs of both cases are similar, we only list the proof under the second condition.
By Lemma \ref{lem5b.2} and the (\ref{4.0.1.3}),
\begin{eqnarray*}  Tr\[\(D(A-\nu J)^{-1}\)^{2}\]=2\sum_{j\in\mathbb{N}}\frac{1}{\lambda_j^{2}}\label{5b.7.1}.  \end{eqnarray*}
Thus we have \begin{eqnarray*}  \frac{\omega}{2(1-\omega)^2} Tr\[\(J\int_0^TD(s)ds\)^2\]=\sum_{j\in\mathbb{N}}\frac{1}{\lambda_j^{2}}\label{5b.8.1}.  \end{eqnarray*}
Notice that $\frac{\omega}{2(1-\omega)^2} Tr\[\(J\int_0^TD(s)ds\)^2\]\leq 1$. By the  same discussion  as in the proof of Proposition \ref{prop5b.1}, we have
\begin{eqnarray*} I(A-\nu J,A-\nu J-D)=0.  \label{5.9.1} \end{eqnarray*}
By Proposition \ref{prop5a.2}, $i_\omega(\gamma)=I(A-\nu J,A-\nu J-D)=0.$ The proof is complete.
\end{proof}

\subsection{Stability criteria }\label{subsec4.3}

In this section, we only consider  the case $S=I_{2n}$, and the general case is similar. Recall that $\gamma$ is the fundamental solution with respect $B$ and $M=\gamma(T)$,
we denote $\tilde{\gamma}$ be the fundamental solution with respect to $B+D$, and write $\widetilde{M}=\tilde{\gamma}(T)$.
\begin{prop}\label{prop5c.1} Suppose $D_1\leq D\leq D_2$, where $D_1<0$, $D_2>0$. If \text{for} j=1,2,
\begin{eqnarray*}Tr ((D_j(A-B)^{-1})^2)\leq1\quad \text{and}\quad
Tr ((D_j(A-\nu J-B)^{-1})^2)\leq1\quad\end{eqnarray*} then $i_{\omega}(\gamma)=i_{\omega}(\tilde{\gamma})$, and
\begin{eqnarray} e_\omega(\widetilde{M})/2\geq |i_1(\gamma)-i_{\omega}(\gamma)|,\,\ where \,\ \omega=e^{\nu T} \label{5c.1}. \nonumber \end{eqnarray}
Especially, if for $j=1,2$,
 \begin{eqnarray}Tr \[\(D_j\(A-\frac{\sqrt{-1}\pi}{T} J-B\)^{-1}\)^2\]\leq1,\nonumber\end{eqnarray} then
\begin{eqnarray} e(\widetilde{M})/2\geq |i_1(\gamma)-i_{-1}(\gamma)| \label{5c.1.0}. \nonumber \end{eqnarray}
\end{prop}

\begin{proof}  Since $Tr ((D_j(A-B)^{-1})^2)\leq1$ for $j=1,2$, by Corollary \ref{cor5b.1a},
\begin{eqnarray*}
I(A-B, A-B-D_j)=0.
\end{eqnarray*}
Hence, for $ j=1,2,$
\begin{eqnarray*} I(A,A-B-D_j)=I(A,A-B)+I(A-B,A-B-D_j)=I(A,A-B),  \label{6.1} \end{eqnarray*} thus by (\ref{indexn}) \begin{eqnarray*}
I(A,A-B-D)=I(A,A-B),\label{6.2} \end{eqnarray*} and from Proposition
\ref{prop5a.2}, we have \begin{eqnarray*}
i_1(\gamma)=i_1(\tilde{\gamma}).\label{6.3} \end{eqnarray*} Similar, $Tr
((D_j(A-\nu J-B)^{-1})^2)\leq1$ for $j=1,2$ implies \begin{eqnarray*}
i_{\omega}(\gamma)=i_{\omega}(\tilde{\gamma}).\label{6.4} \end{eqnarray*}
 From (\ref{5a.25}), \begin{eqnarray*} e_\omega(\widetilde{M})/2\geq |i_{1}(\tilde{\gamma})-i_{\omega}(\tilde{\gamma})|= |i_{1}(\gamma)-i_{\omega}(\gamma)|.\label{6.5} \end{eqnarray*}
The desired result is proved.
\end{proof}

$Tr ((D_j(A-B)^{-1})^2)$ could be estimated by using the trace formula. If
moreover $MJ=JM$ and $M^T=M$,  we could have a more simple estimation.
\begin{cor}\label{cor5c.1} Under the condition of Proposition \ref{prop5c.1},
if moreover $MJ=JM$, $M^T=M$,  for $j=1,2$,
 \begin{eqnarray} Tr\[\(J\int_0^T\hat{D}_j(s)ds \cdot M(M-\omega I_{2n})^{-1}\)^2\]- Tr\[\(J\int_0^T\hat{D}_j(s)ds\)^2M(M-\omega I_{2n})^{-1}\]\leq 1, \label{5c.2}  \end{eqnarray}
 and  \begin{eqnarray} Tr\[\(J\int_0^T\hat{D}_j(s)ds \cdot M(M- I_{2n})^{-1}\)^2\]- Tr\[\(J\int_0^T\hat{D}_j(s)ds\)^2M(M- I_{2n})^{-1}\]\leq 1, \label{5c.2a}  \end{eqnarray}
 where $\hat{D}_j(t)=\gamma_0^T(t)D_j(t)\gamma_0(t)$,
 then \begin{eqnarray} e_\omega(\widetilde{M})/2\geq |i_1(\gamma)-i_{\omega}(\gamma)| \label{5c.3}.\nonumber  \end{eqnarray}
\end{cor}

\begin{proof} From Proposition \ref{prop5b.21}, in case $MJ=JM$, $M^T=M$, the equality (\ref{5c.2}) implies \begin{eqnarray}
Tr((D_j(A-\nu J-B)^{-1})^{2})\leq1. \nonumber\end{eqnarray}  By Proposition \ref{prop5c.1},
$i_{\omega}(\gamma)=i_{\omega}(\tilde{\gamma})$. Similarly, by (\ref{5c.2a}), $i_1(\gamma)=i_1(\tilde{\gamma})$. The result is from (\ref{5a.25}).
\end{proof}

\begin{thm}\label{prop5c.2.1}
If $M=I_{2n}$, $D>0$ (or $D<0$),  $\frac{\omega}{(1-\omega)^2} Tr\[\(J\int_0^T\hat{D}(s)ds\)^2\]\leq 1$
 then \begin{eqnarray} e_\omega(\widetilde{M})/2=n \label{5c.3}.  \end{eqnarray}
\end{thm}

\begin{proof} Firstly, we will prove the result in the case of $D>0$.  Since $M=I_{2n}$, by \cite[Chapter 9]{Lon4},  we have
\bea i_{\omega}(\gamma)=i_{1}(\gamma)+n.\label{6.6}\nonumber \eea On the
other hand, since $D>0$, by (\ref{5a.14}) \bea I(A,A-B-D)\geq
I(A,A-B)+\dim\ker(A-B)=I(A,A-B)+2n.\label{6.7} \nonumber\eea Thus \bea
i_1(\tilde{\gamma})\geq i_{1}(\gamma)+2n.\label{6.8}\nonumber \eea By the
condition $\frac{\omega}{(1-\omega)^2}
Tr((J\int_0^T\hat{D}(s)ds)^2)\leq 1$, we have \bea
i_\omega(\tilde{\gamma})= i_\omega(\gamma).\label{6.9} \eea
  The result follows from (\ref{5a.25}).

  In the case $D<0$, we have $$ I(A,A-B-D)\leq I(A,A-B),  $$
  this is equivalent to $i_1(\tilde{\gamma})\leq i_{1}(\gamma)$. On the other hand, we have $i_\omega(\tilde{\gamma})= i_\omega(\gamma)$.
The  result follows from (\ref{5a.25}). The proof is complete.
\end{proof}
By taking $\omega=-1$, we have
\begin{cor}\label{cor6..1}
If $M=I_{2n}$, $D>0$ (or $D<0$), $ -\frac{1}{4}Tr\[\(J\int_0^T\hat{D}(s)ds\)^2\]\leq 1$,
then $e(M)/2=n$, that is $\tilde{M}$ is elliptic.
\end{cor}

In the special case $B(t)\equiv0$, then $\gamma(t)\equiv I_{2n}$ is a constant path, it is well known
$i_1(\gamma)=-n$, and $i_\omega(\gamma)=0$ for $\omega\in \mathbb U\setminus\{1\}$(see  \cite{Lon4}).

\begin{cor}\label{pro5c.3}
Suppose $B=0$ and $D>0$ (or $D<0$) if one of the following  conditions satisfies:
\begin{itemize}
 \item[(i)] $\frac{\omega}{(1-\omega)^2} Tr\[\(J\int_0^TD(s)ds\)^2\]\leq 1$,
 \item[ (ii)]  $D(t)=D(T-t)$ and  $\frac{\omega}{2(1-\omega)^2} Tr\[\(J\int_0^TD(s)ds\)^2\]\leq 1$,
\end{itemize}
  then
\begin{eqnarray} e_\omega(\widetilde{M})/2=n \label{5c.4}. \nonumber \end{eqnarray}
\end{cor}

\begin{proof}  The result under condition (i) comes directly from Theorem \ref{prop5c.2.1}, since $\hat{D}=D$ for $B=0$. For condition (ii), by Proposition \ref{prop5b.22}, $i_\omega(\tilde{\gamma})=i_\omega(\gamma)=0$. In this case $\gamma\equiv I_{2n}$ is a constant solution. By some similar argument to the proof of Theorem \ref{prop5c.2.1}, we prove the result.
\end{proof}

We will give some hyperbolic criteria
\begin{prop}\label{prop5c.4} Suppose $M$ is hyperbolic,
 $Tr \[\(D(A-\nu J-B)^{-1}\)^2\]\leq1$ for $\nu\in\[0,\frac{\sqrt{-1}\pi}{T}\]$, then
$\widetilde{M}$ is hyperbolic.
\end{prop}

\begin{proof} Please note that $Tr ((D(A-\nu J-B)^{-1})^2)\leq1$, thus $A-\nu J-B-sD$ is non-degenerate for $s\in[0,1]$. This is equivalent to $A_\omega-B_\omega-sD_\omega$ is non-degenerate. Therefore $\widetilde{M}-\omega I_{2n}$ is nonsingular for $\omega\in \mathbb U$, thus $\widetilde{M}$ is hyperbolic.
\end{proof}

When $B$ is constant path, our stability criteria can be easily  used. Next example will give a new stability criteria.
\begin{exam} Suppose $B(t)\equiv B $ is constant path of matrices, $JB=BJ$ and $\exp(JBT)=I_{2n}$. This  happens when $B=diag(\alpha_1,\alpha_2,...\alpha_n,
\alpha_1,\alpha_2,...\alpha_n )$, and  $\alpha_jT/2\pi\in\mathbb{Z}$ for $j=1,...,n$.  Consider the linear Hamiltonian systems
\bea \dot{z}(t)=J(B+D(t))z(t), \label{4.60} \eea with $D(t)=D(t+T)\geq0$ and $\int_0^TD(t)dt>0$.  Let $\lambda(t)=\lambda_{max}(D(t))$ which is the largest eigenvalue of $B(t)$, then (\ref{4.60}) is spectrally stable if
\bea \int_0^T\lambda(t)dt<2. \label{4.61}\eea
In fact, noting that $D(t)\leq \lambda(t)I_{2n}$, let $\tilde{\gamma}(t)$ and $\tilde{\gamma}_1(t)$ be the fundamental solutions corresponding to $B+D(t)$ and $B+\lambda(t)I_{2n}$ respectively, then
$$i_\omega(\tilde{\gamma})\leq i_\omega(\tilde{\gamma}_1), \,\ \forall \omega\in \mathbb U. $$
By some easy computation, the condition (\ref{4.61}) implies $i_{-1}(\tilde{\gamma}_1)=i_{-1}(\gamma)$. On the other hand, by the proof of Theorem \ref{prop5c.2.1}, we have $i_1(\tilde{\gamma})\geq i_{1}(\gamma)+2n$ and $i_{-1}(\gamma)=i_1(\gamma)+n$, which yields the result by (\ref{5a.26}). Please note that, in the case $B=0$,
if we we instead (\ref{4.61}) by the condition (i) of Corollary \ref{pro5c.3}, we also get $e_\omega(\widetilde{M})/2=n$, which is a generalization of Krein's stability criteria.
\end{exam}

\subsection{ Estimate the  Morse index for  $\bar{S}$-periodic orbits in Lagrangian system}\label{subsec4.4}

In this section, we will estimate the Morse index of $\bS$-periodic orbits in Lagrangian systems by using the trace formula.
 For $T>0$, suppose $x(t)$ is
a critical point of the functional
$$F(x)=\int_0^TL(t,x,\dot{x}),  \forall\,\, x\in E=\left\{x \,\left|\, x\in W^{1,2}(\mathbb{R},\mathbb{R}^n), x(t)=\bS x(t+T)\right.\right\}$$
where $L\in C^2(\mathbb{R}\times \mathbb{R}^{2n},\mathbb{R})$ and
satisfies circle type symmetry  \cite{HS}  \bea
L(t,x,\xi)=L(t+T,\bS^Tx,\bS^T\xi). \label{2a.1} \eea
 It is well known that
$x(t)$ is a solution of the corresponding Euler-Lagrangian equation:
\bea \frac{d}{dt}L_p(t,x,\dot{x})-L_x(t,x,\dot{x})=0,\label{n3.l2}
x(0)=\bS x(T), \label{n3.l2} \dot{x}(0)=\bS\dot{x}(T). \eea For such an
extremal loop, define \bea P(t)=L_{p,p}(t,x(t),\dot{x}(t)),
Q(t)=L_{x,p}(t,x(t),\dot{x}(t)),
R(t)=L_{x,x}(t,x(t),\dot{x}(t)).\nonumber\eea
Note that   \bea
F''(x)=-\frac{d}{dt}\(P\frac{d}{dt}+Q\)+Q^T\frac{d}{dt}+R.
\label{n3.f}\nonumber\eea
 For $\omega\in\mathbb{U}$, set $ D_{\omega \bS}=\{y\in
W^{1,2}([0,T];\mathbb{C}^n)\,|\, y(0)=\omega \bS y(T) \}$.  We
define the $\omega$-Morse index $\phi_\omega(x)$ of $x$ to be the
dimension of the negative definite subspace of \bea \langle
F''(x)y_1,y_2 \rangle, y_1,y_2\in D_{\omega \bS}.\nonumber  \eea

For $\omega=e^{\nu T}$ with imaginary number $\nu$,  recall that   \bea  \cal{A}(\nu)=-\(\d+\nu\)P(t)\(\(\frac{d}{dt}+\nu\)+Q\)+Q^T(t)\(\frac{d}{dt}+\nu\)+R(t),\nonumber  \eea
with domain $D_{\bS}$.
We also denote by $\phi_\omega(\cal A)$ the $\omega$-Morse index of $\cal A$, which is defined to be the dimension of the negative definite space of
\bea \langle \mathcal{A}y_1,y_2 \rangle, y_1,y_2\in
D_{\omega \bS}.\nonumber  \eea
 Obviously,
\bea \phi_\omega(\cal A(0))=\phi_1(\cal A(\nu)). \eea

The next lemma is obvious.
\begin{lem}\label{lem5.1} Suppose $R_1\geq 0$, then
\bea \phi_1({\cal A}(\nu)+R_1)\leq  \phi_1({\cal A}(\nu)). \eea
\end{lem}

When we transform  the Sturm-Liouville system to linear Hamiltonian system, it is obvious
 \bea \dim\ker (A-\nu J-B)= \dim\ker ({\cal A}(\nu)). \label{4.23} \eea
 Moreover,
 the Morse index is essentially same as the relative Morse index ( Maslov-type index)(see \cite{Lon4} or \cite{HS}). Recall that $B_\lambda(t)=\left(\begin{array}{cc}P^{-1}(t)& -P^{-1}Q(t) \\
-Q(t)^TP^{-1}(t)  & Q(t)^TP^{-1}(t)Q(t)-R(t)-\lambda R_1(t)
\end{array}\right).$ We have the following proposition.

\begin{prop}
\bea I(A-\nu J-B,A-\nu J-B_1)= \phi_1({\cal A}(\nu)+R_1)- \phi_1({\cal A}(\nu))=i_\omega(\ga_1)-i_\omega(\ga_0). \label{4.24}  \eea
\end{prop}
\begin{proof} Let $\gamma_\lambda$ be the fundamental solution corresponding to $B_\lambda$, then from \cite[P172]{Lon4}, we have
\bea \phi_1({\cal A}(\nu)+\lambda R_1)=i_\omega(\gamma_\lambda).  \eea
Thus \bea \phi_1({\cal A}(\nu) + R_1)-\phi_1({\cal A}(\nu))=i_\omega(\gamma_1)-i_\omega(\gamma_0), \eea
the result is from Proposition \ref{prop5a.2}.
\end{proof}

By (\ref{4.23}) and (\ref{4.24}), all the result in \S \ref{subsec4.2} can be used to estimate the Morse index and non-degenerate  of linear Lagrangian systems, however, there are some new estimation
for  the Lagrangian system.

\begin{thm}\label{thm5.2}
Let $\nu\in \mathbb C$, assume $\mathcal{A}(\nu)>0$, if $R_1\geq -K$, where
 $K\in\mathcal{B}(n)$ and $K>0$.  Then
\bea \phi_1(\mathcal{A}(\nu)+R_1)\leq \inf\{Tr ((K(\mathcal{A}(\nu))^{-1})^k), k\in\mathbb N\}. \label{th5.2f}\eea

\end{thm}

\begin{proof} Please note that in this case, all the eigenvalues $\{1/\lambda_j\}$ of $D(\mathcal{A}(\nu))^{-1}$ are positive, and $K(\mathcal{A}(\nu))^{-1}$ is a trace class operator. Hence for any positive integers $l$,
 \begin{eqnarray*}  Tr\[\(K\mathcal{A}(\nu)^{-1}\)^{l}\]>\sum_{|\lambda_j|\leq 1}\frac{1}{\lambda_j^{l}} \label{5b.1.1}.  \end{eqnarray*}
Similar argument to  the proof of Proposition \ref{prop5b.1} implies the result.
\end{proof}
\begin{cor}\label{cor5b.1}
Under the conditions of Theorem \ref{thm5.2}, if $Tr ((D(\mathcal{A}(\nu))^{-1})<1$, then
$$\phi_1(\mathcal{A}(\nu)+R_1)=\phi_1(\mathcal{A}(\nu))=0$$ and $\mathcal{A}(\nu)+R_1$ is non-degenerate.
\end{cor}

 Next, we will consider some special case
that $\mathcal{A}(\nu)=-\left(\d+\nu\right)^2-R(t)$. Let
$R^+(t)=\frac{1}{2}(R(t)+|R(t)|)$, then $R^+(t)\geq0$, and $R(t)\leq
R^+(t)$,  we have

\begin{thm}\label{thm5.3}
For  imaginary number $\nu$, such that $-\left(\d+\nu\right)^2$ is invertible,
\bea \phi_1\(-\(\d+\nu\)^2-R(t)\)\leq -\omega T\cdot Tr\[\int_0^TR^+(t)dt\cdot S(S-\omega)^{-2}\], \label{th5.3f}\eea
where $\omega=e^{\nu T}$.
\end{thm}

\begin{proof}  For any $\varepsilon>0$, $R^+(t)+\varepsilon I_n>0$, and $\phi_1(-(\d+\nu)^2-R(t))\leq \phi_1(-(\d+\nu)^2-(R^+(t)+\varepsilon I_n)$. The result follows from (\ref{aa.1}) and Theorem \ref{thm5.2}.
\end{proof}

\section{Stability of Lagrangian orbits}
In this section, we will give the  application of the trace formula  on the stability for elliptic Lagrangian orbits. To do this, in  \S \ref{subsec5.1} we will recall some elementary results on Maslov-type index and Morse index of Lagrangian orbits. In \S \ref{subsec5.2}, we will prove Theorem \ref{la1.1}.  Details on the function $f(\bb,\omega)$ in Theorem \ref{la1.1} via the trace formula (\ref{ht2}) will be listed in \S \ref{subsec5.3}. At last, in \S \ref{subsec5.4}, by the first order trace formula (\ref{lt.2}) we will give another estimation for the hyperbolic region which is not too sharper but with more simple estimation.
\subsection{A brief review on  Lagrangian orbits }\label{subsec5.1}

Following Meyer and Schmidt \cite{MS}, the linear variational
equation of the elliptic equilibria is  decoupled into three
subsystems, the first and second subsystems are from the first
integral and the third is the essential part.
The essential
part $\ga=\ga_{\bb,e}(t)$ of the fundamental solution of the
Lagrangian orbit \cite[P.275]{MS} satisfies \bea
\dot{\gamma}(t) &=& JB_{\beta,e}(t)\gamma(t),   \lb{2.17}\\
\gamma(0) &=& I_{4},    \lb{2.18}\eea
with \be B_{\beta,e}(t)=\left(\begin{array}{ccccc}1
& 0 & 0 & 1\\
                       0 & 1 & -1 & 0 \\
                       0 & -1 &\frac{2e\cos(t)-1-\sqrt{9-\beta}}{2(1+e\cos t)} & 0 \\
                       1 & 0 & 0 & \frac{2e\cos(t)-1+\sqrt{9-\beta}}{2(1+e\cos t)}
                       \end{array}\right), \lb{2.19}   \nonumber  \ee
where $e$ is the eccentricity, and $t$ is the truly anomaly.

Let \be  J_2=\left(\begin{array}{cc} 0 & -1 \\ 1 & 0 \end{array}\right), \qquad
    \hat{K}_{\bb,e}(t)=\left(\begin{array}{cc}\frac{3+\sqrt{9-\beta}}{2(1+e\cos t)} & 0 \\
                                     0 & \frac{3-\sqrt{9-\beta}}{2(1+e\cos t)} \end{array}\right),  \lb{2.20}\nonumber\ee
and the corresponding Sturm-Liouville system is \bea -\ddot{y}-2J_2\dot{y}+\hat{K}_{\beta,e}y=0.  \nonumber  \eea
For $(\bb,e)\in [0,9)\times [0,1)$, $\omega\in\mathbb{U}$,   we set
\bea \ol{D}(\omega,2\pi)=\{y\in
W^{2,2}([0,2\pi];\mathbb{C}^n)\,|\, y(0)=\omega y(2\pi), \dot{y}(0)=\dot{y}(2\pi)
\}.\nonumber\eea
and
\bea  \mathcal{A}(\bb,e,\nu)
&=& -\(\frac{d}{dt}+\nu\)^{2}-2J_2\(\frac{d}{dt}+\nu\)+\hat{K}_{\bb,e}(t).    \lb{2.29}\nonumber\eea
 Then for pure imaginary number $\nu$, $\mathcal{A}(\beta,e,\nu)$ are  self-adjoint operators on $L^2([0,2\pi],\mathbb{C}^n)$ with domain $\overline{D}(\omega,2\pi)$
 and depend on the parameters $\bb$ and $e$.
 We simply  denote the operator by $\mathcal{A}_{\omega}(\beta,e,\nu)$ and omit $\omega$ when $\omega=1$.
 Let $\phi(\mathcal{A}_\omega)=\phi_1(\mathcal{A}_\omega)$ be the Morse index of $\mathcal{A}_\omega$.
 It is obvious that $\mathcal{A}_\omega>0$ if and only if $\phi(\mathcal{A}_\omega)=\upsilon(\mathcal{A}_\omega)=0$,
where
$$
\upsilon(\mathcal{A}_\omega)=\dim\ker (\mathcal{A}_\omega).
$$

 For any $x(t)\in \overline{D}(1,2\pi)$, direct computations show that
 \bea
 e^{-t\nu}\mathcal{A}(\beta,e,0)e^{t\nu}x(t)=\mathcal{A}(\beta,e,\nu)x(t),
 \eea
 thus for $\omega=e^{2\pi \nu}$, we have
 \bea
 \phi(\mathcal{A}_{\omega}(\beta,e,0))=\phi(\mathcal{A}(\beta,e,\nu))\quad \text{and}\quad \upsilon(\mathcal{A}_{\omega}(\beta,e,0))=\upsilon(\mathcal{A}(\beta,e,\nu)).
 \eea
    Obviously \bea \phi(\mathcal{A}_{\omega}(\bb,e,0))=I\(-\frac{d^2}{dt^2}, \mathcal{A}_{\omega}(\bb,e,0) \).\nonumber \eea
    By the relationship between Morse index with Maslov-type index \cite[p.172]{Lon4}, we have that for any $\bb$ and $e$ the Morse index $\phi(\mathcal{A}_{\om}(\bb,e,0))$ and nullity $\upsilon(\mathcal{A}_{\om}(\bb,e,0))$
satisfy
\begin{equation}  \phi(\mathcal{A}_{\om}(\bb,e,0)) = i_{\om}(\ga_{\bb,e}),\quad\text{and} \quad \upsilon(\mathcal{A}_{\om}(\bb,e,0)) = \upsilon_{\om}(\ga_{\bb,e}), \qquad
           \forall \,\om\in\mathbb{U}. \label{2.30}\nonumber\end{equation}
In particular,  by (55) and (58) in \cite[Lemma 4.1]{HS1}, we obtain
\be  i_1(\ga_{\bb,e})=\phi(\mathcal{A}(\bb,e,0))=i_1(\ga_{\bb,e})=0, \qquad \forall \,(\bb,e)\in [0,9]\times [0,1). \lb{2.31}\ee


In the case $e=0$, $B_{\bb,0}(t)$ is a constant matrix and $i_\omega(\gamma_{\beta,0})$, $\upsilon_\omega(\ga_{\beta,0})$ could be  computed directly. We list the result for $\omega=-1$ and $\omega=e^{i\sqrt{2}\pi}$ below.
\begin{thm}\label{th3.2} (\cite{HLS})
For any $\omega=e^{2\pi\nu}\in\mathbb{U}$, $\beta\in(1,9]$,   $\mathcal{A}(\bb,0,\nu)>0$ or equivalently \bea i_\omega(\ga_{\bb,0})= \phi(\mathcal{A}(\bb,0,\nu))=\upsilon(\mathcal{A}(\bb,0,\nu))=0.
\label{l1}\eea
For $\omega=e^{i\sqrt{2}\pi/2}$,
 $\upsilon(\mathcal{A}(1,0,i\sqrt{2}\pi/2))=1$, and
\bea i_{e^{i\sqrt{2}\pi/2}}(\ga_{\beta,0})= \phi(\mathcal{A}(\bb,0,i\sqrt{2}\pi/2))\geq1, \,\ for \,\ \bb\in[0,1).\label{l2}\eea
For $\omega=-1$, $\upsilon(\mathcal{A}(3/4,0,i/2))=2$ and $\upsilon(\mathcal{A}(\bb,0,i/2))=0$ if $\bb\neq3/4$,
 \bea i_{-1}(\ga_{\bb,0})=\phi(\mathcal{A}(\bb,0,i/2))=\left\{\begin{array}{ll}2 & \quad
           {\mathrm if}\; \bb\in [0,3/4),  \\
           \\
 0, & \quad {\mathrm if}\; \bb\in[3/4,9] .\end{array}\right.\lb{l3}\eea
\end{thm}

\subsection{Stability analysis via trace formula }\label{subsec5.2}

Set
\be D_{\beta,e}(t)=B_{\beta,e}(t)-B_{\beta,0}(t)=\frac{e\cos(t)}{1+e\cos(t)}K_\beta ,\nonumber \ee
where \be K_\beta=\left(\begin{array}{cccc}0 & 0 & 0 & 0\\
                       0 & 0 & 0 & 0 \\
                       0 & 0 &\frac{3+\sqrt{9-\beta}}{2} & 0 \\
                       0 & 0 & 0 & \frac{3-\sqrt{9-\beta}}{2} \end{array}\right), \nonumber \ee
then \be -J\frac{d}{dt}-B_{\beta,e}=-J\frac{d}{dt}-B_{\beta,0}-D_{\beta,e}. \nonumber \ee
Let  $\cos^\pm(t)=(\cos(t)\pm |\cos(t)|)/2$, and  denote \bea K^\pm_\beta=\cos^\pm(t)K_\beta, \nonumber \eea
which can be considered as two  bounded self-adjoint operators on $L^2([0,2\pi], \mathbb C^{2n})$; moreover   $K^+_\beta\geq0$ and $K^-_\beta\leq0$.  It is obvious that
\bea -J\frac{d}{dt}-\nu J-B_{\beta,e}&\geq &-J\frac{d}{dt}-\nu J-B_{\beta,0}- \frac{e \cos^+(t)}{1+e\cos^+(t)}K_\beta \nonumber \\ &\geq &-J\frac{d}{dt}-\nu J-B_{\beta,0}-eK^+_\beta, \lb{3.a1} \eea
equivalently, \bea \mathcal{A}(\bb,e,\nu)&\geq & \mathcal{A}(\bb,0,\nu)-  \frac{e}{1+e\cos^+(t)}\cos^+(t)\hat{K}_{\beta,0} \nonumber \\ &\geq &\mathcal{A}(\bb,0,\nu)-  e\cos^+(t)\hat{K}_{\beta,0} \lb{3.a1.1}. \eea
 Similarly,   \bea -J\frac{d}{dt}-\nu J-B_{\beta,e}&\leq &-J\frac{d}{dt}-\nu J-B_{\beta,0}- \frac{e}{1+e\cos^-(t)}K^-_\beta \nonumber \\ &\leq &-J\frac{d}{dt}-\nu J-B_{\beta,0}- \frac{e}{1-e}K^-_\beta, \lb{3.a2}  \eea
which is equivalent to
 \bea \mathcal{A}(\bb,e,\nu)&\leq & \mathcal{A}(\bb,0,\nu)-  \frac{e}{1+e\cos^-(t)}\cos^-(t)\hat{K}_{\beta,0}  \nonumber \\ &\leq &\mathcal{A}(\bb,0,\nu)-  \frac{e}{1-e}\cos^-(t)\hat{K}_{\beta,0}. \lb{3.a2.1}  \eea

\begin{lem}\label{lem5.3} For an imaginary number $\nu$, such that $-J\frac{d}{dt}-\nu J-B_{\beta,0}$ is invertible, we have
\bea Tr\[ \(K_{\beta}^+\(-J\frac{d}{dt}-\nu J-B_{\beta,0}\)^{-1}\)^2\]=Tr\[ \(K_{\beta}^-\(-J\frac{d}{dt}-\nu J-B_{\beta,0}\)^{-1}\)^2\]\nonumber\eea
\end{lem}
\begin{proof} Define an operator $G:x(t)\rightarrow x(t+\pi)$ on the domain $\ol{D}(1,2\pi)$, then $G^{2}=id$. Direct calculation shows that
\begin{eqnarray*}
\(-J\frac{d}{dt}-\nu J-B_{\beta,0}\)^{-1} G=G \(-J\frac{d}{dt}-\nu J-B_{\beta,0}\)^{-1}.
\end{eqnarray*}
Moreover, $K_{\beta,0} G=GK_{\beta,0}$ because  $K_{\beta,0}$ is a constant matrix.
Therefore,
\bea
Tr\[ \(G \cos^{+}(t)K_{\beta}\(-J\frac{d}{dt}-\nu J-B_{\beta,0}\)^{-1}G\)^2\]
&=&Tr\[ \(G \cos^{+}(t)G K_{\beta}\(-J\frac{d}{dt}-\nu J-B_{\beta,0}\)^{-1}\)^2\]\nonumber \\
&=&Tr\[ \(\cos^{-}(t)K_{\beta}\(-J\frac{d}{dt}-\nu J-B_{\beta,0}\)^{-1}\)^2\]. \nonumber
\eea
\end{proof}
Under the assumption of Lemma \ref{lem5.3},  we  denote \bea f(\beta,\omega)=Tr\[ \(K_{\beta}^-\(-J\frac{d}{dt}-\nu J-B_{\beta,0}\)^{-1}\)^2\]=Tr\( \(K_{\beta}^+(-J\frac{d}{dt}-\nu J-B_{\beta,0})^{-1}\)^2\), \lb{fbb} \eea which is a positive function.   The following theorem holds true.
\begin{thm}\label{th2.1}
For $\bb\in[0,3/4)$, $\gamma_{\bb,e}$ is spectrally stable if \bea 0\leq e<\frac{1}{1+\sqrt{f(\bb,-1)}}. \lb{th2.1f}  \eea
\end{thm}
\begin{proof}   Obviously,
\bea
  Tr\( \(\frac{e}{1-e}K_{\beta}^-\(-J\frac{d}{dt}-\frac{\sqrt{-1}}{2} J-B_{\beta,0}\)^{-1}\)^2\)&=&  \frac{e^2}{(1-e)^2} Tr\( \(K_{\beta}^-\(-J\frac{d}{dt}-\frac{\sqrt{-1}}{2} J-B_{\beta,0}\)^{-1}\)^2\) \nonumber\\
&=&  \frac{e^2}{(1-e)^2}f(\beta,-1). \nonumber \eea  Thus,   (\ref{th2.1f} ) is equivalent to $\frac{e^2}{(1-e)^2}f(\beta,-1)<1$ which implies $Tr\( \(\frac{e}{1-e}K_{\beta}^-(-J\frac{d}{dt}-\frac{\sqrt{-1}}{2} J-B_{\beta,0})^{-1}\)^2\)<1$. By the continuity of the trace, for $\epsilon>0$ small enough,
$Tr\( \((\frac{e}{1-e}K_{\beta}^--\epsilon I_{2n})(-J\frac{d}{dt}-\frac{\sqrt{-1}}{2} J-B_{\beta,0})^{-1}\)^2\)<1.$ Obviously,  $\frac{e}{1-e}K_{\beta}^--\epsilon I_{2n}<0$. By
Theorem  \ref{thm4.2} and Theorem \ref{thm4.3}, $-J\frac{d}{dt}-\frac{\sqrt{-1}}{2} J-B_{\beta,0}-\frac{e}{1-e}K_{\beta}^-$ is non-degenerate and
$$ I\(-J\frac{d}{dt}-\frac{\sqrt{-1}}{2} J-B_{\beta,0},-J\frac{d}{dt}-\frac{\sqrt{-1}}{2} J-B_{\beta,0}-\frac{e}{1-e}K_{\beta}^-\)=0.  $$
From (\ref{3.a2}), $ I\left(-J\frac{d}{dt}-\frac{\sqrt{-1}}{2} J-B_{\beta,0}-\frac{e}{1-e}K_{\beta}^-, -J\frac{d}{dt}-\frac{\sqrt{-1}}{2} J-B_{\beta,e}\right)\geq0  $, consequently,
$$I\(-J\frac{d}{dt}-\frac{\sqrt{-1}}{2} J-B_{\beta,0}, -J\frac{d}{dt}-\frac{\sqrt{-1}}{2} J-B_{\beta,e}\)\geq0.  $$
By (\ref{l3}) \bea i_{-1}(\gamma_{\bb,e})\geq i_{-1}(\gamma_{\bb,0})=2. \nonumber \eea
By (\ref{2.31}) and  (\ref{5a.26}), $e(\gamma_{\beta,e})/2=2$. The desired result is proved.
 \end{proof}

\begin{thm}\label{th2.2}
For $\bb\in(3/4,1)$, $\gamma_{\bb,e}$ is spectrally stable if \bea 0\leq e<f(\beta,-1)^{-\frac{1}{2}}.\lb{th2.2f1}  \eea and \bea
 0\leq e<\frac{1}{1+f(\beta,e^{i\sqrt{2}\pi})^{\frac{1}{2}}}. \lb{th2.1f2}  \eea
\end{thm}
\begin{proof}
 Firstly, we'll show that (\ref{th2.2f1}) implies  \bea i_{-1}(\gamma_{\bb,e})= 0, \lb{c.1} \eea and  the proof is similar to the proof of Theorem\ref{th2.1}. In fact, please note
 \bea  Tr\( \(eK_{\beta}^+\(-J\frac{d}{dt}-\frac{\sqrt{-1}}{2} J-B_{\beta,0}\)^{-1}\)^2\)&=&  e^2 Tr\( \(K_{\beta}^+\(-J\frac{d}{dt}-\frac{\sqrt{-1}}{2} J-B_{\beta,0}\)^{-1}\)^2\) \nonumber\\
&=&  e^2f(\beta,-1).  \nonumber\eea
  Thus,   (\ref{th2.1f} ) impli1es $Tr\( \(eK_{\beta}^+(-J\frac{d}{dt}-\frac{\sqrt{-1}}{2} J-B_{\beta,0})^{-1}\)^2\)<1$ , then for $\epsilon>0$ small enough,
$$Tr\( \((eK_{\beta}^++\epsilon I_{2n})(-J\frac{d}{dt}-\frac{\sqrt{-1}}{2} J-B_{\beta,0})^{-1}\)^2\)<1.$$  Obviously,  $eK_{\beta}^++\epsilon I_{2n}>0$. Again, by Theorem \ref{thm4.2} and  Theorem \ref{thm4.3},   $-J\frac{d}{dt}-\frac{\sqrt{-1}}{2} J- B_{\beta,0}-eK_{\beta}^+$ is non-degenerate and $I(-J\frac{d}{dt}-\frac{\sqrt{-1}}{2} J-B_{\beta,0},-J\frac{d}{dt}-\frac{\sqrt{-1}}{2} J-B_{\beta,0}-eK_{\beta}^+)=0$. By (\ref{3.a1}),
$$I\(-J\frac{d}{dt}-\frac{\sqrt{-1}}{2} J-B_{\beta,0}- eK_{\beta}^+,-J\frac{d}{dt}-\frac{\sqrt{-1}}{2} J- B_{\bb,e}\)\leq0.$$
Therefore
\bea I\(-J\frac{d}{dt}-\frac{\sqrt{-1}}{2} J-B_{\beta,0},-J\frac{d}{dt}-\frac{\sqrt{-1}}{2} J- B_{\bb,e}\)\leq0. \nonumber \eea By (\ref{l3}), we have (\ref{c.1}).

On the other hand,  almost the same proof as that of Theorem \ref{th2.1} shows that  (\ref{th2.1f2}), (\ref{l2}) implies \bea i_{e^{i\sqrt{2}\pi}}(\gamma_{\bb,e})\geq i_{e^{i\sqrt{2}\pi}}(\gamma_{\bb,0})\geq1. \lb{c.2} \eea
The result comes from (\ref{c.1}), (\ref{c.2}), (\ref{2.31}) and (\ref{sta1}).
\end{proof}

\begin{rem} \label{a5.5} It has been proved in \cite{HLS},\cite{HS1} that $\gamma_{\beta,e}(2\pi)$ is linear stable when $(\beta,e)$ is in the stable region and not on the bifurcation curves.  This implies that under the condition in Theorem \ref{th2.1} and Theorem \ref{th2.2}, $\gamma_{\beta,e}$  is linear stable. Moreover,   the  normal form of $\gamma_{\beta,e}(2\pi)$ was given in  \cite{HLS},\cite{HS1}. Precisely, for $(\beta,e)$ in the stable region given in Theorem \ref{th2.1}, $\gamma_{\bb,e}(2\pi)\approx R(\theta_1)\diamond R(\theta_2)$ for some $\theta_1,\theta_2\in(\pi,2\pi)$; for $(\beta,e)$ in the stable region given in Theorem \ref{th2.2}, $\gamma_{\bb,e}(2\pi)\approx R(\theta_1)\diamond R(\theta_2)$ for some $\theta_1\in ((2-\sqrt{2})\pi,\pi),\theta_2\in(\sqrt{2}\pi,2\pi)$.

\end{rem}

To estimate the hyperbolic region,  denote
\bea \hat{f}(\beta)=\sup\{f(\bb,\omega),\omega\in\mathbb{U}\},\label{eq5.47a} \eea
and we have
\begin{thm}\label{th2.3}
For $\bb\in(1,9]$, $\gamma_{\bb,e}$ is hyperbolic if \bea e<\hat{f}(\bb)^{-1/2}. \lb{th2.3f} \eea \end{thm}
\begin{proof} Similar to the proof of Theorem \ref{th2.2}, the condition (\ref{th2.3f}) implies that for any $\omega\in\mathbb{U}$
\bea i_\omega(\gamma_{\bb,e})\leq i_\omega(\gamma_{\bb,0})=0,  \nonumber\eea
 and \bea \upsilon(A_\omega(\bb,e,0))= \upsilon(A_\omega(\bb,0,0))=0, \nonumber \lb{eee} \eea which implies that $\gamma_{\bb,e}$ is hyperbolic.    \end{proof}
Combining Theorem \ref{th2.1}, Theorem \ref{th2.2} with Theorem \ref{th2.3} and Remark \ref{a5.5}, we have Theorem \ref{la1.1}.
The function $f(\beta,\omega)$ will be dealt with in the next subsection, and based on this,
with the help of Mathlab, we can draw a picture of the stable region and hyperbolic region in Figure 1.

\subsection{The precise form of $f(\bb,\omega)$}\label{subsec5.3}
\label{sec:4}
In this subsection, we compute  $f(\bb,\omega)$ by trace formula (\ref{ht2}). In order to make the calculation easier, we need to
use some transformation first. By the definition, for $e=0$,
\be B_{\bb,0}(t)= B_\beta=\left(\begin{array}{cccc}1 & 0 & 0 & 1\\
                       0 & 1 & -1 & 0 \\
                       0 & -1 &\frac{-1-\sqrt{9-\beta}}{2} & 0 \\
                       1 & 0 & 0 & \frac{-1+\sqrt{9-\beta}}{2} \end{array}\right). \lb{2.20}\nonumber\ee
 For $\bb\in(0,9]\ \setminus\{1\}$, let
$P_\beta=\left(\begin{array}{cc}P_{11} & P_{12}\\ P_{21} & P_{22}\end{array}\right) $ be the $4\times4$ transformation matrices, where
\be  P_{11}=\left(
              \begin{array}{cc}
                0 & P_{11}(1,2) \\
                P_{11}(2,1) & 0 \\
              \end{array}
            \right)
=\left(\begin{array}{cc}  0 &       \frac{ (2+2\sqrt{1-\beta})^{1/4}(\sqrt{9-\beta}-\sqrt{1-\beta})  }{ 2(1-\beta)^{1/4}  \sqrt{4+\sqrt{1-\beta}-\sqrt{9-\beta}}   }\\ \frac{ -(2-2\sqrt{1-\beta})^{3/4}(2+\sqrt{9-\beta}-\sqrt{1-\beta}) }{2 (1-\beta)^{1/4}(3+\sqrt{9-\beta})\sqrt{4-\sqrt{1-\beta}-\sqrt{9-\beta} }  } &0\end{array}\right),
\nonumber\ee
\be  P_{12}=\left(
              \begin{array}{cc}
                P_{12}(1,1) & 0 \\
                0 & P_{12}(2,2) \\
              \end{array}
            \right)
=\left(\begin{array}{cc} \frac{-(2-2\sqrt{1-\beta})^{1/4}(\sqrt{9-\beta}+\sqrt{1-\beta}) }{ 2(1-\beta)^{1/4}\sqrt{ 4-\sqrt{1-\beta}-\sqrt{9-\beta}        } }          &            0\\ 0&  \frac{(2+2\sqrt{1-\beta})^{3/4}(2+\sqrt{9-\beta}+\sqrt{1-\beta})}{2(1-\beta)^{1/4}(3+\sqrt{9-\beta}) \sqrt{4+\sqrt{1-\beta}-\sqrt{9-\beta}} }\end{array}\right),
\nonumber\ee
\be  P_{21}=\left(
              \begin{array}{cc}
                P_{21}(1,1) & 0 \\
                0 & P_{21}(2,2) \\
              \end{array}
            \right)
=\left(\begin{array}{cc}  \frac{(2-2\sqrt{1-\beta})^{3/4}(4+\sqrt{9-\beta}+\sqrt{1-\beta}) }{2(1-\beta)^{1/4}(3+\sqrt{9-\beta})\sqrt{ 4-\sqrt{1-\beta}-\sqrt{9-\beta}} } &                                                                                                                0\\ 0&              \frac{-2(2+2\sqrt{1-\beta})^{1/4} }{(1-\beta)^{1/4}\sqrt{4+\sqrt{1-\beta}-\sqrt{9-\beta}} }                              \end{array}\right),
\nonumber\ee
and
\be  P_{22}=\left(
              \begin{array}{cc}
                0 & P_{22}(1,2) \\
                P_{22}(2,1) & 0 \\
              \end{array}
            \right)
=\left(\begin{array}{cc}                                                                                                                0&       \frac{ -(2+2\sqrt{1-\beta})^{3/4}(\sqrt{9-\beta}+4-\sqrt{1-\beta})}{2 (1-\beta)^{1/4}(3+\sqrt{9-\beta}) \sqrt{4+\sqrt{1-\beta}-\sqrt{9-\beta}}}  \\   \frac{ 2(2-2\sqrt{1-\beta})^{1/4}}{(1-\beta)^{1/4} \sqrt{4-\sqrt{1-\beta}-\sqrt{9-\beta}} }                                      &                                                                                                                       0\end{array}\right).
\nonumber\ee
In fact, $P_\beta$ is obtained with the help of matlab. Direct computation shows that
 \bea P_\beta^TJP_\beta=J.\nonumber \eea
 For $\beta\in(0,1)$,  $P_\beta$ is real, thus it is a symplectic matrix,  and for $\beta\in(1,9]$, $P_\beta$ is complex matrix. To continue, we need the notation of symplectic sum, which was introduced by Long \cite{Lon2} and \cite{Lon4}.
Given any two $2m_k\times 2m_k$ matrices of square block form
$M_k=\left(\begin{array}{cc}A_k&B_k\\
                                C_k&D_k\end{array}\right)$ with $k=1, 2$,
the symplectic sum of $M_1$ and $M_2$ is defined by
\bea M_1\diamond M_2=\left(
  \begin{array}{cccc}
   A_1 &   0 & B_1 &   0\\
                            0   & A_2 &   0 & B_2\\
                           C_1 &   0 & D_1 &   0\\
                           0   & C_2 &   0 & D_2  \\
  \end{array}
\right).\nonumber
\eea

Set  $\theta_1(\beta)=-\sqrt{\frac{1}{2}(1-\sqrt{1-\beta})},$ and $\theta_2(\beta)=\sqrt{\frac{1}{2}(1+\sqrt{1-\beta})}$.
Let $$B_j(\beta)=\left(\begin{array}{cc} 0 & -\theta_j(\bb)\\
                                             \theta_j(\bb)& 0\end{array}\right), \,\ for \,\ j=1,2, $$
 and set  \be S_\beta=B_1(\beta)\dm B_2(\beta).\nonumber \ee
Direct computation shows that \bea P_\beta^{-1}JB_\beta P_\beta=JP^T_\beta B_{\beta} P_\beta=S_\beta, \,\ \beta\in(0,1)\cup(1,9].\label{trans1}\eea
Obviously \bea \exp(J_2B_k(\beta)t)=R(\theta_kt)=\left(\begin{array}{cc}\cos (\theta_kt)&-\sin(\theta_kt)\\
                                \sin (\theta_kt)&\cos(\theta_kt)\end{array}\right),   \,\  k=1, 2, \nonumber\eea
and hence  \bea P_\beta^{-1}\gamma_{\beta,0}(t)P_\beta=R(\theta_1t)\diamond R(\theta_2t). \label{5.43}\eea
In order to get the diagonal matrix, we introduce a unitary matrix  $U=\frac{1}{\sqrt{2}}\left(
                                                                         \begin{array}{cc}
                                                                           I_{2} & \sqrt{-1} I_{2}\\
                                                                           I_{2} & -\sqrt{-1} I_{2} \\
                                                                         \end{array}
                                                                       \right)
$, then we have \bea U P_\beta^{-1}\gamma_{\beta,0}(t)P_\beta U^{-1}=e^{i \Theta t}.\eea
where $\Theta=diag\(\theta_{1},\theta_{2},\theta_{3},\theta_{4}\)$ , $\theta_{3}=-\theta_{1}, \theta_{4}=-\theta_{2}$. \\

Especially  \bea U P_\beta^{-1}\gamma_{\beta,0}(2\pi)P_\beta U^{-1}=e^{2\pi i \Theta}.\eea
Change the basis by $P_\beta U^{-1}$, then $f(\bb,\omega)$ could be computed by (\ref{ht2}),  we have
\bea f(\beta,\omega)&=&Tr((K_{\beta}^{-}(-J\frac{d}{dt}-\nu J-B_{\beta,0})^{-1})^{2})\nonumber\\
&=&-2Tr\big(J\int_{0}^{2\pi}\gamma_{\beta,0}^{T}(t)K_{\beta}^{-}(t)\gamma_{\beta,0}(t)
J\int_{0}^{t}\gamma_{\beta,0}^{T}(s)K_{\beta}^{-}(s)\gamma_{\beta,0}(s)ds dt
\cdot\gamma_{\beta,0}(2\pi)(\gamma_{\beta,0}(2\pi)-\omega I_{4})^{-1}\big) \nonumber\\
&&+Tr\big([J\int_{0}^{2\pi}\gamma_{\beta,0}^{T}(t)K_{\beta}^{-}(t)\gamma_{\beta,0}(t)dt
\cdot\gamma_{\beta,0}(2\pi)(\gamma_{\beta,0}(2\pi)-\omega I_{4})^{-1}]^{2}\big) \nonumber\\
&=&2Tr\big(J\int_{0}^{2\pi}\widetilde{\gamma}_{\beta,0}^{T}(t)\widetilde{D}^{-}_{\beta}(t)\widetilde{\gamma}_{\beta,0}(t)
J\int_{0}^{t}\widetilde{\gamma}_{\beta,0}^{T}(s)\widetilde{D}^{-}_{\beta}(s)\widetilde{\gamma}_{\beta,0}(s)ds dt
\cdot\widetilde{\gamma}_{\beta,0}(2\pi)(\widetilde{\gamma}_{\beta,0}(2\pi)-\omega I_{4})^{-1}\big) \nonumber\\
&&-Tr\big([J\int_{0}^{2\pi}\widetilde{\gamma}_{\beta,0}^{T}(t)\widetilde{D}^{-}_{\beta}(t)\gamma_{\beta,0}(t)dt
\cdot\widetilde{\gamma}_{\beta,0}(2\pi)(\widetilde{\gamma}_{\beta,0}(2\pi)-\omega I_{4})^{-1}]^{2}\big) \nonumber\\
&=&2\int_{0}^{2\pi}\int_{0}^{t}Tr\big(\widetilde{\gamma}_{\beta,0}(s)\widetilde{\gamma}_{\beta,0}^{T}(-t)\cdot J \widetilde{D}^{-}_{\beta}(t)
\cdot\widetilde{\gamma}_{\beta,0}(t)\widetilde{\gamma}_{\beta,0}^{T}(-s)\cdot J \widetilde{D}^{-}_{\beta}(s)ds dt
\cdot M_{\beta}(\omega)\big) \nonumber\\
&&-\int_{0}^{2\pi}\int_{0}^{2\pi}Tr\big(\widetilde{\gamma}_{\beta,0}(s)\widetilde{\gamma}_{\beta,0}^{T}(-t)\cdot J \widetilde{D}^{-}_{\beta}(t)\cdot M_{\beta}(\omega)
\cdot\widetilde{\gamma}_{\beta,0}(t)\widetilde{\gamma}_{\beta,0}^{T}(-s)J \widetilde{D}^{-}_{\beta}(s)ds dt
\cdot M_{\beta}(\omega)\big),
\label{5.44}\eea
where $\widetilde{D}^{-}_{\beta}(s)=U^{-T}P_\beta^T K_{\beta}^{-}(s) P_\beta U^{-1}$,
$\widetilde{\gamma}_{\beta,0}(t)=UP_{\beta}^{-1}\gamma_{\beta,0}(t)P_{\beta}U^{-1}=e^{i\Theta t}$, and
$M_{\beta}(\omega)=\widetilde{\gamma}_{\beta,0}(2\pi)(\widetilde{\gamma}_{\beta,0}(2\pi)-\omega I_{4})^{-1}$.
The last equation from the facts that $J \widetilde{\gamma}_{\beta,0}(t)=\widetilde{\gamma}_{\beta,0}(-t)J$ and
$\widetilde{\gamma}_{\beta,0}(s)$ commutes with $\widetilde{\gamma}_{\beta,0}(2\pi)$ .
In order to get the trace, we need to calculate $J\widetilde{D}_{\beta}^{-}(t)$ and
$M_{\beta}(\omega)$.
Let $\omega=e^{2\pi i u}$, $u\in\mathbb R$, direct calculation shows that
\bea
M_{\beta}(\omega)=
\left(
  \begin{array}{cccc}
    k_{1} & 0 & 0& 0 \\
    0 & k_{2} & 0 & 0 \\
    0 & 0 & k_{3} & 0 \\
    0 & 0  & 0 & k_{4}  \\
  \end{array}
\right), \ \ \ k_{j}=\frac{e^{2\pi i \theta_{j}}}{e^{2\pi i \theta_{j}}-e^{2\pi i u}}, \label{5.45}
\eea
and
$
P_{\beta}^{T}K_{\beta}P_{\beta}=
\left(
  \begin{array}{cccc}
    a & 0 & 0 & b \\
    0 & h & f & 0 \\
    0 & f & g & 0 \\
    b & 0 & 0 & c \\
  \end{array}
\right)
$ ,
where $a, b, c, h, f, g $ have explicit expression and depend on parameter $\beta$,
\bea\left\{
        \begin{array}{c}
        a=P_{21}(1,1)P_{21}(1,1)d_{1}\\
        b=P_{21}(1,1)P_{22}(1,2)d_{1}\\
        c=P_{22}(1,2)P_{22}(1,2)d_{1}\\
       \end{array}
\right.,\ \ \ \
\left\{
        \begin{array}{c}
        h=P_{21}(2,2)P_{21}(2,2)d_{2}\\
        f=P_{21}(2,2)P_{22}(2,1)d_{2}\\
        g=P_{22}(2,1)P_{22}(2,1)d_{2}
       \end{array}
\right.,\quad\text{and}\quad
\left\{
        \begin{array}{c}
        d_{1}=\frac{3+\sqrt{9-\beta}}{2}\\
        d_{2}=\frac{3-\sqrt{9-\beta}}{2}
       \end{array}.
\right.
\eea
Let $\widetilde{D}_{\beta}=U^{-T}P_\beta^T K_{\beta} P_\beta U^{-1}$,  by direct computation,  we have
\bea J\widetilde{D}_{\beta}=\frac{1}{2}\cdot
\left(
  \begin{array}{cccc}
     D_{11}&   D_{12} &  D_{13} & D_{14} \\
     D_{21} &  D_{22} &  D_{23} & D_{24} \\
     D_{31} &  D_{32} &  D_{33} & D_{34} \\
     D_{41} &  D_{42} &  D_{43} & D_{44} \\
  \end{array}
\right),  \label{formd}
\eea
where
\bea\label{eq5.46a}
D_{11}&=&-(a+g),  \ \ \ \
D_{22}=-(h+c), \ \ \ \
D_{33}=a+g,  \ \ \ \
D_{44}=h+c, \nonumber \\
D_{12}&=&-D_{21}=-i(f-b), \ \ \ \
D_{23}=D_{32}=-i(f+b), \ \ \ \
D_{24}=-D_{42}=c-h, \nonumber \\
D_{13}&=&-D_{31}=g-a,\ \ \ \
D_{14}=D_{41}=-i(f+b),\ \ \ \
D_{34}=-D_{43}=i(b-f). \ \ \ \
\eea
Obviously,
\bea J\widetilde{D}_{\beta}^{-}(s)=\cos^{-}(s)\cdot J\widetilde{D}_{\beta}.
\eea

In order to make the computation clearer, we introduce
\bea
f_{1}(\beta,\omega)=
\int_{0}^{2\pi}\int_{0}^{t}Tr\big[\widetilde{\gamma}_{\beta,0}(s)\widetilde{\gamma}_{\beta,0}^{T}(-t)J \widetilde{D}^{-}_{\beta}(t)
\cdot\widetilde{\gamma}_{\beta,0}(t)\widetilde{\gamma}_{\beta,0}^{T}(-s)J \widetilde{D}^{-}_{\beta}(s)\cdot M_{\beta}(\omega)\big]ds dt,
\eea
and
\bea
f_{2}(\beta,\omega)=
\int_{0}^{2\pi}\int_{0}^{2\pi}Tr\big[\widetilde{\gamma}_{\beta,0}(s)\widetilde{\gamma}_{\beta,0}^{T}(-t)J \widetilde{D}^{-}_{\beta}(t)
\cdot M_{\beta}(\omega)\cdot\widetilde{\gamma}_{\beta,0}(t)\widetilde{\gamma}_{\beta,0}^{T}(-s)J \widetilde{D}^{-}_{\beta}(s)\cdot M_{\beta}(\omega)\big]ds dt.
\eea
Therefore
\bea
f(\beta,\omega)=2f_{1}(\beta,\omega)-f_{2}(\beta,\omega).
\eea
Direct computation shows that
\bea f_{1}(\beta,\omega)=
\frac{1}{4}\sum_{\tiny \begin{array}{c}
 n=1  \\
 m=1
\end{array}}^{4}
D_{nm}D_{mn} k_{n}
\frac{2e^{\pi(\theta_{m}-\theta_{n})i}+\pi i(\theta_{m}-\theta_{n})[(\theta_{m}-\theta_{n})^{2}-1]}
{2[(\theta_{m}-\theta_{n})^{2}-1]^{2}},
\eea
and
\bea
f_{2}(\beta,\omega)=
\frac{1}{4}\sum_{\tiny \begin{array}{c}
 n=1  \\
 m=1
\end{array}}^{4}
D_{nm}D_{mn} k_{n}k_{m}\frac{2+e^{\pi(\theta_{m}-\theta_{n})i}+e^{-\pi(\theta_{m}-\theta_{n})i}}{[(\theta_{m}-\theta_{n})^{2}-1]^{2}},
\eea
where the blocks ${D_{nm}}$ are defined by (\ref{eq5.46a}). Thus  $f_{1}(\beta,\omega)$, $f_{2}(\beta,\omega)$ and $f(\bb,\omega)$ are elementary functions. Based on the precise form of the above functions, we can draw the  the curves $\Gamma_i$, $i=1,...,4$ in Figure 1 with the help of Matlab.

\subsection{Hyperbolicity analysis via the first order trace formula }\label{subsec5.4}

Recall that in (\ref{eq5.47a}), $\hat{f}(\beta)$ is defined by taking maximum, and maybe it is not an elementary function. Another way to estimate the hyperbolic region is to use the trace formula for Lagrangian system (\ref{lt.2}). It will be seen that the estimation of the hyperbolic region given by  the trace formula (\ref{ht2}) for Hamiltonian system is sharper than that given by the trace formula (\ref{lt.2}) for Lagrangian system. However, the later is more computable.

From (\ref{l1}),  for $\beta\in(1,9]$,  $\nu$ is imaginary number, $ \mathcal{A}(\beta,0,\nu)>0$. Recall that $\hat{K}_{\bb,0}(t)=\left(\begin{array}{cc}\frac{3+\sqrt{9-\beta}}{2} & 0 \\ 0 & \frac{3-\sqrt{9-\beta}}{2} \end{array}\right)$, for $\omega=e^{2\pi\nu}\in\mathbb{U}$, we define
\bea g(\beta,\nu)=-Tr\(JK_{\beta}
\cdot\gamma_{\beta,0}(2\pi)(\gamma_{\beta,0}(2\pi)-e^{2\pi\nu} I_{4})^{-1}\big)\). \label{g1} \eea
From (\ref{lt.2}) or (\ref{eq3.49a}), \bea Tr\(\frac{e\cos^+(t)}{1+e\cos(t)}\hat{K}_{\bb,0}\mathcal{A}(\beta,0,\nu)^{-1} \) &=& -Tr\(J\int_{0}^{2\pi}\gamma_{\beta,0}^{T}(t)\frac{e\cos^+(t)}{1+e\cos(t)}K_{\beta}\gamma_{\beta,0}(t)dt
\cdot\gamma_{\beta,0}(2\pi)(\gamma_{\beta,0}(2\pi)-\omega I_{4})^{-1}\big)\) \nonumber \\ &=&-\int_{0}^{2\pi}\frac{e\cos^+(t)}{1+e\cos(t)}dt\cdot Tr\(JK_{\beta}
\cdot\gamma_{\beta,0}(2\pi)(\gamma_{\beta,0}(2\pi)-\omega I_{4})^{-1}\big)\) \nonumber \\ &=& \(\pi-\frac{4}{\sqrt{1-e^2}}\tan^{-1}\sqrt{\frac{1-e}{1+e}}\)g(\beta,\nu), \eea
where the second equality is from the fact that $\gamma_{\beta,0}(t)=\exp(JB_{\beta,0}t)$ commutes with $\gamma_{\beta,0}(2\pi)$, and the third equality is
from \bea \int_{-\pi/2}^{\pi/2}\frac{e\cos(s)}{1+e\cos(s)}ds= \pi-\frac{4}{\sqrt{1-e^2}}\tan^{-1}\sqrt{\frac{1-e}{1+e}}. \nonumber \eea
Noting that $\pi-\frac{4}{\sqrt{1-e^2}}\tan^{-1}\sqrt{\frac{1-e}{1+e}}\geq0$ for $e\in[0,1)$, and  seting \bea \hat{g}(\beta)=\sup\{g(\beta,\nu), \nu\in\sqrt{-1}\mathbb{R} \}. \label{fbeta}\eea
In order to calculate $g(\beta,\nu)$, we change the basis by $P_\beta U^{-1}$, then
\bea  g(\bb,\nu)=Tr(iJ\widetilde{D}_{\beta}M_\beta(\omega)).  \eea
From (\ref{formd}), direct computation shows that
\begin{lem}\label{lem6.4}For $\beta\in (1,9]$ and $\nu=\sqrt{-1}u\in\sqrt{-1}\mathbb R$,  \label{lem6.2} \bea g(\beta,\nu)= 2Re\left(\frac{\sqrt{2}(-3-\beta+3\sqrt{1-\beta})}{4\sqrt{1-\beta}\sqrt{-1+\sqrt{1-\beta}}}
\frac{e^{-\sqrt{2}\pi\sqrt{-1+\sqrt{1-\beta}}}-e^{\sqrt{2}\pi\sqrt{-1+\sqrt{1-\beta}}}}{2\cos(2\pi u)-e^{-\sqrt{2}\pi\sqrt{-1+\sqrt{1-\beta}}}-e^{\sqrt{2}\pi\sqrt{-1+\sqrt{1-\beta}}}} \right).\eea
\end{lem}
Similar to the proof of Theorem \ref{th2.3}, we have the following theorem.
\begin{thm}\label{thm6.1} For $\beta\in(1,9]$, $\gamma_{\beta,e}$ is hyperbolic if
\bea \pi-\frac{4}{\sqrt{1-e^2}}\tan^{-1}\sqrt{\frac{1-e}{1+e}}<1/\hat{g}(\beta). \label{6.21}\eea

\end{thm}

We can use Theorem \ref{thm6.1} or Theorem \ref{th2.3} to estimate the hyperbolic region.
Next, we draw the following figure to compare the hyperbolic regions given by the two theorems  respectively.

\begin{figure}[H]
 \centering
   \includegraphics[height=0.36\textwidth,width=0.76\textwidth]{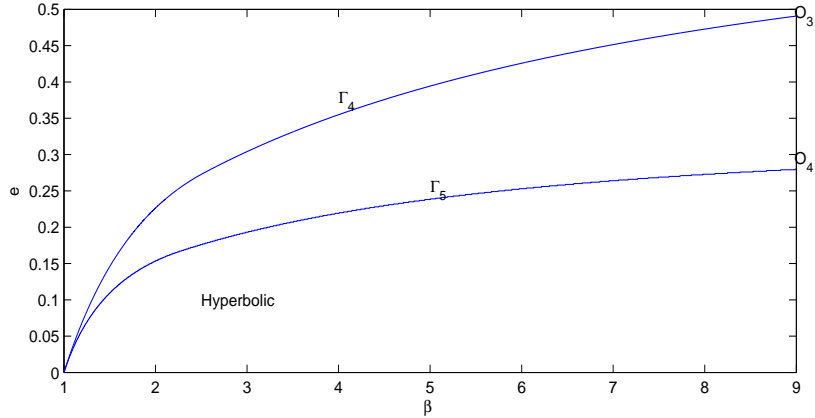}
     \caption{The hyperbolic region  given by Theorem \ref{thm6.1} and Theorem \ref{th2.3}.}
\end{figure}
\noindent

\noindent In Figure 2, the points $O_{3}\approx(9, 0.4907)$, $O_{4}\approx(9, 0.2800).$ The curves  $$\Gamma_{4}=\left\{(\beta,e)\,\left|\,
e=\hat{f}(\bb)^{-1/2}, 1\leq\beta\leq9\right.\right\},\ \ \Gamma_{5}=\left\{(\beta,e)\,\left|\,
\pi-\frac{4}{\sqrt{1-e^2}}\tan^{-1}\sqrt{\frac{1-e}{1+e}}=\hat{g}(\beta)^{-1}, 1\leq\beta\leq 9\right.\right\},$$
where $\Gamma_{4}$ is given by theorem \ref{th2.3}, which is obtained by (\ref{eq5b.5.1.1a}), and $\Gamma_{5}$ is given by theorem \ref{thm6.1}, which is obtained by using (\ref{lt.2}).\\

Since $\hat{g}(\beta)$ is not easy to be computed, we will control $\hat{g}(\beta)$ by some elementary  function.
\begin{lem}\label{lem6.1} For $\beta\in(1,9]$,
\bea \hat{g}(\beta)\leq\frac{\beta^{\frac{1}{4}}\sqrt{\beta+15}}{\sqrt{2(\beta-1)}}
\frac{\sqrt{e^{-2\sqrt{2}\pi \hat{c}}+e^{2\sqrt{2}\pi \hat{c}}-2\cos(2\sqrt{2}\pi \hat{d})}}
{|(e^{-\sqrt{2}\pi \hat{c}}-e^{\sqrt{2}\pi \hat{c}})\sin(\sqrt{2}\pi \hat{d})|},\eea
where $\hat{c}=Re(\sqrt{-1+\sqrt{1-\beta}})$, $\hat{d}=Im(\sqrt{-1+\sqrt{1-\beta}})$.

\end{lem}
\begin{proof} Let $\hat{c}=Re(\sqrt{-1+\sqrt{1-\beta}})$, $\hat{d}=Im(\sqrt{-1+\sqrt{1-\beta}})$,
direct computation shows that $\big|\sqrt{-1+\sqrt{1-\beta}}\big|=\beta^{\frac{1}{4}}$ for $\beta\in(1,9]$, applying Lemma \ref{lem6.4}, we have
\bea
g(\beta,\nu)&=&
2Re\(\frac{\sqrt{2}(-3-\beta+3\sqrt{1-\beta})}{4\sqrt{1-\beta}\sqrt{-1+\sqrt{1-\beta}}}
\frac{e^{-\sqrt{2}\pi\sqrt{-1+\sqrt{1-\beta}}}-e^{\sqrt{2}\pi\sqrt{-1+\sqrt{1-\beta}}}}{2\cos(2\pi u)-e^{-\sqrt{2}\pi\sqrt{-1+\sqrt{1-\beta}}}-e^{\sqrt{2}\pi\sqrt{-1+\sqrt{1-\beta}}}} \)\nonumber\\
&\leq&
\Big|\(\frac{\sqrt{2}(-3-\beta+3\sqrt{1-\beta})}{2\sqrt{1-\beta}\sqrt{-1+\sqrt{1-\beta}}}
\frac{e^{-\sqrt{2}\pi\sqrt{-1+\sqrt{1-\beta}}}-e^{\sqrt{2}\pi\sqrt{-1+\sqrt{1-\beta}}}}{2\cos(2\pi u)-e^{-\sqrt{2}\pi\sqrt{-1+\sqrt{1-\beta}}}-e^{\sqrt{2}\pi\sqrt{-1+\sqrt{1-\beta}}}} \)\Big|\nonumber \\
&=&
\frac{\sqrt{2}\sqrt{(3+\beta)^{2}+9(\beta-1)}}{2\sqrt{\beta-1}\beta^{\frac{1}{4}}}
\frac{|e^{-\sqrt{2}\pi(\hat{c}+\hat{d} i)}-e^{\sqrt{2}\pi(\hat{c}+\hat{d} i)}|}{|2\cos(2\pi u)-
e^{-\sqrt{2}\pi(\hat{c}+\hat{d} i)}-e^{\sqrt{2}\pi(\hat{c}+\hat{d} i)}|} \nonumber \\
&=&
\frac{\beta^{\frac{1}{4}}\sqrt{\beta+15}}{\sqrt{2}\sqrt{\beta-1}}
\frac{\sqrt{e^{-2\sqrt{2}\pi \hat{c}}+e^{2\sqrt{2}\pi \hat{c}}-2\cos(2\sqrt{2}\pi \hat{d})}}
{\sqrt{(2\cos(2\pi u)-(e^{-\sqrt{2}\pi \hat{c}}+e^{\sqrt{2}\pi \hat{c}})\cos(\sqrt{2}\pi \hat{d}))^{2}+
((e^{-\sqrt{2}\pi \hat{c}}-e^{\sqrt{2}\pi \hat{c}})\sin(\sqrt{2}\pi \hat{d}))^{2}}}\nonumber\\
&\leq&
\frac{\beta^{\frac{1}{4}}\sqrt{\beta+15}}{\sqrt{2(\beta-1)}}
\frac{\sqrt{e^{-2\sqrt{2}\pi \hat{c}}+e^{2\sqrt{2}\pi \hat{c}}-2\cos(2\sqrt{2}\pi \hat{d})}}
{|(e^{-\sqrt{2}\pi \hat{c}}-e^{\sqrt{2}\pi \hat{c}})\sin(\sqrt{2}\pi \hat{d})|}.\nonumber
\eea
Hence
\bea \hat{g}(\beta)\leq\frac{\beta^{\frac{1}{4}}\sqrt{\beta+15}}{\sqrt{2(\beta-1)}}
\frac{\sqrt{e^{-2\sqrt{2}\pi \hat{c}}+e^{2\sqrt{2}\pi \hat{c}}-2\cos(2\sqrt{2}\pi \hat{d})}}
{|(e^{-\sqrt{2}\pi \hat{c}}-e^{\sqrt{2}\pi \hat{c}})\sin(\sqrt{2}\pi \hat{d})|},\nonumber\eea
where $\hat{c}=Re\(\sqrt{-1+\sqrt{1-\beta}}\)$, $\hat{d}=Im\(\sqrt{-1+\sqrt{1-\beta}}\)$,
this completes the proof.
\end{proof}

\begin{cor}\label{cor6.1}
 For $\beta\in(1,9]$, $\gamma_{\beta,e}$ is hyperbolic if
\bea \pi-\frac{4}{\sqrt{1-e^2}}\tan^{-1}\sqrt{\frac{1-e}{1+e}}<
\frac{\sqrt{2(\beta-1)}}{\beta^{\frac{1}{4}}\sqrt{\beta+15}}
\frac{|(e^{-\sqrt{2}\pi \hat{c}}-e^{\sqrt{2}\pi \hat{c}})\sin(\sqrt{2}\pi \hat{d})|}
{\sqrt{e^{-2\sqrt{2}\pi \hat{c}}+e^{2\sqrt{2}\pi \hat{c}}-2\cos(2\sqrt{2}\pi \hat{d})}},\label{cor6.5f}
\eea
where $\hat{c}=Re(\sqrt{-1+\sqrt{1-\beta}})$, $\hat{d}=Im(\sqrt{-1+\sqrt{1-\beta}})$.
\end{cor}
Denote $h(\beta)$ be the right item of (\ref{cor6.5f}), and let $\beta_0$ be the point such that $$h(\beta_{0})=\max\{h(\beta):\beta\in(1,9]\}.$$ With the help of Mathlab, we know that $\beta_{0}\approx 3.0334$, correspondingly, $e\approx0.1797$. Hence $$h(\beta_0)\geq h(3.0334)=0.3154.$$
It was proved in \cite{HLS} that if  $\gamma_{\beta_0,e}$ is hyperbolic, then $\gamma_{\beta,e}$ is hyperbolic for any $\beta\geq\beta_0$, then we have
\begin{cor}\label{cor6.2}
 For $\beta\in[\beta_{0},9]$, $\gamma_{\beta,e}$ is hyperbolic if
\bea \pi-\frac{4}{\sqrt{1-e^2}}\tan^{-1}\sqrt{\frac{1-e}{1+e}}<
0.3154.\eea That is, $\gamma_{\beta,e}$ is hyperbolic if $(\beta,e)\in[3.0334,9]\times[0,0.1797]$.
\end{cor}
By using Corollary \ref{cor6.1}, \ref{cor6.2}, we can draw a picture of the hyperbolic region as follows.
\begin{figure}[H]
\begin{minipage}[t]{0.5\linewidth}
\centering
\includegraphics[height=0.64\textwidth,width=0.98\textwidth]{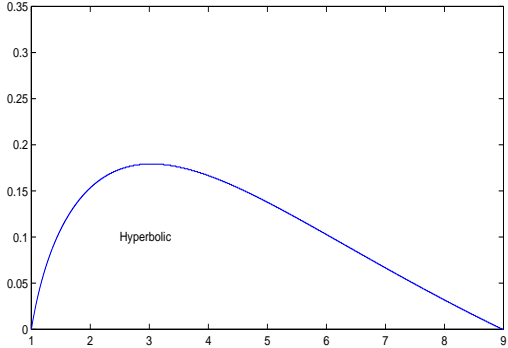}
     \caption{The  hyperbolic region given by Cor.\ref{cor6.1}}
\label{fig:side:a}
\end{minipage}%
\begin{minipage}[t]{0.5\linewidth}
\centering
\includegraphics[height=0.64\textwidth,width=1\textwidth]{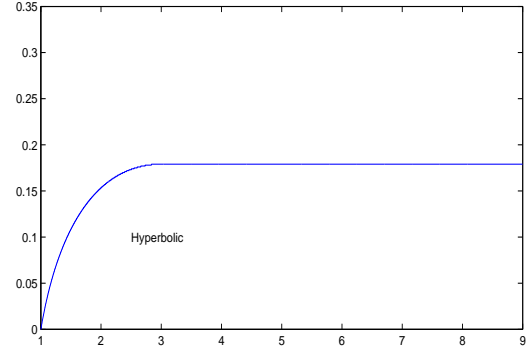}
     \caption{The  hyperbolic region given by Cor.\ref{cor6.2}}
\label{fig:side:b}
\end{minipage}
\end{figure}

\begin{rem} \label{5.5} From the proof of Theorems \ref{th2.1} and Theorem \ref{th2.2}, $\ga_{\bb,e}$ is $-1$-nondegenerate if $(\bb,e)$ belongs to the set $\{(\beta,e)| 0\leq e<1/(1+\sqrt{f(\beta,-1)}), 0\leq\beta<3/4\}$  or
$\{(\beta,e)| 0\leq e<1/\sqrt{f(\beta,-1)}, 3/4<\beta\leq9\}$.
However, using (\ref{lt.2}), we get that $\ga_{\bb,e}$ is $-1$-nondegenerate if $(\bb,e)$ belongs to the set
$\{(\beta,e)|\ \ \pi-\frac{4}{\sqrt{1-e^2}}\tan^{-1}\sqrt{\frac{1-e}{1+e}}<1/g(\beta,\frac{\sqrt{-1}}{2}), 3/4<\beta\leq9\}$.
\end{rem}

\begin{figure}[H]
 \centering
   \includegraphics[height=0.36\textwidth,width=0.76\textwidth]{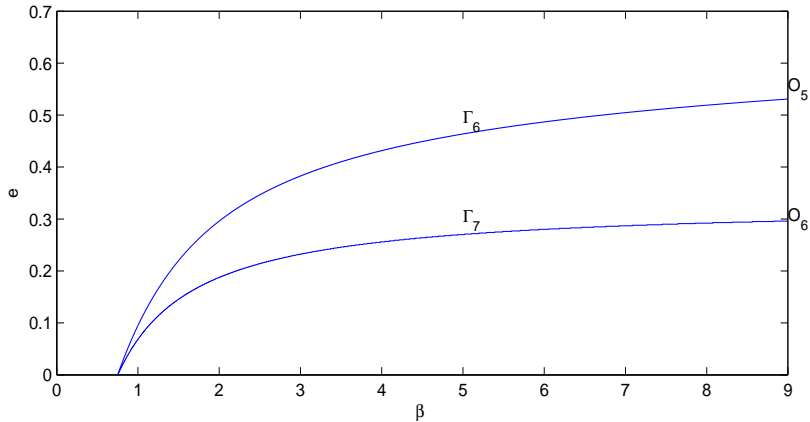}
     \caption{The -1-nondegenerate region given by  Remark 5.12.}
\end{figure}
\noindent
In Figure 5, the points $O_{5}\approx(9, 0.5309)$, $O_{6}\approx(9, 0.2961).$ The curves  $$\Gamma_{6}=\left\{(\beta,e)\,\left|\,
e=f(\beta,-1)^{-1/2}, 1\leq\beta\leq9\right.\right\},$$ and $$ \Gamma_{7}=\left\{(\beta,e)\,\left|\, \pi-\frac{4}{\sqrt{1-e^2}}\tan^{-1}\sqrt{\frac{1-e}{1+e}}=g\(\beta,\frac{\sqrt{-1}}{2}\)^{-1}, 1\leq\beta\leq 9\right.\right\}.$$

The same reasoning as above implies that we can estimate the non-degenerate region by the trace formulas in Theorem \ref{thm1.1} for  $k$. As $k$ is lager, the estimation of the non-degenerate region is sharper, however, the trace formula is more complex and less computable.

\medskip

\noindent {\bf Acknowledgements.} The  authors sincerely thank Y. Long and C. Zeng for their encouragements and interests. The first author sincerely thanks Y.Long and S.Sun's helpful discussion for the Lagrangian orbits.


\begin{thebibliography}{99}




\bibitem{APS} M. F. Atiyah, V. K. Patodi and I. M. Singer,
Spectral asymmetry and Riemannian geometry. III. Math. Proc. Cambridge Philos. Soc. 79 (1976), no. 1, 71-99.

\bibitem{B} S. V. Bolotin, On the Hill determinant of a periodic orbit. (Russian)
Vestnik Moskov. Univ. Ser. I Mat. Mekh. 1988, no. 3, 30-34, 114.

\bibitem{BT} S. V. Bolotin and D. V. Treschev,   Hill's formula.   Russian Math. Surveys, Vol.65(2010) no.2, 191-257.

\bibitem{Ch} K. Chen, Existence and minimizing properties of
retrograd orbits in three-body problem with various choice of mass.
 Ann. of Math.  167 (2008)  no.2, 325-348.


\bibitem{CM} A. Chenciner and R. Montgomery, A remarkable periodic solution of
the three body problem in the case of equal masses.  Ann. of Math.  152 (2000), no. 3, 881-901.



\bibitem{Dan}J. M. A. Danby, The stability of the triangular
Lagrangian point in the general problem of three bodies.
 Astron. J. 69. (1964), 294-296.

\bibitem{De} R. Denk, On Hilbert-Schmidt operators and determinants corresponding
to periodic ODE systems. Differential and integral operators
(Regensburg, 1995), 57-71, Oper. Theory Adv. Appl., Vol.102,
Birkh\"auser, Basel, 1998.

\bibitem{Ga} M. Gascheau,  Examen d'une classe d'\'{e}quations
diff\'{e}rentielles et application \`{a} un cas particulier du
probl\`{e}me des trois corps.  Comptes Rend. 16. (1843), 393-394.

\bibitem{FT} D. L. Ferrario and S. Terracini, On the existence of
collisionless equivariant minimizers for the classical $n$-body
problem.  Invent. Math.,  155 (2004)  no. 2, 305--362.

\bibitem{Hi} G. W. Hill, On the part of the motion of the lunar perigee which is a function of the mean motions of the sun and moon.  Acta Math. 8 (1886), no. 1, 1-36.

\bibitem{HLS}X. Hu, Y.Long, S. Sun, Linear stability of elliptic Lagrangian solutions of the classical planer
three-body problem via index theory. Preprint. arXiv:math/1206.6162 v1 Jun.27.(2012)

\bibitem{HS} X. Hu and S. Sun, Index and stability of symmetric periodic orbits in Hamiltonian systems
 with its application to figure-eight orbit. Commun. Math. Phys. 290 (2009), no. 2, 737-777.

\bibitem{HS1}X. Hu and S. Sun, Morse index and stability of elliptic Lagrangian
solutions in the planar three-body problem. Adv. Math. 223 (2010), no. 1, 98-119.


\bibitem{HW} X. Hu and P. Wang, Conditional Fredholm determinant of $S$-periodic orbits in Hamiltonian systems.  J. Funct. Anal. 261 (2011), no. 11, 3247-3278.
\bibitem{Lag} J. L. Lagrange,  Essai sur le probl\`{e}me des
trois corps. Chapitre II. {\OE}uvres Tome {6}, Gauthier-Villars,
Paris. (1772), 272-292.


\bibitem{Lon2} Y. Long,  Bott formula of the Maslov-type index
theory.  Pacific J. Math. 187 (1999), no. 1, 113-149.

\bibitem{Lon3} Y. Long,  Precise iteration formulae of the
Maslov-type index theory and ellipticity of closed characteristics.
  Adv. Math. 154 (2000), no. 1, 76-131.

\bibitem{Lon4} Y. Long, Index Theory for Symplectic Paths with
Applications, Progress in Math. Vol.207, Birkh\"auser. Basel. 2002.

\bibitem{Ka} T. Kato, Perturbation Theory for Linear Operators, second
edition. Springer-Verlag, Berlin Heidelberg New York Tokyo, 1984


\bibitem{K1} M.G. Krein, On tests for the stable boundedness of solutions of periodic canonical systems. PrikL Mat. Mekh. 19 (1955), no. 6,
641-680.

\bibitem{K2} M.G. Krein, Foundation of the theory of $\lambda$-zones
  of stability of a canonical systems of linear differential
  equations with periodic coefficients, In Memoriam: A.A.Andronov,
  Izdat.Akad.Nauk SSSR, Moscow,(1955) 413-498

\bibitem{MSS1} R. Mart\'{\i}nez, A. Sam\`{a}, C. Sim\'{o},
Stability diagram for 4D linear periodic systems with applications
to homographic solutions.   J. Differential Equations 226 (2006), no. 2, 619-651.

\bibitem{MSS2} R. Mart\'{\i}nez, A. Sam\`{a}, C. Sim\'{o}, Analysis of
the stability of a family of singular-limit linear periodic systems in
$\mathbb{R}^4$. Applications.   J. Differential Equations 226 (2006), no. 2, 652-686.


\bibitem{MS} K. R. Meyer, D. S. Schmidt, Elliptic relative equilibria in
the N-body problem.   J. Differential Equations 214 (2005), no. 2, 256-298.

\bibitem{ou} Y.Ou,   Hyperbolicity of elliptic Lagrangian orbits
in planer three body problem. To appear in Sci. China Math.

\bibitem{Po} A. Poincar\'{e},  Sur les d\'{e}terminants d'ordre infini. Bull. Soc. math. France, 14 (1886) 77-90.

\bibitem{RS} M. Reed and B. Simon, Methods of modern mathematical
phisics, IV. Analysis of operators. Academic Press [Harcourt Brace Jovanovich, Publishers], New York-London, 1978.

\bibitem{R1} G. E. Roberts, Linear stability of the elliptic Lagrangian
triangle solutions in the three-body problem.  J. Differential Equations 182 (2002), no. 1, 191-218.

\bibitem{R2}  E. J. Routh, On Laplace's three particles with
a supplement on the stability or their motion.  Proc. London Math.
Soc. 6. (1875) 86-97.



\bibitem{Si} B. Simon, Trace ideals and their applications. Second edition. Mathematical Surveys and Monographs, 120. American Mathematical Society, Providence, RI, 2005.






\end{thebibliography}
\end{document}